\documentclass[a4paper, 10pt]{article} 

\pdfoutput=1 
\usepackage{amsmath, amsfonts, amscd, amssymb,  amsthm, mathrsfs, ascmac}
\usepackage{latexsym}
\usepackage{longtable, geometry}
\usepackage[english]{babel}
\usepackage[utf8]{inputenc}
\usepackage{bbm}
\usepackage{enumitem}
\usepackage{nicefrac}
\usepackage{bm}
\usepackage[svgnames,psnames]{xcolor}
\usepackage[colorlinks,citecolor=DarkGreen,linkcolor=FireBrick,urlcolor=blue,linktocpage,unicode]{hyperref}
 \usepackage{braket}
\usepackage{xparse}
\usepackage{tikz-cd}
\usetikzlibrary{positioning}
\usepackage{authblk}
\usepackage{url, hyperref}

\newtheorem{Thm}{Theorem}[section]
\newtheorem{Def}[Thm]{Definition}
\newtheorem{Assum}[Thm]{Condition}
\newtheorem{Lemm}[Thm]{Lemma}
\newtheorem{Prop}[Thm]{Proposition}
\newtheorem{Coro}[Thm]{Corollary}
\newtheorem{Exa}[Thm]{Example}
\theoremstyle{definition}

\newtheorem{Rem}[Thm]{Remark}

\newcommand{\be}{\begin{equation}}
\newcommand{\ee}{\end{equation}}

\newcommand{\ba}{\begin{align}}
\newcommand{\ea}{\end{align}}

\newcommand{\ben}{\begin{equation*}}
\newcommand{\een}{\end{equation*}}

\def\i<#1>{\langle #1 \rangle}
\def\l<#1>{\left\langle #1 \right\rangle}
\def\b<#1>{\big\langle #1 \big\rangle}
\def\wrt{\ \text{w.r.t. }\ }
\newcommand{\la}{\langle}
\newcommand{\ra}{\rangle}
\newcommand{\bs}{\boldsymbol}

\newcommand{\D}{\mathrm{dom}}

\newcommand{\BbbR}{\mathbb{R}}
\newcommand{\BbbN}{\mathbb{N}}
\newcommand{\BbbZ}{\mathbb{Z}}

\newcommand{\BbbC}{\mathbb{C}}

\newcommand{\vepsilon}{\varepsilon}
\newcommand{\vphi}{\varphi}
\newcommand{\no}{\nonumber \\}

\newcommand{\ilim}[1][]{\mathop{\varinjlim}\limits_{#1}}

\newcommand{\im}{\mathrm{i}}

\newcommand{\up}{\uparrow}
\newcommand{\down}{\downarrow}

\newcommand{\A}{\mathfrak{A}}
\newcommand{\M}{\mathfrak{M}}
\newcommand{\N}{\mathfrak{N}}
\newcommand{\X}{\mathfrak{X}}
\newcommand{\F}{\mathfrak{F}}
\newcommand{\Cone}{\mathfrak{P}}
\newcommand{\h}{\mathfrak{H}}

\newcommand{\fra}{\mathfrak}
\newcommand{\nph}{\fra{n}_{\vphi}}

\newcommand{\BbbL}{\mathbb{L}}

\newcommand{\MLM}{H_{\Lambda}^{\rm MLM}}
\newcommand{\LRA}{\longrightarrow}

\numberwithin{equation}{section}

\makeatletter
  
  \@addtoreset{equation}{section}
\makeatother

\title{
\bfseries 

An algebraic approach to revealing magnetic structures of ground states in many-electron systems
}

\date{}

\author[1]{Tadahiro Miyao}

\affil[1]{Department of Mathematics,  Hokkaido University

Sapporo 060-0810,  Japan\\

e-mail: \texttt{miyao@math.sci.hokudai.ac.jp}
}

\begin{document}

\maketitle
\begin{abstract}
Mathematical understanding of the origin of ferromagnetism is still incomplete and remains an important research topic in mathematical physics.
In this paper, we give a model-independent mathematical framework describing the magnetic features of the ground states in many-electron systems. Within this framework, we also present a new approach to understanding magnetic orders in  macroscopic systems. 
Based on these, we construct a general theory that explains the stability of magnetic orders in the ground states despite the interaction of electrons with the environment.
Methodologically, the theory presented in this paper is formulated using  von Neumann algebras and their associated standard forms.
A benefit of working in such an algebraic setting is that we can define operator inequalities that preserve the ordered structures that naturally follow from the standard forms; by exploiting these operator inequalities, we can develop new descriptions of the magnetic structures of the ground states.
As specific applications of the proposed theory, we first analyze the Marshall--Lieb--Mattis theorem, Lieb's theorem, and their stabilities under various perturbations. Next, the Nagaoka--Thouless theorem and its stability are addressed. 
In addition, we interpret various other examples from the new theory and give a unified perspective on existing results.

\end{abstract}
\tableofcontents

\section{Introduction}\label{Sect1}
\subsection{Overview}
Magnets have a rich history. For example, there are descriptions of magnetite in ancient Chinese and Greek literature.
As is well known, magnets have the property of attracting objects such as iron.
This mysterious property of magnets must have fascinated the ancient people.
Magnets have also been useful in practical terms, as seen in the compasses of medieval Europe.\footnote{
For the history of the magnet, see \cite{Mattis2006,Yamamoto2017} and the references therein.
}
The magnetism of matter is also indispensable for modern science and technology; we are surrounded by magnetic devices such as hard disks.
\medskip

As we can see, magnets are very familiar to us, but do we know their origin and mechanism?
 Although there is a vast amount of knowledge about magnets, modern mathematical physicists believe their origin is still not well understood.
As background, we would like to explain a little about this.
In 1925, it was proposed that electrons have spin, which means that electrons are tiny magnets \cite{uhlenbeck1925ersetzung, UHLENBECK1926}.
Spin is a mysterious concept unique to quantum theory that has  no counterpart in classical theory.
There are many electrons (of the order of $10^{23}$) moving around in the matter around us.
The spin of a single electron is too weak to attract macroscopic iron.
The modern qualitative understanding of  magnets is that the spins of many electrons in  materials are aligned and strengthen each other, thereby attracting distant iron.
In other words, the wonders of the invisible quantum world are amplified to affect the macroscopic world. 
To what extent is this fascinating story rigorously proven?
\medskip

The earliest attempts to explain magnets, or more precisely ferromagnetism, using quantum theory date back to the work of Heisenberg \cite{Heisenberg1928}. Since then, many researchers have contributed to the field.
The modern approach to magnetism was pioneered by Kanamori \cite{Kanamori1963}, Gutzwiller \cite{Gutzwiller1963}, and Hubbard \cite{Hubbard1963}.
They attempted to explain ferromagnetism using a highly simplified model, now called the Hubbard model. The Hamiltonian of the Hubbard model on a lattice $\Lambda$ is given by
\begin{align}
H_{\Lambda}^{\mathrm{H}}=\sum_{x,y\in \Lambda}\sum_{\sigma=\uparrow, \downarrow} t_{xy}
c_{x\sigma}^*c_{y\sigma}+\sum_{x, y\in \Lambda}\frac{U_{xy}}{2} (n_x-1)(n_y-1),\label{Hami}
\end{align}
where 
$c_{x\sigma}$
is the electron annihilation operator at vertex  $x$ with spin $\sigma$, and $n_x$ is the number operator of  the electron at vertex $x$.
$H_{\Lambda}^{\mathrm{H}}$ is a self-adjoint operator acting in a fermionic Fock space.
For a more detailed definition of this operator, see Section \ref{Sect5}.
The model is simple but contains three fundamental factors in many electron systems.
The first is the electron hopping in the first term on the right-hand side of \eqref{Hami}; the second is the Coulomb interaction between the electrons in the second term on the right-hand side; and the third is the Fermi statistics of the electrons, which is expressed by the fact that the Hilbert space in which the Hamiltonian acts is a fermionic Fock space.
It is believed that ferromagnetism is manifested by the exquisite intertwining of these three fundamental factors.
Even today, research continues in an attempt to express this interplay mathematically.
It is important to note that ferromagnetism cannot be explained by the electron hopping or Coulomb interaction alone.
 Therefore, the problem of magnetism in the Hubbard model is essentially non-perturbative.
 This fact makes a rigorous study of the origin of magnetism challenging.
\medskip

The number of papers on the Hubbard model is enormous.
However, most of them are within the framework of theoretical physics and contain numerical calculations or mathematically unjustified approximations and assumptions.
Therefore, in order to truly understand the properties of magnets, we need to provide a rigorous proof of the above stories.
Below, we will review two previous studies that are particularly important to this research topic.
\medskip

The first rigorous result on ferromagnetism in the Hubbard model was given independently by Nagaoka \cite{Nagaoka1965} and Thouless \cite{Thouless_1965} in 1965.
They considered a system with a substantial Coulomb interaction and only one hole and proved that ferromagnetism appears in the ground states of this system.
 Nowadays, this result is called the  Nagaoka--Thouless theorem.
\medskip

In 1989, Lieb considered the Hubbard model on a connected, bipartite lattice and proved that macroscopic magnetization appears in the ground states when the system is half-filling \cite{Lieb1989}.
Today, this result is called Lieb's theorem.
In Lieb's theorem, there is no longer any restriction on the magnitude of the Coulomb interaction, whereas in the Nagaoka--Thouless theorem, a very large Coulomb interaction is assumed.
\medskip

These two theorems have significantly influenced subsequent rigorous researches on metallic ferromagnetism and will play  essential roles in this paper.
In addition to the two theorems above, there are several other critical rigorous results, e.g.,  
\cite{
Freericks1995,Katsura2013,Kollar1996,Kubo1990,Mielke1993,Shen1994,Shen1996,Tian2004,TSU1997,Tsunetsugu1997,Ueda1992,Wei2015,Yanagisawa1995,Yoshida2021}.
For a comprehensive review of rigorous research on ferromagnetism, see \cite{Tasaki2020} and references therein. 
Some of these results will be analyzed in detail in later sections. 
 In this way, even if we restrict ourselves to the rigorous results on metallic ferromagnetism, the readers can now see the variety of previous studies.
  At this point, the following natural questions should arise:
  \begin{description}
 \item[Q1.] Is there any connection between these seemingly unrelated results?
 If there is a connection, how can the structure be expressed mathematically?
\item[Q2.] What new consequences does the newly recognized structure imply?

\end{description}
 The goal of this paper is to answer these questions. 
 Answering  {\bf Q1} is nothing less than constructing a theory of magnetism that can describe,  in a unified way, the results that were previously thought to be independent.
 Whether or not this newly recognized theoretical framework can lead to new insights into the origin of magnetism is an essential question in assessing the novelty of this theory.

 \medskip
 
 Let us sketch our answer to {\bf Q1}.
Several of the rigorous results have a common structure. 
Using the language of operator algebras, we construct a mathematical framework that can adequately describe this structure. The two central themes that will be addressed in this framework are as follows:
\begin{description}
\item[1.] How can we mathematically define magnetic orders in a system on a finite lattice?
\item[2.] How can we characterize macroscopic magnetic orders?
\end{description}
In the following, we will outline our approach to these themes.
A many-electron system on a finite lattice $\Lambda$ is described by
a von Neumann algebra $\M_{\Lambda}$ and a weight $\vphi_{\Lambda}$ on $\M_{\Lambda}$. 
From the pair $\{\M_{\Lambda}, \vphi_{\Lambda}\}$, a standard form 
$
F(\M_{\Lambda}, \vphi_{\Lambda})=\{\M_{\Lambda}, \h_{\Lambda}, \Cone_{\Lambda}, J_{\Lambda}\}
$
 can be defined naturally, where $\h_{\Lambda}$ is the Hilbert space on which $\M_{\Lambda}$ acts, $\Cone_{\Lambda}$ is the self-dual cone in $\h_{\Lambda}$, and $J_{\Lambda}$ is the modular conjugation. 
 The detailed construction of these will be given in Section \ref{Prel}.
In the theory we construct in the present paper, 
the following two elements describe the magnetic order in  a system on a finite lattice $\Lambda$:
\begin{description}
 \item[\it Positivity:]
 We say that an  element in $\h_{\Lambda}$
 is called positive w.r.t.  $\Cone_{\Lambda}$ if it belongs to $\Cone_{\Lambda}$. We refer this  positivity determined by  $\Cone_{\Lambda}$  simply as the positivity on $\Lambda$. 
\item[\it Symmetry:]
Since we are considering a many-electron system, it is natural to assume that the total spin operators, ${\bs S}_{\Lambda}=(S^{(1)}_{\Lambda}, S^{(2)}_{\Lambda}, S^{(3)}_{\Lambda})$, act on $\h_{\Lambda}$.
 We refer the conservation law concerning the total spin operators simply as the  symmetry on $\Lambda$.
\end{description}
In the subsequent sections, we will explore in detail how the combination of positivity and symmetry determines the magnetic properties of the ground state of the system on $\Lambda$.
\medskip

The characteristics of macroscopic systems are described by considering the infinite volume limit or the thermodynamic limit.
This is an idea that has been widely accepted. From this point of view,
it is natural to think that  macroscopic magnetic orders are derived from the magnetic orders of the systems on finite lattices. How can we express this concept mathematically? Let us state the problem in more detail.
Given an increasing sequence of finite lattices $\{\Lambda_n : n\in \BbbN\}$, let us consider the corresponding many-electron systems described by  $\{\M_{\Lambda_n}, \vphi_{\Lambda_n}\}\,  (n\in \BbbN)$, where,
for each $n\in\BbbN$, $\M_{\Lambda_n}$ is a von Neumann algebra and $\vphi_{\Lambda_n}$ is a weight on $\M_{\Lambda_n}$. The question we asked  can then be stated in more detail as follows.
How can we reconcile the rational perspective of the emergence of macroscopic magnetic order at the thermodynamic  limit with the positivity and the symmetry described earlier? Our answer to this question can be summarized as follows: macroscopic magnetic orders appear due to the consistent association of the magnetic orders in different finite lattices. Mathematically, this idea of consistency can be expressed using the conditional expectations $\mathscr{E}_{mn} : \M_{\Lambda_n} \to \M_{\Lambda_m}\, (m<n)$ defined as follows.
When $m<n$, $\h_{\Lambda_m}$ can be regarded as a closed subspace of $\h_{\Lambda_n}$ to consider the orthogonal  projection operator $P_{mn}$ from $\h_{\Lambda_n}$ to $\h_{\Lambda_m}$. In this case, $\mathscr{E}_{mn}$ is given by $\mathscr{E}_{mn}(x)=P_{mn} xP_{mn}\, (x\in \M_{\Lambda_n})$. The conditional expectations are required to satisfy the following equation:
\be
P_{mn} \Cone_{\Lambda_n}=\Cone_{\Lambda_n}\ (m<n). \label{PpropP}
\ee
Recall that the magnetic structures of systems in different lattices are characterized by the positivities and symmetries in each system. Thus, the conditional expectation $\mathscr{E}_{mn}$ links the positivity and symmetry in lattice  $\Lambda_m$ with those in lattice $\Lambda_n$ in a consistent manner.
Therefore, the macroscopic magnetic order can be understood as a network formed by the magnetic orders (i.e., the positivities and symmetries) of different finite lattices connected through  the conditional expectations. This picture will be represented using graph theory in Section \ref{Sect4}.
\medskip

Our response to {\bf Q2} is as follows.
The framework developed to answer {\bf Q1} has the great advantage of describing the magnetic properties in a model-independent way.
By utilizing this advantage, we present a new theory of the stability of magnetic orders. Let us explain why this stability theory is so critical.
Electrons in actual magnetic materials are constantly perturbed by the surrounding environment, such as lattice vibrations and light (quantized radiation fields).
On the other hand, despite the perturbations, the magnetic materials around us maintain their effectiveness as magnets. 
For example, a magnet stuck to the blackboard in a lecture room shows no sign of falling off during the lecture.
If the Nagaoka--Thouless theorem and Lieb's theorem describe the nature of magnetism, then they must be stable against perturbations caused by the environment. 
In this paper, we construct a coherent theory of magnetic stability to justify this consideration, using the framework. 
It is worth emphasizing that a structure similar to Eq. \eqref{PpropP} plays an essential role in constructing this stability theory.
\medskip

Next, let us discuss the features of the methodology. As we have already seen, the first feature is that we make full use of the {\it standard forms} of von Neumann algebras. The theory of standard forms was originated by Haagerup \cite{Haagerup1975} and plays a fundamental role in the current theory of operator algebras. The self-dual cone appearing in the definition of the  standard form corresponds to the positivity we mentioned earlier, which describes the magnetic properties of the system.
The second feature is the practical application of operator inequalities that are naturally defined from the self-dual cones.
 Those inequalities are entirely different from the usual operator inequalities in textbooks on functional analysis and operator theory.\footnote{
 Let $\X$ be a given Hilbert space, and let $A$ and $B$ be bounded linear operators on $\X$.
 The ordinary operator inequality is defined as follows: if $\la \eta|A\eta\ra \ge  0\, (\eta\in \h)$, then denote $A\ge 0$. On the other hand, the definition of the OPOI is as follows: Let $\Cone$ be a self-dual cone in $\X$. If $A$ satisfies $\la \eta|A\xi\ra \ge  0\, (\eta, \xi\in \Cone)$, then we denote $X \unrhd  0$. 
 One of the valuable features of the OPOIs is the following property: if $A\unrhd 0$ and $B\unrhd 0$, then $AB\unrhd 0$.
 In ordinary operator inequalities, such a property does not hold in general.
 }
  It should be noted that Miura did the pioneering work on these inequalities in \cite{Miura2003}. In this paper, we will call those  inequalities {\it order-preserving operator inequalities} (OPOIs) to distinguish them from ordinary operator inequalities.
The applications of OPOIs to mathematical physics have been systematically studied in \cite{Miyao2016(2),Miyao2019-2,Miyao2021}. In particular, the usefulness of OPOIs in many-electron systems has been discussed in detail in \cite{Miyao2012,Miyao2016,Miyao2017,Miyao2021-2}. In this paper, we synthesize some of the analytical methods in these papers, i.e., we give an OPOIs description independent of a particular model.
\medskip

Finally, we briefly explain the difference between Ref. \cite{Miyao2019} and this paper.
In  \cite{Miyao2019}, a mathematical characterization of  magnetic orders is also studied without relying on a specific model. 
Ref. \cite{Miyao2019} can be seen as a sketch of this paper, as the basic idea is the same as this paper but written without complicated mathematics.
As mentioned earlier, the description of this paper assumes knowledge of basic operator algebras. This assumption allows for a deeper mathematical description of the magnetic structures, which was not possible in \cite{Miyao2019}. The main new results are as follows:
\begin{itemize}
\item  We give a more in-depth characterization of the magnetic structures of the ground states for finite systems.
\item  The magnetic orders in macroscopic systems are described in more detail.
Note that not much discussion of magnetic orders in macroscopic systems has been done in  \cite{Miyao2019}.
\item  In various examples, one often considers a many-electron system on a periodic realization of an infinite lattice. In this paper, we extract the magnetic properties independent of the lattice realization.
\end{itemize}

\subsection{Outline}
This paper is organized as follows.
In Section \ref{Prel}, we  summarize the knowledge of operator algebras that will be needed throughout this paper. First, we will quickly explain the standard forms associated with von Neumann algebras. Next, we recall the definition of the conditional expectation between von Neumann algebras, which will play an essential role in the following sections. Finally, we briefly review the order structures on Hilbert spaces that naturally arise from the standard forms and the basic properties of the operators that preserve these order structures.  The order-preserving operator inequalities introduced here will be essential in the descriptions in the following sections.
\medskip

In Section \ref{Sect3}, we construct a model-independent framework for describing  magnetic properties of  ground states in a many-electron system on a finite lattice. We use this framework to argue that the magnetic properties are stable even when there are interactions between the electrons and the surrounding environment. The concept of stability classes introduced in this section is key to understanding the stability of magnetic properties in ground states.\medskip

In Section \ref{Sect4}, we analyze in detail the magnetic orders and their stability in macroscopic systems.
First, we give a mathematical definition of the  magnetic order in a macroscopic system. Its fundamental idea is as follows. Consider a monotonically increasing sequence of finite lattices. On each lattice, a many-electron system is given. The magnetic properties of the many-electron system on each finite lattice are described in the theory of Section \ref{Sect3}.
The macroscopic magnetic order is defined to be  a consistent relation between the magnetic properties of each many-electron system, which can be mathematically expressed as a sequence of conditional expectations. It should be noted that such a view of magnetic order has not been used before. This section further discusses the stability of the magnetic order under  interactions with the environment. We also show that the theory constructed in this section is independent of  realizations of the crystal lattice.\medskip

In Sections \ref{Sect5} and \ref{Sect6}, two crucial concrete examples are discussed in detail. First, in Section \ref{Sect5}, the Marshall--Lieb--Mattis  stability class is introduced, and its properties are clarified; and the readers will see that the Marshall--Lieb--Mattis theorem, Lieb's theorem, and its various extensions can be described in a unified manner.
As typical examples, the Holstein--Hubbard model describing the electron-phonon interaction, the Kondo lattice model, and a model describing the interaction of the Kondo lattice system with phonons will 
be examined in detail in terms of the Marshall--Lieb--Mattis stability class.
In Section \ref{Sect6}, we introduce the Nagaoka--Thouless stability class and study its properties in some depth. Then, using this stability class, we argue the Nagaoka--Thouless theorem and its various extensions in a unifying manner. As an example, the Holstein--Hubbard model describing the electron-phonon interaction will be addressed in detail using the Nagaoka--Thouless stability class.
\medskip

In Section \ref{Sect7}, we outline several stability classes that cannot be covered in Sections \ref{Sect5} and \ref{Sect6} to clarify that the theory presented in this paper contains a wealth of examples.
In addition, some open problems are discussed here.\medskip

In Appendix \ref{SectA}, we develop a general theory of the ergodicity of the semigroups generated by Hamiltonians. The general theory given here covers all the concrete examples treated in this paper. Then, in Appendix \ref{SectB}, we prove the ergodic properties of the semigroups generated by the eight Hamiltonians discussed in Sections \ref{Sect5} and \ref{Sect6}. Finally, in Appendix \ref{SectC}, we present some practical results on order-preserving operator inequalities.

\subsubsection*{Acknowledgements}
This work was supported by JSPS KAKENHI Grant Numbers 18K03315, 20KK0304.
I would like to thank Kazuhiro  Nishimata for his assistance in drawing some of the figures and  Keiko Miyao  for the helpful conversations and encouragement.

\section{Preliminaries}\label{Prel}
\subsection{Standard form of a von Neumann algebra}\label{PrelSta}
In this section, we will prepare basic terminologies related to the  theory of operator algebras necessary to construct a mathematical theory of stability of magnetism  that is the subject of this paper.
 First, we recall some basic definitions  and facts about von Neumann algebras and their associated standard forms. Next, we briefly  review some  fundamentals  of conditional expectations, which will play an essential role in this paper. Finally, we describe positivities on Hilbert spaces naturally induced by standard forms and basic properties of linear operators that preserve the positivities. Note that the definitions and  properties described in Subsections \ref{PrelSta} and \ref{PrelCond}  are  discussed in detail in \cite{Takesaki2003}.

Let $\M$ be a von Neumann algebra on a complex separable Hilbert space $\h$.
Let $\vphi$ be a faithful semi-finite normal weight on  $\M$.\footnote{For the definition of weights on von Neumann algebras, see \cite[Chapter VII, Definition 1.1]{Takesaki2003}.} 
We set $\frak{n}_{\vphi}=\{x \in \M : \vphi(x^*x)<\infty\}$.
Equipped with the sesquilinear functional $\la x|y\ra=\vphi(x^*y)\, (x, y\in \frak{n}_{\vphi})$, $\frak{n}_{\vphi}$ is a pre-Hilbert space.
We denote by $L^2(\M, \vphi)$ the  completion of $\frak{n}_{\vphi}$.

The action of $\M$ onto $L^2(\M, \vphi)$ is given by 
\be
\pi_{\vphi}(a) x=ax,\ \ a\in \M,\ x\in \frak{n}_{\vphi}.
\ee
Note that,  by using the  inequality $(ax)^*(ax) \le \|a\|^2 x^*x$, we can extend the action of $\M$
on $\frak{n}_{\vphi}$ to that on $L^2(\M, \vphi)$.
It is well-known that the representation $\{\pi_{\vphi}, L^2(\M, \vphi)\}$ is a non-degenerate faithful normal $*$-representation of $\M$. 
Since $\M$ is $\sigma$-finite, the von Neumann algebras $\{\pi_{\vphi}(\M), L^2(\M, \vphi)\}$  and $\{\M, \h\}$ are  spatially isomorphic.\footnote{
The brief outline of the proof is as follows:
Because $\M$ is $\sigma$-finite, there exists a cyclic and separating vector $\xi$
 for $\M$, see, e.g., \cite[Proposition 2.5.6]{Bratteli1987}. Now set $\vphi_{\xi}(x)=\la \xi|x\xi\ra\, (x\in \M)$. We readily confirm that $\vphi_{\xi}$ is  
 a faithful semi-finite normal weight on  $\M$. By using arguments similar to those in \cite[Chapter IX, Section 1]{Takesaki2003}, the two standard forms $\{\pi_{\vphi}(\M), L^2(\M, \vphi), L^2(\M, \vphi)_+, J_{\vphi}\}$
 and $\{\pi_{\vphi_{\xi}}(\M), L^2(\M, \vphi_{\xi}), L^2(\M, \vphi_{\xi})_+, J_{\vphi_{\xi}}\}$
 are unitarly equivalent to each other.  Because $\pi_{\vphi_{\xi}}(\M)=\M$ and $L^2(\M, \vphi_{\xi})=\h$, we conclude the desired claim.
}  Therefore, we often identify  $\pi_{\vphi}(\M)$ with $\M$ and 
$L^2(\M, \vphi)$ with $\h$.

Let 
\be
\A_{\vphi}=\nph\cap \nph^*\subseteq L^2(\M, \vphi).
\ee
Then, due to \cite[Chapter VII, Theorem 2.6]{Takesaki2003},  $\A_{\vphi}$ is a full left Hilbert algebra with the involution $x^{\sharp}:=x^*\, (x\in \A_{\vphi})$ such that 
 $\pi_{\vphi}(\M)=\mathcal{R}_{\ell}(\A_{\vphi})$, where
 \be
 \mathcal{R}_{\ell}(\A_{\vphi})=\pi_{\vphi}(\A_{\vphi})^{\prime\prime}.
 \ee
Here, for a given set, $S$, of bounded operators, $S^{\prime\prime}$ stands for the double commutant of $S$.

Consider the involution: $x \in \A_{\vphi}\mapsto x^*\in \A_{\vphi}$ and denote it by $S_0$. Then $S_0$ is closable. We denote by $S$ the closure of $S_0$. The modular operator, $\Delta_{\vphi}$,  and the modular conjugation, $J_{\vphi}$,  of $\A_{\vphi}$ is defined by the polar decomposition of $S$:
$
S=J_{\vphi}\Delta_{\vphi}^{1/2}.
$
It holds that 
\begin{align}
\Delta_{\vphi}^{\im t} \pi_{\vphi}(\M) \Delta_{\vphi}^{-\im t}=\pi_{\vphi}(\M)\ \ \forall t\in \BbbR, \ \ \ 
J_{\vphi}\pi_{\vphi}(\M)J_{\vphi}=\pi_{\vphi}(\M)^{\prime},
\end{align}
where, for a given set, $S$, of bounded operators, $S^{\prime} $ denotes the commutant of $S$.

We set 
\be
\A_{\vphi, 0}=\bigg\{
\xi\in \bigcap_{n\in \BbbZ}\D(\Delta_{\vphi}^n) : \Delta_{\vphi}^n \xi\in \A_{\vphi}, \ n\in \BbbZ
\bigg\}.
\ee
Then, by \cite[Theorem 2.2]{Takesaki2003}, $\A_{\vphi, 0}$ is a Tomita algebra associated with $\{\Delta_{\vphi}^{\im \alpha} : \alpha\in \BbbC\}$ such that 
\be
J_{\vphi}\A_{\vphi, 0}=\A_{\vphi, 0},\ \ \mathcal{R}_{\ell}(\A_{\vphi, 0})=\mathcal{R}_{\ell}(\A_{\vphi}).
\ee
Define
\be
L^2(\M, \vphi)_+=\overline{\{\xi J_{\vphi}\xi : \xi \in \A_{\vphi, 0}\}},
\ee
where the bar denotes the closure in the strong topology.
As is well-known, $L^2(\M, \vphi)_+$ is a self-dual cone, see \cite[Chapter IX, Theorem 1.2]{Takesaki2003} for detail. The quadruple 
\be
F(\M, \vphi)=\{\M, L^2(\M, \vphi), L^2(\M, \vphi)_+, J_{\vphi}\}
\ee
 is called a {\it standard form} of $\M$.

\subsection{Conditional expectations}\label{PrelCond}

In what follows, we denote by $\frak{W}_0(\M)$ the set of all faithful semi-finite normal weights on $\M$.
Take $\vphi\in \frak{W}_0(\M)$, arbitrarily.
Let $\N$ be a von Neumann subalgebra of $\M$ such that the restriction $\vphi\restriction \N$  of $\vphi$ to $\N$ is semi-finite.
Let $\mathscr{E}$ be the conditional expectation of $\M$ onto $\N$ with respect to $\vphi$, that is, 
$\mathscr{E}$ is a linear mapping  from $\M$ to $\N$ satisfying the following conditions:
\begin{itemize}
\item $\|\mathscr{E}(x)\| \le \|x\|,\ \ x\in \M$;
\item $\mathscr{E}(x)=x,\ \ x\in \N$;
\item $\vphi=\vphi \circ \mathscr{E}$.
\end{itemize}
By Takesaki\rq{}s theorem \cite[Chapter IX, Theorem 1.2]{Takesaki}, it holds that $\Delta_{\vphi}^{\im t}\N\Delta_{\vphi}^{-\im t}=\N$ for all $t\in \BbbR$.
For simplicity of notation, we set $\psi=\vphi\restriction \N$.
Let $P_{\vphi}$ be the orthogonal projection from $L^2(\M, \vphi)$ to $L^2(\N, \psi)$.
Then we have
$
P_{\vphi}x=\mathscr{E}(x)\ (x\in \frak{n}_{\vphi}).
$
In addition, $P_{\vphi}$ commutes with $\Delta_{\vphi}$ and $J_{\vphi}$.

\begin{Def}\label{DefCons1}\upshape
Suppose we are in the situation described above.
We say that the standard forms $ F (\M, \vphi)$ and $ F (\N, \psi)$ are {\it consistent}, if 
it holds that  
$L^2(\N, \psi)_+\subseteq L^2(\M, \vphi)_+$ and 
$
P_{\vphi}L^2(\M, \vphi)_+=L^2(\N, \psi)_+.
$
We express this as $ F (\M, \vphi) \xrightarrow[\ \mathscr{E}\ ]{}  F (\N, \psi)$. In the following,  the subscript $\mathscr{E}$ is often omitted if no confusion occurs.

\end{Def}

As we will see,  the concept of the consistency between standard forms plays an important role in the present paper.

\subsection{Order preserving operator inequalities}
In this section, we will briefly explain a theory of order-preserving operator inequalities (OPOIs) developed in \cite{Miyao2016(2),Miyao2019,Miyao2021}.
The inequalities treated in this subsection are different from the usual operator inequalities and allow for the visual representation of the abstract concepts discussed in this paper. The reader will see in the following sections that the practical application of the properties of these inequalities will allow for a more efficient analysis of complicated Hamiltonians. Technical topics that are useful, but may detract from the flow of the paper, are discussed in detail in Appendix \ref{SectC}.

Suppose we are given a von Neumann algebra $\M$ and a faithful semi-finite normal weight on $\M$.

\begin{Def} \label{DefHilC}\upshape
A vector $ u \in L^2(\M, \vphi)_+$ is said to be {\it  positive w.r.t.} $ L^2(\M, \vphi)_+$. We write this as $u \geq 0$ w.r.t. $ L^2(\M, \vphi)_+$. A vector $v \in L^2(\M, \vphi)$ is called {\it strictly positive w.r.t.} $
L^2(\M, \vphi)_+
$,  whenever $\i<v| u>>0$ for all $ u \in  L^2(\M, \vphi)_+\setminus\{0\}$. We write this as $v>0$ w.r.t. $ L^2(\M, \vphi)_+$.
\end{Def}

The following operator inequalities  will play a fundamental  role in the present paper.
\begin{Def}\upshape 
Let $A$ be a bounded operator on $L^2(\M, \vphi)$. 
\begin{itemize}
\item $A$ is {\it  positivity preserving} if $A L^2(\M, \vphi)_+\subseteq  L^2(\M, \vphi)_+$. We write this as $A\unrhd 0\wrt  L^2(\M, \vphi)_+$.\footnote{
The inequality symbol for positivity preservation is borrowed from \cite{Miura2003}.
}
\item $A$ is {\it positivity improving} if,  for all $u \in L^2(\M, \vphi)_+ \setminus \{0\}$, it holds that  $A u >0\wrt  L^2(\M, \vphi)_+$.  We write this as $A\rhd 0\wrt  L^2(\M, \vphi)_+$.
\end{itemize}

\end{Def}

The following proposition will often be helpful.

\begin{Prop}\label{AJAJP}
Let $J_{\vphi}$ be the modular conjugation associated with $\{\M, \vphi\}$.
For  each $a\in \M$, it holds that $aJ_{\vphi}aJ_{\vphi} \unrhd 0$ w.r.t. $L^2(\M, \vphi)_+$.
\end{Prop}
\begin{proof}
See \cite[Chapter IX, Theorem 1.2 (iii)]{Takesaki2003}.
\end{proof}

\begin{Def}
\upshape
Let $A$ be a self-adjoint operator on $L^2(\M, \vphi)$,   bounded from below.
The semigroup $\{e^{-tA}\}_{t\ge 0}$ is said to be {\it ergodic} w.r.t. $ L^2(\M, \vphi)_+$, if the following (i) and (ii) are satisfied:
\begin{itemize}
\item[(i)] $e^{-tA} \unrhd 0$ w.r.t. $ L^2(\M, \vphi)_+$ for all $t\ge 0$;
\item[(ii)] for each $u, v\in  L^2(\M, \vphi)_+ \setminus \{0\}$, there is a $t\ge 0$
such that $\langle u| e^{-tA} v\rangle >0$. Note that $t$ could depend on $u$ and $v$.
\end{itemize}

\end{Def}

The following lemma   immediately follows  from the definitions:
\begin{Lemm}
Let $A$ be a self-adjoint operator on $L^2(\M, \vphi)$, bounded from below. If $e^{-tA}\rhd 0$ w.r.t. $ L^2(\M, \vphi)_+$ for all $t>0$, then $\{e^{-tA}\}_{t\ge 0}$ is ergodic w.r.t. $ L^2(\M, \vphi)_+$.
\end{Lemm}

The basic result here is:
\begin{Thm}[Perron--Frobenius--Faris]\label{pff}
Let $A$ be a self-adjoint operator, bounded from below. Assume that $E(A)=\inf \mathrm{spec}(A)$ is an eigenvalue of $A$, where $\mathrm{spec}(A)$ indicates the spectrum of $A$. Let $\mathcal{V}$ be the  eigenspace corresponding to $E(A)$. If $\{e^{-tA}\}_{t\ge 0}$ is ergodic w.r.t. $ L^2(\M, \vphi)_+$, then $\dim \mathcal{V}=1$ and $\mathcal{V}$ is spanned by a   strictly positive vector w.r.t. $ L^2(\M, \vphi)_+$.
\end{Thm}

\begin{proof}
See \cite{Faris1972}.
\end{proof}

The following simple proposition is often useful.
\begin{Prop} \label{NonVOver}Let $A$ be a bounded operator on $L^2(\M, \vphi)$ with  $A\neq 0$. Assume that $u > 0$ w.r.t. $L^2(\M, \vphi)_+$. If  $A  \unrhd 0$ w.r.t. $L^2(\M, \vphi)_+$,
then $Au \neq  0$.
\end{Prop}
\begin{proof}
See, \cite[Proposition 2.25]{Miyao2019}.
\end{proof}

In the remainder of this subsection, we will introduce the basic techniques for manipulating OPOIs. 
Given a Hilbert space $\h$, we denote by $\mathscr{L}(\h)$ the set of all bounded operators on $\h$.
We readily confirm the following lemma:
\begin{Lemm}
Let $A,B\in\mathscr L(L^2(\M, \vphi))$.
 Suppose that  $A,B\unrhd0\wrt   L^2(\M, \vphi)_+$. We have the following:
 \begin{itemize}
 \item[{\rm (i)}] If $a, b\ge 0$, then $aA+bB\unrhd0\wrt   L^2(\M, \vphi)_+$;
 \item[{\rm  (ii)}] $AB\unrhd0\wrt   L^2(\M, \vphi)_+$.
 \end{itemize}
\end{Lemm}
\begin{proof}
For proof, see, e.g., \cite{Miura2003, Miyao2016}. 
\end{proof}

Let $J_{\vphi}$ be the modular conjugation associated with $\{\M, \vphi\}$.
We say that a vector $\xi\in L^2(\M, \vphi)$ is $J_{\vphi}$-{\it real}, if it satisfies $J_{\vphi}\xi=\xi$.
We denote by $L^2(\M, \vphi)_J$ be the set of all $J_{\vphi}$-real vectors.
 If $A\in \mathscr{L}(  L^2(\M, \vphi))$ satisfies $A  L^2(\M, \vphi)_{J} \subseteq   L^2(\M, \vphi)_{J}$, then we say that $A$ {\it preserves the reality w.r.t. $  L^2(\M, \vphi)_+$}.

\begin{Def} \upshape
Let $A, B\in\mathscr L(  L^2(\M, \vphi))$ be reality preserving w.r.t. $  L^2(\M, \vphi)_+$.  
If   $A-B\unrhd 0 \wrt   L^2(\M, \vphi)_+$, then we write this as  
$A\unrhd B \wrt   L^2(\M, \vphi)_+$.

\end{Def}

The following properties do not hold for ordinary operator inequalities and exemplify why OPOIs can be very useful in applications.

\begin{Lemm}
Let $A,B,C,D\in\mathscr L(  L^2(\M, \vphi))$. Suppose $A\unrhd B\unrhd0\wrt  L^2(\M, \vphi)_+$ and $C\unrhd D\unrhd0\wrt  L^2(\M, \vphi)_+$. Then we have $AC\unrhd BD\unrhd0\wrt  L^2(\M, \vphi)_+$.
\end{Lemm}
\begin{proof}
For proof, see, e.g., \cite{Miura2003,Miyao2016}. 
\end{proof}

\section{Stability of the magnetic properties of  ground states }\label{Sect3}
\subsection{Magnetic vectors}

In this section, we provide a fundamental framework describing a theory of   stability of magnetism in many-electron systems on {\it finite} lattices. The results of this section are essential for the discussion in the following sections.

Let $(\Lambda, E_{\Lambda})$ be a finite connected  graph.
In specific applications,  $\Lambda$ corresponds to  a lattice on which electrons move around. As we will see in Sections \ref{Sect5} and \ref{Sect6}, the edge set $E_{\Lambda}$   is determined consistently with the  Hamiltonian describing the many-electron system under consideration.

Consider a system in which there are  $ N$ electrons in $\Lambda$. The Hilbert space of such a system is given by 
\be
\F_{\Lambda, N}=\bigwedge^N \big(\ell^2(\Lambda) \oplus \ell^2(\Lambda)\big), \label{AntiSymm}
\ee
where $\bigwedge^N$ indicates the $N$-fold antisymmetric tensor product.
It is convenient to regard $\F_{\Lambda, N}$ as a subspace of the fermionic Fock space $\F_{\Lambda}=\bigoplus_{n=0}^{2|\Lambda|}\F_{\Lambda, N}$. One of the reasons for this is that we can introduce the creation and annihilation operators of electrons.
Let us denote the creation and annihilation operators of an electron with spin $\sigma$ at vertex $x$ as $c_{x\sigma}^*$ and $c_{x\sigma}$, respectively. It is well known that these operators satisfy the following anticommutation relations: 
\be
\{c_{x\sigma}, c_{x'\sigma'}^*\}=\delta_{xx'} \delta_{\sigma\sigma'},\ \ \{c_{x\sigma}, c_{x'\sigma'}\}=0,
\ee
where $\delta_{ab}$ stands for  Kronecker's delta. The Fock vacuum in $\F_{\Lambda}$ is denoted by $|\varnothing\ra_{\Lambda}$. The definitions related to the fermionic Fock space are explained in detail in \cite{Arai2016, Bratteli1997}.

 Various magnetic properties of this system are described by  using spin operators;
the  spin operators at vertex $x\in \Lambda$,  ${\bs S}_{x}=(S_x^{(1)}, S_x^{(2)}, S_x^{(3)})$,  are defined by
\be
S_x^{(j)}=\frac{1}{2} \sum_{\sigma, \sigma\rq{}=\up, \down}c_{x\sigma}^* \big(s^{(j)}\big)_{\sigma \sigma\rq{}} c_{x\sigma\rq{}},\ \ j=1, 2, 3, \label{DefSpinC}
\ee
where $s^{(j)}$ are the Pauli matrices: 
\be
s^{(1)}=\begin{pmatrix}
0 & 1 \\
1 & 0
\end{pmatrix}, \ \ 
s^{(2)}=\begin{pmatrix}
0 & -\im \\
\im & 0
\end{pmatrix},\ \ 
s^{(3)}=\begin{pmatrix}
1 & 0 \\
0 & -1
\end{pmatrix}.
\ee
These  operators satisfy the standard commutation
relations:
\be
[S^{(i)}_x , S^{(j)}_y] = \im \delta_{xy} \sum_{k=1 ,2, 3}
\vepsilon_{ijk} S^{(k)}_x, 
\ee
where $\vepsilon_{ijk}$ is the Levi--Civita symbol.

The total momentum operators, ${\bs S}_{\Lambda}=\big(S_{\Lambda}^{(1)}, S_{\Lambda}^{(2)}, S_{\Lambda}^{(3)} \big)$,
 are defined by 
 \be
 S_{\Lambda}^{(i)}=\sum_{x\in \Lambda} S_x^{(i)},\ \ i=1, 2, 3,
 \ee 
and,  the Casimir operator is denoted by ${\bs S}^2_{\Lambda}$:
\be
{\bs S}_{\Lambda}^2=\big(S_{\Lambda}^{(1)}\big)^2+\big(S_{\Lambda}^{(2)}\big)^2+\big(S_{\Lambda}^{(3)}\big)^2.
\ee
The following wording will be used frequently in this paper.
\begin{Def} \upshape
We say that a state $\psi$ in $\F_{\Lambda, N} $ has  total spin $S$, if 
${\bs S}_{\Lambda}^2\psi=S(S+1)\psi$ holds.
\end{Def}

Here, we introduce subspaces of $\F_{\Lambda, N}$ that are useful for the study of magnetism.
\begin{Def}\upshape
Let $\h_{\Lambda, N}$ be a subspace of $\F_{\Lambda, N}$. Let $P$ be the orthogonal projection 
from $\F_{\Lambda, N}$ to $\h_{\Lambda, N}$. We say that $\h_{\Lambda, N}$ is {\it reducible}, if $P$ commutes with $S_{\Lambda}^{(1)}, S_{\Lambda}^{(2)}$ and $S_{\Lambda}^{(3)}$.
We denote the restriction of $S_{\Lambda}^{(i)}$ to $\h_{\Lambda, N}$ by the same symbol, 
if no confusion occurs.
Let  $\mathscr{H}_{\Lambda, N}$  be the set of all reducible subspaces of $\F_{\Lambda, N}$.
\end{Def}
A typical example, which will be used in later sections, is given below.
\begin{Exa}\label{ElHilEx}\upshape
The fermionic Fock space $\F_{\Lambda, N=|\Lambda|}$ for $|\Lambda|$ electrons is of particular importance in the study of strongly correlated electron systems. Especially, the following subspace  of $\F_{\Lambda, |\Lambda|}$ will be used in the study of the Heisenberg model:
\be
\h_{\Lambda, |\Lambda|}=Q_{\Lambda} \F_{\Lambda, |\Lambda|},\ \  Q_{\Lambda}=\prod_{x\in \Lambda} (n_{x\up}-n_{x\down})^2.
\ee
 The subspace $\h_{\Lambda, |\Lambda|}$ describes states in which each site is occupied by a single electron. It is not hard to check that $\h_{\Lambda, |\Lambda|}$ is reducible. See Section \ref{Sect5} for detailed discussions.
\end{Exa}

In this paper, we will consider various systems in which electrons interact with environmental systems. As environmental systems, we have in mind  lattice vibrations (phonons) and quantized radiation fields (photons). To describe such systems, let us introduce some terms.
\begin{Def}\label{DefIEESP}\upshape
Fix $\h_{\Lambda, N} \in \mathscr{H}_{\Lambda, N}$, arbitrarily.
We say that a Hilbert space $\frak{X}$ is an {\it interacting  electron-environment  (IEE) space} associated with $\h_{\Lambda, N}$, if 
there exists an isometric linear mapping $\kappa$ from $\h_{\Lambda, N}$ into $\mathfrak{X}$.
We denote by $\mathscr{I}(\h_{\Lambda, N})$ the set of all IEE spaces associated with $\h_{\Lambda, N}$. Trivially, the $N$-electron space $\F_{\Lambda, N}$ belongs to $\mathscr{I}(\h_{\Lambda, N})$. 
\end{Def}
For each $M\in \mathrm{spec}(S_{\Lambda}^{(3)} \restriction \h_{\Lambda, N})$, the closed subspace
$\X[M]=\ker(S^{(3)}_{\Lambda}\restriction \h_{\Lambda, N}-M) \cap \X$ is called the {\it $M$-subspace}.\footnote{
The total spin operator $S_{\Lambda}^{(3)} \restriction \h_{\Lambda, N}$ can be naturally extended to an operator on $\X$. We will write this extension with the same symbol unless there is confusion.
}
The following decomposition will be fundamental for our study:
\be
\X=\bigoplus_{M\in \mathrm{spec}(S_{\Lambda}^{(3)} \restriction \h_{\Lambda, N})} \X[M]. \label{MDecomp}
\ee
 \begin{Exa}\upshape
 Let us discuss a typical example of IEE space here. Given a separable complex Hilbert space $\mathfrak{K}$,
  define $\X=\h_{\Lambda, N} \otimes \mathfrak{K}$. 
 Then $\X$ is an IEE space. To see this, fix a normalized vector $\mu_0\in \mathfrak{K}$, arbitrarily. The mapping 
 $\kappa: \xi\ni \h_{\Lambda, N}\mapsto \xi\otimes \mu_0\in \frak{K}$ gives rise to the isometric linear mapping from $\h_{\Lambda, N}$ into $\X$. The $M$-subspace of $\X$ is given by 
 $\X[M]=\h_{\Lambda, N}[M]\otimes \frak{K}$, where $\h_{\Lambda, N}[M]$ is the $M$-subspace of $\h_{\Lambda, N}$. This type of IEE space will be used frequently in this paper; see Sections \ref{Sect5} and \ref{Sect6}.
 
\end{Exa}

For  a given  $\frak{X}\in \mathscr{I}(\h_{\Lambda, N})$, 
we set 
\be
\A_{\Lambda}(\X)=\{A\in \mathscr{L}(\X) : [A, S^{(3)}_{\Lambda}]=0\},\label{DefAX}
\ee
where $\mathscr{L}(\X)$ indicates the set of all bounded operators  on $\X$.
\begin{Def}\upshape
Define
\be
\mathscr{Q}(\X)=\Big\{
\{\M, \vphi\} : \mbox{$\M$ is a von Neumann subalgebra of $\A_{\Lambda}(\X)$ and  $\vphi\in \mathfrak{W}_0(\M)$}
\Big\}.
\ee
Here, recall that $\mathfrak{W}_0(\M)$ is the set of semi-finite faithful normal weights on $\M$.
Each element in $\mathscr{Q}(\X)$ is called an {\it IEE system} associated with $\X$.
\end{Def}
Each $\{\M, \vphi\}\in \mathscr{Q}(\frak{X})$ gives rise to  the  standard form $ F (\M, \vphi)=\{\M, L^2(\M, \vphi), L^2(\M, \vphi)_+, J_{\vphi}\}$ as discussed in Section \ref{PrelSta}.  
We readily confirm  the following  decomposition:
\be
\M=\bigoplus_{M\in \mathrm{spec}(S_{\Lambda}^{(3)} \restriction \h_{\Lambda, N})} \M[M],\ \ 
\M[M]=P_M \M P_M,
\ee
where $P_M$ is the orthogonal projection from $\X$ to $\X[M]$. For simplicity of notation, we set
$\vphi_M=\vphi\restriction \M[M]$.
Because $\{\M[M], \vphi_M\}\in \mathscr{Q}(\X[M])$, we have the corresponding decomposition:
\be
F(\M, \vphi)= \bigoplus_{M\in \mathrm{spec}(S_{\Lambda}^{(3)} \restriction \h_{\Lambda, N})} F(\M[M], \vphi_M),
\ee
where, given standard forms $F(\M_1, \vphi_1)$ and $F(\M_2, \vphi_2)$,  $F(\M_1, \vphi_1) \oplus F(\M_2, \vphi_2)$ denotes the following  standard form:
\be
\{
\M_1\oplus \M_2, L^2(\M_1, \vphi_1) \oplus L^2(\M_2, \vphi_2), L^2(\M_1, \vphi_1)_+ \oplus L^2(\M_2, \vphi_2)_+,  J_1\oplus J_2
\}.
\ee

\begin{Def}\label{MagVecDef}\upshape
Suppose we are given an IEE system $\{\M, \vphi\}$. 
A vector $\psi\in L^2(\M, \vphi)(=\frak{X})$ is {\it magnetic associated with $\{\M, \vphi\}$}, if the following are satisfied:
\begin{itemize}
\item[\rm (i)]
There exists a non-negative number, $S_{\psi}$, such that ${\bs S}_{\Lambda}^2\psi=S_{\psi}(S_{\psi}+1)\psi$.
Namely, $\psi$ has  total spin $S_{\psi}$.
\item[\rm (ii)] 
For all  $M\in \mathrm{spec}(S_{\Lambda}^{(3)} \restriction \h_{\Lambda, N})$ with $|M| \le S_{\psi}$,   $P_M\psi$ is {\it strictly} positive w.r.t. $L^2(\M[M], \vphi_M)_+$.
\end{itemize}
We denote by $G(\M, \vphi)$ the set of all magnetic vectors associated with $\{\M, \vphi\}$.
\end{Def}
An important point in the above definition is the condition on the positivity of the vector in (ii).
Roughly speaking, this  strict positivity characterizes the magnetic structure of the vector; 
this will be made clear in the proof of Theorem \ref{BasicThm}.

In what follows, we always assume that $G(\M, \vphi)$ is nonempty. 
The following theorem is a foundation of the theory we will construct in this paper. 

\begin{Thm}\label{BasicThm}
The mapping $S_{\bullet}: \psi\in G(\M, \vphi) \mapsto S_{\psi}$ is constant.
\end{Thm}
\begin{proof}
Let $\psi_1, \psi_2\in G(\M, \vphi)$. We set $\psi_{1, 0}=P_{M=0} \psi_1$ and $\psi_{2, 0}=P_{M=0} \psi_2$,
Then $\psi_{1, 0}$ and $\psi_{2, 0}$ are strictly positive w.r.t. $L^2(\M[0], \vphi_0)_+$, which implies that 
$\la \psi_{1, 0}|\psi_{2, 0}\ra>0$. 
In addition, because $P_M$ commutes with ${\bs S}_{\Lambda}^2$, it holds that 
${\bs S}_{\Lambda}^2\psi_{1, 0}=S_{\psi_1}(S_{\psi_1}+1)\psi_{1, 0}$ and ${\bs S}_{\Lambda}^2\psi_{2, 0}=S_{\psi_2}(S_{\psi_2}+1)\psi_{2, 0}$.
Hence, we have
\begin{align}
S_{\psi_1}(S_{\psi_1}+1) \la \psi_{1, 0}|\psi_{2, 0}\ra=\la {\bs S}_{\Lambda}^2\psi_{1, 0}|\psi_{2, 0}\ra
=\la \psi_{1, 0}|{\bs S}^2_{\Lambda}\psi_{2, 0}\ra=S_{\psi_2}(S_{\psi_2}+1) \la \psi_{1, 0}|\psi_{2, 0}\ra.
\end{align}
Because $\la \psi_{1, 0}|\psi_{2, 0}\ra>0$, this implies that $S_{\psi_1}=S_{\psi_2}$. 
\end{proof}

From Theorem \ref{BasicThm}, the following definition makes sense.
\begin{Def} \upshape
The value of $S_{\bullet}$ on $G(\M, \vphi)$ is called the {\it total spin of the IEE  system  } $\{ \M, \vphi\}$ and is  denoted by  $S(\M, \vphi)$. 
\end{Def}

We will mention here the unitary equivalence between the two IEE systems for later convenience.

\begin{Def}\label{IEESEquiI}\upshape
Consider  two elements $\X$ and $\X'$ in $\mathscr{I}(\h_{\Lambda, N})$.
We say that $\{\M, \vphi\}\in \mathscr{Q}(\X)$ and $\{\M', \vphi'\}\in \mathscr{Q}(\X')$ are {\it equivalent} if there exists a unitary operator $U: \X\to \X'$ satisfying 
$\M'=U\M U^{-1}$ and $\vphi'(\bullet)=\vphi(U^{-1}\bullet U)$. This is denoted symbolically as $
\{\M, \vphi\} \simeq \{\M', \vphi'\}
$.
\end{Def}
The following proposition should be clear from the definition.

\begin{Prop}
Suppose that $\{\M, \vphi\}\in \mathscr{Q}(\X)$ and $\{\M', \vphi'\}\in \mathscr{Q}(\X')$ are equivalent.
Then $F(\M, \vphi)$ and $F(\M', \vphi')$ are equivalent in the following sense:
Let $U : \X\to \X'$ be the unitary operator that gives this equivalence. Then
\be
UL^2(\M, \vphi)_+=L^2(\M', \vphi')_+,\ J_{\vphi}=U^{-1}J_{\vphi'}U.
\ee
In addition, we have $S(\M, \vphi)=S(\M', \vphi')$.
\end{Prop}

The proposition implies the following.
Suppose that two IEE systems  $\{\M, \vphi\}$ and $\{\M', \vphi'\}$
 are equivalent. The equivalence of $F(\M, \vphi)$ and $F(\M', \vphi')$, derived from this hypothesis, means that the magnetic properties of the two systems are equivalent. As one would naturally expect from this, the values of the total spin of the two systems coincide in this case.

\subsection{Magnetic properties of ground states}

 In this section, the terminology ``Hamiltonian on $\Lambda$'' indicates a self-adjoint operator on $\mathfrak{X}\in \mathscr{I}(\h_{\Lambda, N})$, bounded from below.
 A more detailed characterization of the term Hamiltonian will be given in Appendix \ref{SectA}.
 
 Given a Hamiltonian $H$ on $\X$, we set $E(H)=\inf \mathrm{spec}(H)$. We say that $H$ has a {\it ground state} if $E(H)$ is an eigenvalue of $H$. In this case, the eigenvectors of $E(H)$ are called the ground states of $H$, and the value $E(H)$ is called the {\it ground state energy} of $H$.
We consider  Hamiltonians that are consistent in the following sense with an IEE system, which characterizes the magnetic structure of the system:
\begin{Def}\label{DefAHamiC}\upshape
Suppose that we are given an IEE system $\{\M, \vphi\}$ associated with $\X$.
Let $H$ be a Hamiltonian acting in $\X$. 
We say that $H$ is {\it adapted to $\{\M, \vphi\}$}, if the following are satisfied:
\begin{itemize}
\item[\rm (i)] $e^{-\beta H}$ commutes with the total spin operators, $S_{\Lambda}^{(1)}, S_{\Lambda}^{(2)}$ and $S_{\Lambda}^{(3)}$ for all $\beta \ge 0$.
Note that if $H$ is a bounded operator, this condition implies that $H$ and $S_{\Lambda}^{(i)}\, (i=1,2,3)$ are commutative in the usual sense.

\item[\rm (ii)] $H$ has a ground state.
\item[\rm (iii)]  For each $M\in \mathrm{spec}(S_{\Lambda}^{(3)} \restriction \h_{\Lambda, N})$, the semigroup $\{e^{-\beta H_M}\}_{\beta \ge 0}$ is ergodic w.r.t. $L^2(\M[M], \vphi_M)_+$, where $H_M=H\restriction \X[M]$. 
Here, recall that $\X[M]$ stands for the $M$-subspace of $\X$, see \eqref{MDecomp}.
\end{itemize}
The set of all Hamiltonians adapted to $\{\M, \vphi\}$ is denoted by $A_{\Lambda, N}(\M, \vphi)$.

\end{Def}

The condition (iii) of  Definition \ref{DefAHamiC} implies that the Hamiltonian under consideration is consistent with the magnetic structure of the system. This point is essential in the following discussion.

\begin{Prop}\label{GSinG}
Suppose that we are given a Hamiltonian  $H$ in  $A_{\Lambda, N}(\M, \vphi)$. Then the ground states of $H$    are $(2S(\M, \vphi)+1)$-fold degenerate and have total spin $S(\M, \vphi)$.
In addition, we can choose a ground state of $H$ to belong to $G(\M, \vphi)$.
\end{Prop}
\begin{proof}
We denote by $\psi_H$ a ground state of $H$. There exists an $M_0\in \mathrm{spec}(S^{(3)}_{\Lambda} \restriction \h_{\Lambda, N})$ such that $\psi_{H, M_0}:=P_{M_0}\psi_H\neq 0$.
By (iii) of Definition \ref{DefAHamiC}, the semigroup $\{e^{-\beta H_{M_0}}\}_{\beta \ge 0}$ is ergodic w.r.t. $ L^2(\M[M_0], \vphi_{M_0})_+$. Hence, by applying  Theorem \ref{pff}, we find that 
$\psi_{H, M_0}$ is the unique ground state of $H_{M_0}$ and can be chosen   to be  strictly positive w.r.t. $ L^2(\M[M_0], \vphi_{M_0})_+$. Because $H_{M_0}$ commutes with ${\bs S}^2_{\Lambda}$ ,
there exists a non-negative number $S$ such that ${\bs S}^2_{\Lambda} \psi_{H, M_0}=S(S+1)\psi_{H, M_0}$ holds.

Let $S^{\pm}_{\Lambda}=S_{\Lambda}^{(1)} \pm \im S_{\Lambda}^{(2)}$. We set $\psi_{H, M_0+1}=S_{\Lambda}^{(+)} \psi_{H, M_0}$. Because $H_{M_0+1} S_{\Lambda}^{(+)}=S_{\Lambda}^{(+)} H_{M_0}$, we obtain $H_{M_0+1} \psi_{H, M_0+1}=E(H)\psi_{H, M_0+1} $, where $E(H)$ is the ground state energy of $H$. Namely, $\psi_{H, M_0+1}$ is a ground state of $H_{M_0+1}$ as well. 
As $\{e^{-\beta H_{M_0+1}}\}_{\beta \ge 0}$ is ergodic w.r.t. $L^2(\M[M_0+1], \vphi_{M_0+1})_+$, Theorem \ref{pff} shows that this $\psi_{H, M_0+1}$  is the unique ground state of $H_{M_0+1}$ and can be either strictly positive or strictly negative w.r.t $L^2(\M[M_0+1], \vphi_{M_0+1})_+$. 
Moreover, since ${\bs S}_{\Lambda}^2$ commutes with $S_{\Lambda}^{(+)}$, $\psi_{H, M_0+1}$ also has total spin $S$.

Next, set $\psi_{H,  M_0-1}=S_{\Lambda}^{(-)} \psi_{H, M_0}$. By the similar arguments as above, we see that $\psi_{H,  M_0-1}$ is the unique ground state of $H_{M_0-1}$, can be strictly positive or negative w.r.t. $L^2(\M[M_0-1], \vphi_{M_0-1})_+$ and has   total spin $S$.

Repeating the above arguments, we can construct a sequence of vectors $\{\psi_{H, M} : |M| \le S\}$ that has the following properties: 
\begin{itemize}
\item $\psi_{H, M}$ is the unique ground state of $H_M$ and can be strictly positive or negative w.r.t. $ L^2(\M[M], \vphi_{M})_+$.
\item Each $\psi_{H, M}$ has   total spin $S$.
\end{itemize}
 When $|M|>S$, we define $\tilde{\psi}_{H, M}:=0$, and when $|M| \le S$, we define $\tilde{\psi}_{H, M} := \psi_{H, M}$ or $\tilde{\psi}_{H, M} := -\psi_{H, M}$, depending on whether $\psi_{H, M}$ is strictly positive or negative, respectively. Then, if we set  $\tilde{\psi}_H:=\bigoplus_M\tilde{\psi}_{H, M}$, then $\tilde{\psi}_H$ is a ground state of $H$ and belongs to $G(\M, \vphi)$. Using Theorem \ref{BasicThm}, we find that $S=S(\M, \vphi)$.
\end{proof}

The following notation is introduced to describe the results that follow concisely:
\begin{Def}\label{DefSH}\upshape
 Given a Hamiltonian $H\in A_{\Lambda, N}(\M, \vphi)$, 
  we denote by $S_H$ the total spin of the ground states of $H$.
\end{Def}

Proposition \ref{GSinG} can be rephrased as follows.

\begin{Coro}\label{StaGsFixH}
The mapping $S_{\bullet} : H\in A_{\Lambda, N}(\M, \vphi)\mapsto S_{H}$  is constant and satisfies $S_{H}=
S(\M, \vphi)$ for all $H\in A_{\Lambda, N}(\M, \vphi)$. 
\end{Coro}
\begin{Rem}\upshape
\begin{itemize}
\item In strongly correlated electron systems,  the total spins of  ground states are one of the fundamental indicators describing the magnetic properties of the system under consideration.
Corollary \ref{StaGsFixH} asserts that all Hamiltonians belonging to $A_{\Lambda, N}(\M, \vphi)$ always have a constant value of the total spin in their ground state. Proposition \ref{GSinG} and Corollary \ref{StaGsFixH}
 clearly express the basic idea of this paper:  to extract magnetic properties independent of the details of individual Hamiltonians.
\item
The proof of Proposition \ref{GSinG} implies the following important fact:
recall that, immediately below Definition \ref{MagVecDef}, we explained that the strict positivity with respect to $ L^2(\M[M], \vphi_M)_+$ characterizes the magnetic properties of  states.
By (iii) of Definition \ref{DefAHamiC} and Theorem \ref{pff}, we know that all the ground states of  Hamiltonians belonging to $ A_{\Lambda, N}(\M, \vphi)$ are strictly positive with respect to $ L^2(\M[M], \vphi_M)_+$. 
In other words, the ground states all have  common magnetic properties. One consequence of this fact is that the ground states all have the same total spin,  as stated in Corollary \ref{StaGsFixH}.

\item 
We have not used the specific structure of the graph $(\Lambda, E_{\Lambda})$  in this section.
However, in actual applications, when we prove that a given Hamiltonian exhibits the condition (iii) of Definition \ref{DefAHamiC}, we will need detailed information about the graph; 
see, Sections \ref{Sect5} and \ref{Sect6}.

\end{itemize}
\end{Rem}

\subsection{Stability classes }\label{SectStCl}

Let  $\frak{X}, \frak{Y}\in \mathscr{I}(\h_{\Lambda, N})$. Suppose that there exists an isometric linear mapping  from $\frak{Y} $ into $\frak{X}$.
Hence, we can naturally  identify $\frak{Y}$ as a closed subspace of $\frak{X}$.
We denote by $P_{\frak{Y, X}}$ the orthogonal projection from $\frak{X}$ to $\frak{Y}$.
Needless to say, $P_{\frak{Y}, \frak{X}}$ commutes with $S_{\Lambda}^{(i)}\ (i=1, 2, 3)$.
Let $\M$ be   a von Nenmann algebra on $\frak{X}$ satisfying 
$\M\subseteq \A_{\Lambda}(\frak{X})$. Then $\N:=P_{\frak{Y, X}} \M P_{\frak{Y, X}}$ is a von Neumann algebra on $\frak{Y}$.
 Let $\vphi\in \mathfrak{W}_0(\M)$. 
 We assume that $\psi: =\vphi\restriction \N$ is semi-finite. 
 Then we have the associated standard forms $ F (\M, \vphi)$
 and $ F (\N, \psi)$.  Recall that,  since $\mathfrak{X}$ is separable, 
 we have the identifications $\X=L^2(\M, \vphi)$ and $\frak{Y}=L^2(\N, \psi)$.

\begin{Def}\label{DefConForm}\upshape
We say that  the standard forms $ F (\M, \vphi)$
 and $ F (\N, \psi)$
are  {\it consistent}, if the following (i), (ii) and (iii) are satisfied:
\begin{itemize}
\item[(i)] $\vphi \circ\mathscr{E}_{\frak{N, M}}=\vphi$, where
$\mathscr{E}_{\frak{N, M}}(x)=P_{\frak{Y, X}} x P_{\frak{Y, X}}\ (x\in \M)$.
\item[(ii)]  $L^2(\N, \psi)_+ \subseteq L^2(\M, \vphi)_+$.
\item[(iii)] $P_{\frak{Y, X}}L^2(\M, \vphi)_+=L^2(\N, \psi)_+$.
\end{itemize}
We write this as $ F (\M, \vphi) \LRA  F (\N, \psi)$.
\end{Def}

\begin{Rem}\label{RemCons} \upshape
\begin{itemize}
\item[1.] Because $\mathscr{E}_{\N, \M}$ is  a  conditional expectation of $\M$ onto $\N$ with respect to $\vphi$, Definition \ref{DefConForm} is consistent with Definition \ref{DefCons1}.
\item[2.] From  (ii) and (iii), it follows that $P_{\frak{Y, X}} \unrhd 0$ w.r.t. $L^2(\M, \vphi)_+$. This property  will often be  useful in the following arguments.
\item[3.] Let $\{\M_1, \vphi_1\}, \{\M_2, \vphi_2\}$ and $\{\M_3, \vphi_3\}$ be pairs of von Neumann algebra and faithful  semi-finite normal weight, respectively.
If $F(\M_1, \vphi_1) \LRA F(\M_2, \vphi_2)$ and $F(\M_2, \vphi_2)\LRA F(\M_3, \vphi_3)$, then 
$F(\M_1, \vphi_1) \LRA F(\M_3, \vphi_3)$ holds.
\end{itemize}
\end{Rem}

\begin{Thm}\label{StabThmEx}
If $ F (\M, \vphi) \LRA  F (\N, \psi)$, then we have
$
S(\M, \vphi)=S(\N, \psi)
$.
\end{Thm}
\begin{proof}
For simplicity, we assume that $0\in \mathrm{spec}(S_{\Lambda}^{(3)} \restriction \h_{\Lambda, N})$.
Take $\xi\in G(\M, \vphi)$ and $\eta\in G(\N, \psi)$, arbitrarily. 
Set $\xi_0=P^{\X}_{M=0} \xi$ and $\eta_0=P^{\frak Y}_{M=0} \eta$, where $P^{\X}_M$ (resp.  $P_M^{\frak Y}$) is the orthogonal projection  from $\X$ (resp. $\frak Y$) to the $M$-subspace  $\X[M]$ (resp. $\frak{Y}[M]$).
Note that $P^{\frak{Y}}_M=P_{\frak{Y, X}} P^{\X}_M$ holds.
Because $\xi_0$ is strictly positive w.r.t. $L^2(\M[0], \psi_0)_+$ and $P_{\frak{Y, X}} \unrhd 0$ w.r.t. $L^2(\M, \vphi)_+$, it holds that $\la P_{\frak{Y, X}}\xi_0|\eta_0\ra>0$ due to Proposition \ref{NonVOver}.
Set $S_{\xi}=S(\M, \vphi)$ and $S_{\eta}=S(\N, \psi)$. Because $P_{\frak{Y, X}}$ commutes with ${\bs S}^2_{\Lambda}$, we have
${\bs S}^2_{\Lambda} \xi_0=S_{\xi}(S_{\xi}+1) \xi_0$ and ${\bs S}^2_{\Lambda} \eta_0=S_{\eta}(S_{\eta}+1) \eta_0$. Hence, 
\begin{align}
S_{\xi}(S_{\xi}+1) \la P_{\frak{Y, X}}\xi_0|\eta_0\ra=\la P_{\frak{Y, X}} {\bs S}^2_{\Lambda}\xi_0|\eta_0\ra
=\la P_{\frak{Y, X}}\xi_0|{\bs S}_{\Lambda}^2\eta_0\ra=S_{\eta}(S_{\eta}+1) \la P_{\frak{Y, X}}\xi_0|\eta_0\ra,
\end{align}
which implies that $S_{\xi}=S_{\eta}$.
\end{proof}
A point of this theorem is that it is a claim about the total spins of magnetic vectors in {\it different} Hilbert spaces. Let us recall that Theorem \ref{BasicThm} was about magnetic vectors in the same Hilbert space.
This difference is essential for the construction of a theory of magnetic stability.

\begin{Def}\label{DefStaCl}\upshape
Suppose that $\frak{X}\in \mathscr{I}(\h_{\Lambda, N})$ is given.
Fix $\{\M_0, \vphi_0\} \in \mathscr{Q}(\frak{X})$,  arbitrarily.
The {\it stability class associated with} $\{\M_0, \vphi_0\} $ is defined by 
\be
\mathscr{C}_{\Lambda, N}(\M_0, \vphi_0)
=\bigcup_{{\h\in \mathscr{I}(\h_{\Lambda, N})}\atop{ \frak{X} \subseteq \h}}\big\{\{\M, \vphi\} \in \mathscr{Q} (\h) :  F (\M, \vphi) \LRA  F (\M_0, \vphi_0)
\big\}.
\ee
\end{Def}

In other words, the stability class associated with $\{\M_0, \vphi_0\} $  is defined as the collection of all  IEE systems that are consistent with 
 $\{\M_0, \vphi_0\} $. 
 \begin{figure}[h]
 \begin{center}
 \includegraphics{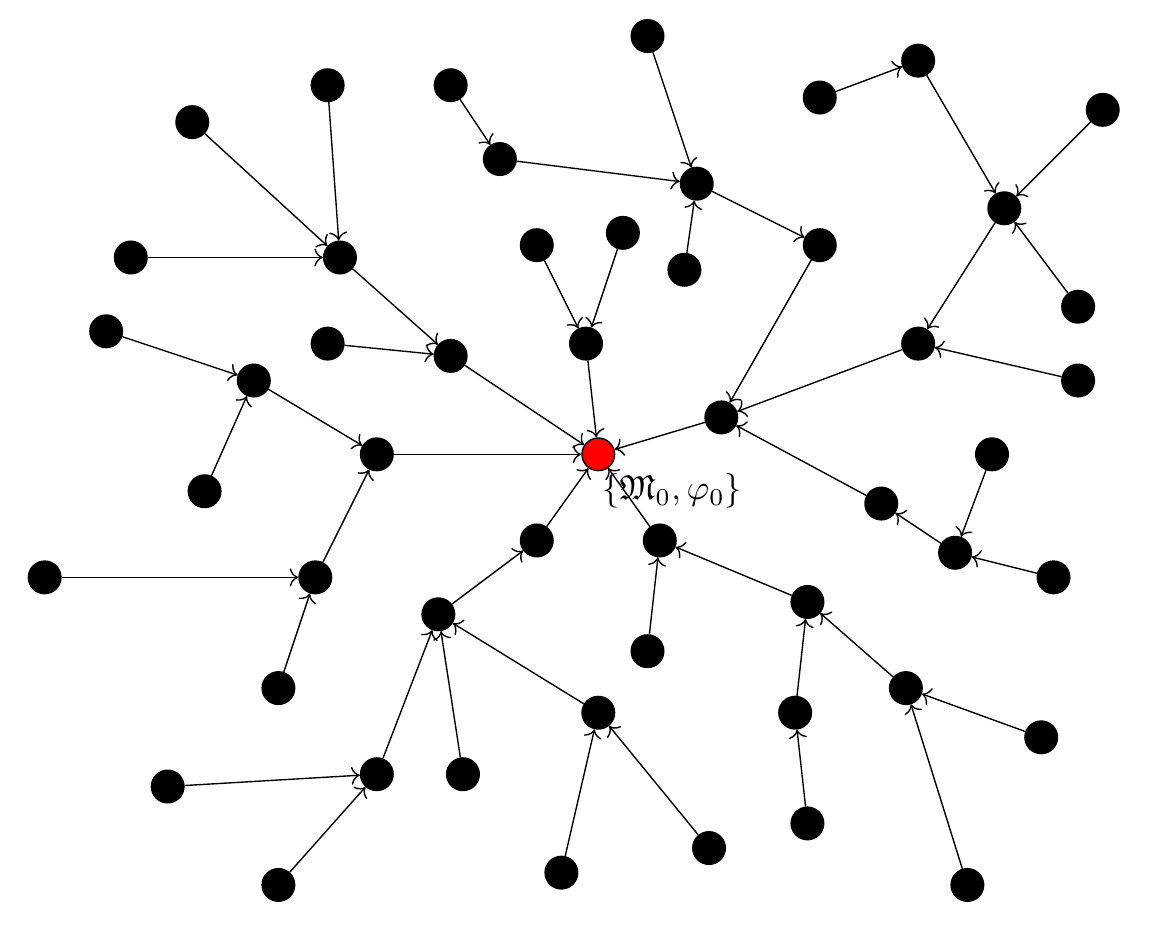}
  \caption{
        $\mathscr{C}_{\Lambda, N}(\M_0, \vphi_0)$ corresponds to a  directed tree with the root $\{\M_0, \vphi_0\}$.
        }\label{Fig1}
        \end{center}
 \end{figure}
 Figure \ref{Fig1} shows a visual image of the definition. More precisely, 
 $\mathscr{C}_{\Lambda, N}(\M_0, \vphi_0)$ can be regarded as a {\it directed tree with the root }$\{\M_0, \vphi_0\}$, more specifically,  an  arborescence,  in the following fashion:\footnote{From the standpoint of category theory, a functor describes this correspondence. In this paper, we will not enter into the category theory aspects of stability classes.}
 \begin{itemize}
 \item Each IEE system corresponds to a vertex.
 \item  Each pair of IEE systems $(\{\M, \vphi\}, \{\M', \vphi'\})$ connected by the relation $F(\M, \vphi) \to F(\M', \vphi')$ corresponds to an edge with the origin $\{\M, \vphi\}$ and terminus $\{\M', \vphi'\}$.
 \end{itemize}
In this way, a stability class $\mathscr{C}_{\Lambda, N}(\M_0, \vphi_0)$ can be thought of as a network propagating the positivity emanating from $\{\M_0, \vphi_0\}$.
 Recalling that the magnetic properties of the system under consideration are characterized by positivity, we see that all IEE systems belonging to  the stability class $\mathscr{C}_{\Lambda, N}(\M_0, \vphi_0)$ have the  magnetic properties that originate in $\{\M_0, \vphi_0\}$.

  From this definition and Theorem \ref{StabThmEx}, we immediately obtain the following corollary.

\begin{Coro}[Stability of the total spin]\label{StaGsC}
The mapping $S(\bullet): \{\M, \vphi\} \ni \mathscr{C}_{\Lambda, N}(\M_0, \vphi_0) \mapsto S(\M, \vphi)$
 is constant and satisfies $S(\M, \vphi)=S(\M_0, \vphi_0)$ for all $\{\M, \vphi\}\in \mathscr{C}_{\Lambda, N}(\M_0, \vphi_0)$.
\end{Coro}

Corollary \ref{StaGsC} asserts that all IEE systems belonging to the stability class $\mathscr{C}_{\Lambda, N}(\M_0, \vphi_0)$ have  the same total spin:  $S(\M_0, \vphi_0)$. This corollary is simple but  essential for applications.

The following proposition clarifies the basic correspondence between IEE systems and their corresponding stability  classes:
\begin{Prop}
Suppose that we are in the setting described in Definition \ref{DefConForm}.
If $ F (\M, \vphi) \LRA  F (\N, \psi)$, then $\mathscr{C}_{\Lambda, N}(\M, \vphi)\subseteq \mathscr{C}_{\Lambda, N}(\N, \psi)$. 
\end{Prop}
\begin{proof}
Take $\{\M\rq{}, \vphi\rq{}\} \in \mathscr{C}_{\Lambda, N}(\M, \vphi)$, arbitrarily.
Hence, we have 
\be
F(\M\rq{} \vphi\rq{}) \LRA F(\M, \vphi) \LRA F(\N, \psi),
\ee
which implies that $F(\M\rq{}, \vphi\rq{})\LRA F(\N, \psi)$ by 3. of Remark \ref{RemCons}. Thus, we conclude that
$\mathscr{C}_{\Lambda, N}(\M, \vphi)\subseteq \mathscr{C}_{\Lambda, N}(\N, \psi)$.
\end{proof}

Next, let us focus on the stability of magnetic properties of ground states of many-electron systems interacting with their environments.
\begin{Def}\upshape
Set
\be
\mathscr{A}_{\Lambda, N}(\M_0, \vphi_0)=\bigcup_{\{\M, \vphi\}\in \mathscr{C}_{\Lambda, N}(\M_0, \vphi_0)}
A_{\Lambda, N}(\M, \vphi), \label{DefALN}
\ee
where 
 $A_{\Lambda, N}(\M, \vphi)$ is defined in  Definition \ref{DefAHamiC}.
 
 The mapping $S_{\bullet}$ introduced in Definition \ref{DefSH} can be naturally extended to $\mathscr{A}_{\Lambda, N}(\M_0, \vphi_0)$. We denote this extended mapping by the same symbol.
\end{Def}

For each Hamiltonian $H\in \mathscr{A}_{\Lambda, N}(\M_0, \vphi_0)$,  the ground states of  $H$  have  magnetic properties (i.e.,  positivity) that originate in $\{\M_0, \vphi_0\}$. 
The following corollary follows from this fact:

\begin{Coro}\label{StaGsBigA}
The mapping $S_{\bullet} : H\in \mathscr{A}_{\Lambda, N}(\M_0, \vphi_0) \mapsto S_{H}$ is constant  and satisfies $S_{H}=S(\M_0, \vphi_0)$ for all $H\in \mathscr{A}_{\Lambda, N}(\M_0, \vphi_0)$.
\end{Coro}
\begin{proof}
Apply Corollaries \ref{StaGsFixH} and \ref{StaGsC}.
\end{proof}
Let us briefly explain why Corollary \ref{StaGsBigA} is helpful in practical applications.
Suppose we are given a Hamiltonian $H_0$ adapted to $\{\M_0, \vphi_0\}$. 
In many cases, the structure of the reference Hamiltonian $H_0$ is simple.
Its ground states have total spin $S(\M_0, \vphi_0)$. Let us  consider a Hamiltonian $H$ adapted to a particular system that interacts with the environment in a complex way.
  In general, the structure of $H$ is more complex than that of $H_0$. However, once we know that $H$ is in $\mathscr{A}_{\Lambda, N}(\M_0, \vphi_0) $, we know from Corollary \ref{StaGsBigA} that 
the total spin of the ground states of $H$ is identical to that of $H_0$, which is much easier to analyze.
 See Sections \ref{Sect5}, \ref{Sect6},  and \ref{Sect7} for specific examples.

\begin{Rem}\upshape
We explain the difference in terminology from previous studies to prevent readers from misunderstanding.
In the previous work \cite{Miyao2019}, the set of Hamiltonians $\mathscr{A}_{\Lambda, N}(\M_0, \vphi_0)$ was referred to as the stability class. For a more mathematically detailed discussion in this paper, we call  $\mathscr{C}_{\Lambda, N}(\M_0, \vphi_0)$ the stability class and distinguish it from $\mathscr{A}_{\Lambda, N}(\M_0, \vphi_0)$.
\end{Rem}

\subsection{Graph isomorphisms and stability classes}
Suppose that $G=(\Lambda, E)$ and $G\rq{}=(\Lambda', E')$ are isomorphic: $G\simeq G\rq{}$. 
In this case, what is the relationship between $\mathscr{C}^G_{\Lambda, N}(\M_0, \vphi_0)$ and $\mathscr{C}^{G'}_{\Lambda, N}(\M_0, \vphi_0)$? 
Here, to make the $G$-dependence clear, we denote the stability class on $G$ as $\mathscr{C}^G_{\Lambda, N}(\M_0, \vphi_0)$.

The result  in this subsection can be summarized  as follows.
\begin{Thm}\label{Universal}
If the two graphs are isomorphic, the corresponding stability classes are identical.
\end{Thm}

This seemingly obvious theorem will be needed when  discussing  realizations of  crystal lattices in Section \ref{Sect4}.

Let $\iota: \Lambda\to \Lambda\rq{}$ be the  isomorphism  between $G$ and  $G'$. 
We define the natural unitary operator $u$ from $\ell^2(\Lambda)$ onto $\ell^2(\Lambda\rq{})$ as follows:
\be
(u f)(x\rq{})=f(\iota^{-1}(x\rq{})),\ \ f\in \ell^2(\Lambda),\ x\rq{}\in \Lambda\rq{}. 
\ee
In this case, $U_{N}=\otimes^N u$ defines the  natural unitary operator from $\F_{\Lambda, N}$ onto $\F_{\Lambda\rq{}, N}$. We readily confirm the following:
\begin{align}
U_NS_x^{(i)}U_N^{-1}=S^{\prime(i)}_{\iota(x)},\   \ 
U_Nc_{x\sigma}U_N^{-1}=c^{\prime}_{\iota(x)\sigma}\ \ (x\in \Lambda), \label{SScc}
\end{align}
where $S_{x\rq{}}^{\prime(i)}$ and $c_{x\rq{}\sigma}^{\prime}$ stand for the spin and annihilation operators on $\F_{\Lambda', N}$, respectively.
In this sense, the structures of the spin operators in the Hilbert spaces $\F_{\Lambda, N}$ and $\F_{\Lambda', N}$ are almost the same.

Before we proceed, let us introduce one term.
\begin{Def}\label{HilIso}\upshape
Let $\h\in \mathscr{H}_{\Lambda, N}$ and $\h'\in \mathscr{H}_{\Lambda', N}$.
We call $\{\h, \h'\}$  an {\it isomorphic pair} if there exists a unitary operator $U : \h\to \h\rq{}$ such that 
\be
U_NS_x^{(i)}U_N^{-1}=S^{\prime(i)}_{\iota(x)}\ \ (i=1, 2, 3,\ x\in \Lambda), 
\ee
where $S^{\prime (i)}_{x\rq{}}$ stand for the spin operators on $\h\rq{}$.
\end{Def}
By \eqref{SScc}, we readily  see that $\{\F_{\Lambda, N}, \F_{\Lambda\rq{}, N}\}$ is an isomorphic pair.

By the above arguments, we obtain:
\begin{Lemm}\label{Lemm1}
The entire reducible subspaces on $\Lambda$, $\mathscr{H}_{\Lambda, N}$, and the entire reducible spaces on $\Lambda\rq{}$, $\mathscr{H}_{\Lambda\rq{}, N}$, are isomorphic in the following sense:
there exists a bijection $\kappa: \mathscr{H}_{\Lambda, N} \to \mathscr{H}_{\Lambda\rq{}, N}$ such that 
$\{\h, \kappa(\h)\}$ is an isomorphic pair for each $\h\in \mathscr{H}_{\Lambda, N}$.
\end{Lemm}

In order to prove Theorem \ref{Universal}, we prepare one more lemma:

\begin{Lemm}\label{I=I}
Let $\h_{\Lambda, N}\in \mathscr{H}_{\Lambda, N}$ and $\h_{\Lambda', N}\rq{} \in \mathscr{H}_{\Lambda', N}$. Suppose that $\{\h_{\Lambda, N},  \h_{\Lambda', N}\rq{}\}$ is an isometric pair. Then we have
$\mathscr{I}(\h_{\Lambda, N})=\mathscr{I}(\h_{\Lambda', N}\rq{})$.
\end{Lemm}
\begin{proof}
Because $\{\h_{\Lambda, N},  \h_{\Lambda', N}\rq{}\}$ is  an isometric pair, there is a unitary operator $U: \h_{\Lambda, N} \to \h_{\Lambda', N}\rq{}$.
Let $\X\in \mathscr{I}(\h_{\Lambda, N})$. Then there exists an isometric linear mapping $\kappa : \h_{\Lambda, N} \to \X$. We easily see that $\kappa \circ U^{-1}$ defines an isometric linear mapping from $\h_{\Lambda', N}\rq{}$
to $\X$. Hence, $\mathscr{I}(\h_{\Lambda, N}) \subseteq \mathscr{I}(\h_{\Lambda', N}\rq{})$. 
By exchanging the roles of $\mathscr{I}(\h_{\Lambda, N})$ and $\mathscr{I}(\h_{\Lambda', N}\rq{})$ in the above discussion, we get $\mathscr{I}(\h_{\Lambda', N}\rq{}) \subseteq \mathscr{I}(\h_{\Lambda, N})$. 
\end{proof}

\subsubsection*{\it Proof of Theorem \ref{Universal}}
 
By using  Lemmas \ref{Lemm1} and  \ref{I=I} and recalling Definition \ref{DefStaCl}, we obtain the desired result. \qed

\subsection{Comments on  the stability classes on finite lattices}

In this section, the stability classes on  finite graphs have been discussed in detail. Due to the structure of the paper, the reader may think that this section is just a preparation for the theory of infinite systems discussed in the next section.
However, recent advances in experimental techniques have led to the realization of electronic systems on small finite lattices in the laboratory, and our knowledge of the magnetic properties of  ground states of such systems is increasing; see \cite{Dehollain2020,Slot2017} and references therein. Therefore, the results of this section are expected to be  helpful in understanding the magnetic properties of such systems and their stability.

\section{Stability of magnetic orders}\label{Sect4}

\subsection{Magnetic orders}

So far, we have kept $\Lambda$ fixed and then discussed the magnetic structure of  each system.  In this section, $\Lambda$ is treated  as a  varying parameter.
By clarifying the relationship between magnetic systems for different $\Lambda$,  the concept of magnetic order can be defined mathematically; see Definition \ref{DefOrderM}.
Specific applications of the general theory constructed in this section will be discussed in Sections \ref{Sect5} and \ref{Sect6}.

Let  $G=(\mathbb{L}, E)$ be a countably infinite connected graph. 
We denote by $ P(\mathbb{L}) $ the   set of  all finite subsets of $\mathbb{L}$.
Given a $\varrho\in (0, 1)$, we consider a net of reducible Hilbert spaces: $
\{\h_{\Lambda, [\varrho|\Lambda|]} \in \mathscr{H}_{\Lambda, [\varrho |\Lambda|]} : \Lambda\in  P(\mathbb{L}) \}, 
$
where $[\cdot]$ stands for the Gauss symbol. The number $\varrho$ is called the {\it electron density}.
Let ${\sf H}_{\varrho}=\{\X_{\Lambda} \in \mathscr{I}(\h_{\Lambda, [\varrho |\Lambda|]}) : \Lambda\in  P(\mathbb{L}) \}$ be a net of IEE spaces, where  $\mathscr{I}(\h_{\Lambda, N})$ is defined in Definition \ref{DefIEESP}.
For any $\Lambda\in P(\mathbb{L})$, suppose that we are given an IEE system $\{\M_{\Lambda}, \vphi_{\Lambda}\}\in \mathscr{Q}(\X_{\Lambda})$.
We set 
\be
F_{\mathbb{L}}=\{\Lambda\subset \mathbb{L} : \mbox{$(\Lambda, E_{\Lambda})$ is an induced  subgraph of $G$, which is finite and connected}\}.\footnote{
Let us consider  two graphs $G=(V, E)$ and $G\rq{}=(V\rq{}, E\rq{})$. If $G\rq{}$ is a subgraph of $G$ and $G\rq{}$ contains all the edges $\{x, y\}\in E$ with $x, y\in V\rq{}$, then $G\rq{}$ is  an {\it induced subgraph of } $G$. See \cite{Diestel2017} for detail.
}\label{DefF(L)}
\ee
In this section, we will examine 
a net  of IEE systems $O_{\varrho}=\{\{\M_{\Lambda}, \vphi_{\Lambda}\} \in \mathscr{Q}(\X_{\Lambda}) : \X_{\Lambda}\in {\sf H}_{\varrho}, \  \Lambda\in  F_{\mathbb{L}} \}$ and  the corresponding net of the standard forms:
$ F (O_{\varrho}):=\{ F (\M_{\Lambda}, \vphi_{\Lambda}) : \Lambda\in  F_{\mathbb{L}} \}$.

We have seen in the previous section that when $\Lambda$ is fixed, the IEE  system $\{\M_{\Lambda}, \vphi_{\Lambda}\}$ naturally determines the corresponding magnetic structure. In this section, we consider the following question: Given a net of IEE systems, how can we define the magnetic order on the macroscopic scale, based on the discussion in the previous section? 
There must be some relation between  $\{\M_{\Lambda}, \vphi_{\Lambda}\}$  for different $\Lambda$ that consistently connects the magnetic properties. A macroscopic magnetic order must emerge from the consistent interconnection of the magnetic structures of the systems  $\{\M_{\Lambda}, \vphi_{\Lambda}\}$  corresponding to the various  $\Lambda$.
From this consideration, we adopt the following definition.

\begin{Def}\label{DefOrderM}\upshape
\begin{itemize}
\item[(i)] We say that $ O_{\varrho}$ is a {\it magnetic system} on $G$, if 
$ F (\M_{\Lambda\rq{}}, \vphi_{\Lambda\rq{}}) \LRA  F (\M_{\Lambda}, \vphi_{\Lambda})$  holds for all  $\Lambda, \Lambda\rq{}\in F_{\mathbb{L}}$ with $\Lambda\subset \Lambda\rq{}$.

\item[(ii)] The magnetic system $ O_{\varrho}$ induces a {\it magnetic order}, if  there exist a non-negative number $s$  and an increasing sequence  of  sets $\{\Lambda_n : n\in \BbbN\}\subset F_{\mathbb{L}}$ such that 
$\bigcup_{n=1}^{\infty} \Lambda_n=\mathbb{L}$ and  
\be
 S(\M_{\Lambda_n}, \vphi_{\Lambda_n})=s|\Lambda_n|+o(|\Lambda_n|)\ \ \mbox{as $n\to \infty$}. \label{LimitForm}
\ee

If $s$ is non-zero, then we say that the magnetic order  is {\it strict}.
\end{itemize}
\end{Def}

It is essential to the above definition that the positivities between the IEE systems for different  $\Lambda$ are consistently connected as stated in (i) of Definition \ref{DefOrderM}. Recalling that the positivity for each $\Lambda$, i.e., $L^2(\M_{\Lambda}, \vphi_{\Lambda})_+$,   describes the magnetic properties of the corresponding system, this definition implies that the consistent connection of the magnetic properties between different regions determines the magnetic order in the infinite (or macroscopic) system. (ii) of Definition \ref{DefOrderM} characterizes the magnetic order numerically using the spin density.

\begin{Rem}\upshape
 In the definition of the magnetic order, i.e., Definition \ref{DefOrderM}, we presuppose that an increasing sequence of sets $\{\Lambda_n : n\in \BbbN\}\subset F_{\mathbb{L}}$ with $\bigcup_n \Lambda_n=\mathbb{}L$ is given. 
Here, we discuss how we construct such a sequence. 
We assume that $G=(\mathbb{L}, E)$ is a directed graph for the sake of exposition. However, the same consideration holds even when $G$ is not a directed graph.
 For an oriented edge $e \in E$, the origin and the terminus of $e$ are denoted by $o(e)$ and $t(e)$, respectively.  Let $E_x=\{e\in E : o(e)=x\}$ be the set of edges with $o(e) =x \in \mathbb{L}$. A path $c$ of $G$ of length $n$ is a sequence $c =(e_1,\dots,  e_n)$ of oriented edges $e_i$ with $o(e_{i+1}) =t(e_i)\, (i =1, \dots, n-1)$. We denote by $\Omega_{x,n}(G)\ (x \in \mathbb{L}, n \in \BbbN)$ the set of all paths of length $n$ for which origin $o(c) =x$. 
Let us construct a graph  naturally  from $\Omega_{x,n}(G)$. To see this, we need to do a little preparation. For each   $c=(e_1, \dots, e_n)\in \Omega_{x, n}(G)$, set $v(c)=\{o(e_1), \dots, o(e_n), t(e_1), \dots, t(e_n)\} \subset \mathbb{L}$ and 
$e(c)=\{e_i : i=1, \dots n\}$.
$v(c)$ (resp. $e(c)$) is all the vertices (resp. edges) that the path $c$ contains. Define $\Lambda_n=\bigcup_{ c\in \Omega_{x, n}(G)} v(c) \subset \mathbb{L}$ and $E_{\Lambda_n}=\bigcup_{ c\in \Omega_{x, n}(G)}e(c)$. 
We readily confirm that the graph $(\Lambda_n, E_{\Lambda_n})$ is an induced graph of $G$, which is finite and connected. In addition, $\Lambda_n \subset \Lambda_{n+1}$ holds. Hence, we get a desired sequence $\{\Lambda_n : n\in \BbbN\}$.
\end{Rem}

 Next, we discuss the stabilities of the magnetic orders.
 \begin{Def}\label{ConsisOO}\upshape
Let ${\sf H}_{\varrho}^{\prime}=\{\X_{\Lambda}^{\prime} \in \mathscr{H}_{\Lambda, [\varrho |\Lambda|]} : \Lambda\in  P(\mathbb{L}) \}$ be another net of IEE spaces.
Suppose that we are given  a magnetic system  $O_{\varrho}^{\prime}=\{\{\M_{\Lambda}^{\prime}, \vphi_{\Lambda}^{\prime}\}\in \mathscr{Q}(\X_{\Lambda}^{\prime}) : 
\X_{\Lambda}^{\prime} \in {\sf H}_{\varrho}^{\prime},\ 
\Lambda\in  F_{\mathbb{L}} \}$.  Then we have the corresponding  net of the
standard forms:  $ F (O_{\varrho}^{\prime})=\{ F (\M_{\Lambda}^{\prime}, \vphi_{\Lambda}^{\prime}) : \Lambda \in  F_{\mathbb{L}} \}$.
We say that the magnetic systems $ O_{\varrho}^{\prime}$ and $ O_{\varrho}$ are {\it consistent}, if 
\be
   F (\M^{\prime}_{\Lambda}, \vphi^{\prime}_{\Lambda})\LRA F (\M_{\Lambda}, \vphi_{\Lambda}) \ \ \forall \Lambda \in  F_{\mathbb{L}} .
\ee
Hence, for $\Lambda, \Lambda^{\prime}\in  F_{\mathbb{L}} $ with $\Lambda\subset \Lambda^{\prime}$,  the following diagram is commutative:
\be
\begin{tikzcd}
  F(\M^{\prime}_{\Lambda^{\prime}}, \vphi^{\prime}_{\Lambda^{\prime}}) \ar[r] \ar[d] & F(\M^{\prime}_{\Lambda}, \vphi_{\Lambda}^{\prime}) \ar[d]\\
   F(\M_{\Lambda^{\prime}}, \vphi_{\Lambda^{\prime}}) \ar[r] &  F(\M_{\Lambda}, \vphi_{\Lambda})
\end{tikzcd}\label{CommD}
\ee
We write this as $ O_{\varrho}^{\prime} \LRA  O_{\varrho}$.
Trivially, if $ O_{\varrho}^{\prime\prime} \LRA  O_{\varrho}^{\prime}$ and $ O_{\varrho}^{\prime} \LRA O_{\varrho}$, then
$ O_{\varrho}^{\prime\prime} \LRA  O_{\varrho}$ holds.
\end{Def}

\begin{Thm}[Stability of   magnetic orders]\label{StaMO}
Suppose that magnetic systems $O_{\varrho}$ and $O_{\varrho}\rq{}$ are given.
Assume that the magnetic system $O_{\varrho}$ induces a magnetic order with a spin density  of $s$.
If $ O_{\varrho}^{\prime} \LRA  O_{\varrho}$, then the magnetic system $O_{\varrho}^{\prime}$ induces a magnetic order with a spin density of  $s$ as well.
\end{Thm}
\begin{proof}
Apply Theorem \ref{StabThmEx}.
\end{proof}

It turns out that if two different magnetic systems are consistent in the sense of Definition \ref{ConsisOO}, then the two systems have similar magnetic orders. The essence of this phenomenon is the consistent propagation of positivity in the network formed by the IEE systems, which is  represented by the commutative diagram \eqref{CommD}.

\begin{Def}\upshape
Fix $\varrho\in (0, 1)$. Suppose that  a magnetic system $O_{\varrho}$ is given. The {\it stability class associated  with $O_{\varrho}$}  is defined by 
\be
\mathscr{C}(O_{\varrho}):=\{O_{\varrho} ^{\prime }:  \mbox{ a magnetic system satisfying $O_{\varrho}^{\prime} \LRA  O_{\varrho}$}\}.
\ee
\end{Def}

Figure \ref{Fig2} shows a visualization of the concept of  the stability class;
$\mathscr{C}(O_{\varrho})$ can be regarded as a {\it directed tree with the root }$O_{\varrho}$, more specifically,  an  arborescence,  in the following fashion:
 \begin{itemize}
 \item Each magnetic system corresponds to a vertex.
 \item  Each pair of magnetic systems $(O^{\prime}_{\varrho}, O^{\prime\prime}_{\varrho})$ connected by the relation $O^{\prime}_{\varrho} \LRA O^{\prime\prime}_{\varrho}$ corresponds to an edge with the origin $O^{\prime}_{\varrho}$ and terminus $O^{\prime\prime}_{\varrho}$.
 \end{itemize}
 \begin{figure}[ht]
 \begin{center}
 \includegraphics[scale=1.0]{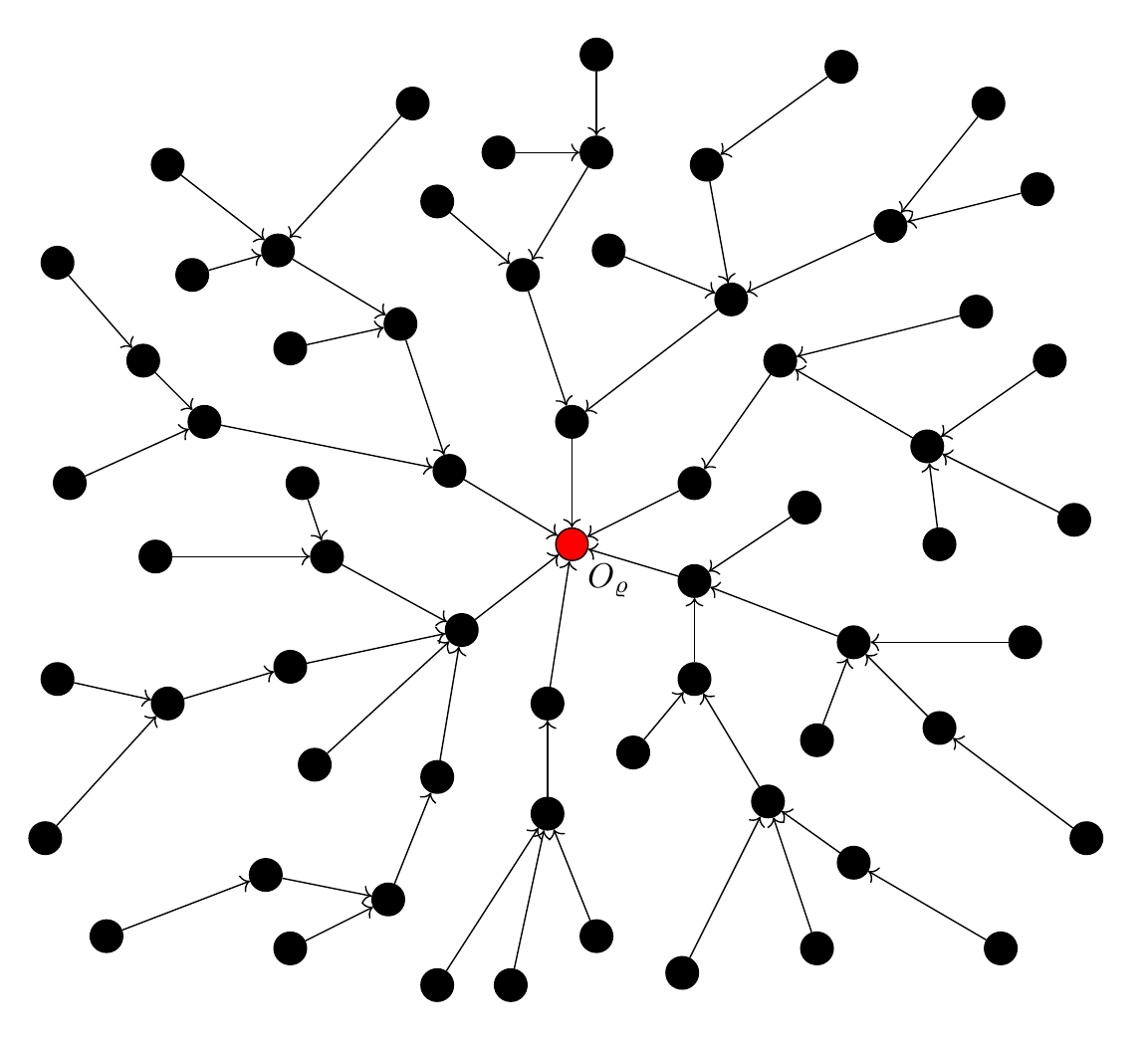}
  \caption{
        $\mathscr{C}(O_{\varrho})$ corresponds to a  directed tree with the root $O_{\varrho}$.
        }\label{Fig2}
        \end{center}
 \end{figure}
Note that the discussion here is almost the same as in Subsection \ref{SectStCl}; see Figure \ref{Fig1}.

We can restate Theorem \ref{StaMO} in terms of the stability class as follows.

\begin{Coro}\label{MOrderC}
If the magnetic system $O_{\varrho}$ induces a magnetic order with a spin density  of $s$, then all magnetic system in $\mathscr{ C}(O_{\varrho})$ induces a magnetic order with a spin density of  $s$ as well.
\end{Coro}

Next, let us mention the equivalence of magnetic systems.
Given  isomorphic graphs $G=(\mathbb{L}, E)$ and $G'=(\mathbb{L}', E')$, consider
nets of IEE spaces
${\sf H}^G_{\varrho}=\{\X_{\Lambda} \in \mathscr{I}(\h_{\Lambda, [\varrho |\Lambda|]}) : \Lambda\in  F_{\mathbb{L}} \}$ and ${\sf H}^{G'}_{\varrho}=\{\X'_{\Lambda'} \in \mathscr{I}(\h'_{\Lambda', [\varrho |\Lambda|]}) : \Lambda'\in  F_{\mathbb{L}'} \}$. Let $\phi$ be the isomorphism  between $G$ and $G'$. Note that  $\phi$ gives rise to the bijection between $F_{\mathbb{L}}$ and $F_{\mathbb{L}'}$.
Suppose that $\{\h_{\Lambda, [\varrho|\Lambda|]}, \h'_{\phi(\Lambda), [\varrho |\Lambda|]}\}$ is  an isomorphic pair\footnote{The term ``isomorphic pair" is defined  in Definition \ref{HilIso}. } for all $\Lambda \in F_{\mathbb{L}}$.
By Lemma \ref{I=I}, we can choose ${\sf H}^G_{\varrho}$ and ${\sf H}^{G\rq{}}_{\varrho}$ to satisfy the following: for each $\Lambda\in F_{\mathbb{L}}$, there exists a unitary operator $U_{\Lambda}: \X_{\Lambda} \to \X_{\phi(\Lambda)}\rq{}$ satisfying 
$
U_{\Lambda} S_x^{(i)} U_{\Lambda}^{-1}=S_{\phi(x)}^{\prime (i)}\ \  (i=1, 2, 3,\  x\in \Lambda),
$
where $S^{\prime (i)}_{x\rq{}}$ stand for the spin  operators  on $\X_{\phi(\Lambda)}\rq{}$.

\begin{Def}\label{MSEqui}\upshape
Suppose that we are in the above setting. Suppose that we are given magnetic systems $O_{\varrho}^G=\{\{\M_{\Lambda}, \vphi_{\Lambda}\} \in \mathscr{Q}(\X_{\Lambda}) : \X_{\Lambda}\in {\sf H}_{\varrho}^G, \  \Lambda\in  F_{\mathbb{L}} \}$ and $O_{\varrho}^{G\rq{}}=\{\{\M_{\Lambda\rq{}}\rq{}, \vphi_{\Lambda\rq{}}\rq{}\} \in \mathscr{Q}(\X_{\Lambda\rq{}}\rq{}) : \X_{\Lambda\rq{}}\rq{}\in {\sf H}_{\varrho}^{G\rq{}}, \  \Lambda\rq{}\in  F_{\mathbb{L}\rq{}} \}$ associated with ${\sf H}^G_{\varrho}$ and ${\sf H}^{G\rq{}}_{\varrho}$, respectively.
If
$\{\M_{\Lambda}, \vphi_{\Lambda}\}\in O_{\varrho}^G$ and $\{\M'_{\phi(\Lambda)}, \vphi'_{\phi(\Lambda)}\}\in O_{\varrho}^{G'}$
 are equivalent in the sense of Definition \ref{IEESEquiI} for all $\Lambda\in F_{\mathbb{L}}$, then $O_{\varrho}^G$ and $O_{\varrho}^{G'}$ are said to be {\it equivalent}.
\end{Def}

\begin{Rem}\upshape
Here, a note on Definition \ref{MSEqui} is given.
Considering Lemma \ref{I=I}, for each $\Lambda\in F_{\mathbb{L}}$, both $\X_{\Lambda}$ and $\X_{\phi(\Lambda)}'$ 
can be regarded as elements of $\mathscr{I}(\h_{\Lambda, [\varrho N]})$. Therefore, the concept of equivalence introduced in Definition \ref{IEESEquiI} can be applied to $\{\M_{\Lambda}, \vphi_{\Lambda}\}$ and $\{\M'_{\phi(\Lambda)}, \vphi'_{\phi(\Lambda)}\}$. Note that in Definition \ref{IEESEquiI}, $\X$ and $\X'$ both belong to $\mathscr{I}(\h_{\Lambda, N})$.

\end{Rem}

From the definition, the following proposition follows immediately.

\begin{Prop}\label{IsoStaCl}
We are in the setting described above.
Suppose that $O_{\varrho}^G$ and $O_{\varrho}^{G'}$ are  equivalent. Then we have the following:
\begin{itemize}
\item[\rm (i)]
$O^G_{\varrho}$ induces a magnetic order with a spin density of  $s$, if and only if,  $O_{\varrho}^{G^{\prime}}$ induces a magnetic order with a spin density of  $s$.
\item[\rm (ii)]

$\mathscr{C}(O^G_{\varrho})$ and $\mathscr{C}(O_{\varrho}^{G'})$ are isomorphic in the following sense:
there exists a bijection, $\upsilon$,  between 
$\mathscr{C}(O^G_{\varrho})$ and $\mathscr{C}(O_{\varrho}^{G'})$ such that 
$O_{\varrho}$ and $\upsilon(O_{\varrho})$ are equivalent for each $O_{\varrho} \in \mathscr{C}(O^G_{\varrho})$.

\end{itemize}
\end{Prop}

\subsection{Magnetic orders in ground states}
In the previous discussion,   Hamiltonians and their ground states did not appear. In this section, we will discuss when  Hamiltonians and their ground states form a magnetic order.
Let ${\sf H}_{\varrho}=\{\X_{\Lambda} \in \mathscr{I}(\h_{\Lambda, [\varrho |\Lambda|]}) : \Lambda\in  P(\mathbb{L}) \}$ be a net of IEE spaces. We are given a magnetic system $O_{\varrho}$  associated with ${\sf H}_{\varrho}$.

Consider a net  formed by  Hamiltonians indexed  by $F_{\mathbb{L}}$: \be
{\bs H}=\{H_{\Lambda} : \mbox{$H_{\Lambda}$ acts in $\mathfrak{X}_{\Lambda}$ for all $\Lambda\in F_{\mathbb{L}}$} \}.\label{HamiNet}
\ee
In general, the graph $G$ and hence $F_{\mathbb{L}}$ are consistently determined by the structure of the net of Hamiltonians, $\bs H$,  under consideration; see Sections \ref{Sect5} and \ref{Sect6} for detail.

We impose the following condition:
\begin{Def}\label{ConsUni}\upshape
Let $\bs H$ be the net of Hamiltonians given by \eqref{HamiNet}.
 We say that  $\bs H$ is {\it adapted to $O_{\varrho}$}, if $H_{\Lambda}\in A_{\Lambda, [\varrho|\Lambda|]}(\M_{\Lambda}, \vphi_{\Lambda})$ with $\{\M_{\Lambda}, \vphi_{\Lambda}\}\in O_{\varrho}$ for all $\Lambda\in F_{\mathbb{L}}$. Here, recall that $ A_{\Lambda, [\varrho|\Lambda|]}(\M_{\Lambda}, \vphi_{\Lambda})$ is given in Definition \ref{DefAHamiC}.

\end{Def}

This definition implies that each Hamiltonian $H_{\Lambda}\in {\bs H}$ is consistent with the magnetic property, or positivity, of the  IEE system $\{\M_{\Lambda}, \vphi_{\Lambda}\}$. 
How such  nets of Hamiltonians  are constructed will be discussed in Appendix \ref{SectA}.
 
 To describe the next result, we introduce a term.
 \begin{Def}\label{DefMagO}\upshape
 Suppose that we are in the setting of Definition \ref{ConsUni}.
 For a given $\Lambda\in F_{\mathbb{L}}$, by Proposition \ref{GSinG}, we can choose a normalized ground state $\psi_{\Lambda}$ of $H_{\Lambda}$  to satisfy $\psi_{\Lambda} \in G(\M_{\Lambda}, \vphi_{\Lambda})$.
 We refer to the net ${\bs \psi}=\{\psi_{\Lambda} : \Lambda\in F_{\mathbb{L}}\}$ that consists of these chosen ground states as a {\it net of magnetic ground states} (NMGS) associated with ${\bs H}$.

For a given number $s\in [0, 1/2]$,  we say that an NMGS ${\bs \psi}$ exhibits a {\it magnetic order} with a  spin density of $s$,   if there exists an increasing sequence  of  sets $\{\Lambda_n : n\in \BbbN\}\subset F_{\mathbb{L}}$ such that 
$\bigcup_{n=1}^{\infty} \Lambda_n=\mathbb{L}$ and  
\be
 S_{H_{\Lambda_n}}=s|\Lambda_n|+o(|\Lambda_n|)\ \ \mbox{as $n\to \infty$},
\ee
where $S_{H_{\Lambda}}$ indicates the total spin of $\psi_{\Lambda}$; see Definition \ref{DefSH}.
The number $s$ is called the {\it spin density} associated with  $ \bs \psi$.
If $s$ is non-zero, then we say that the magnetic order  is {\it strict}.
 \end{Def}
 
 \begin{Prop}
 Suppose that we are give a magnetic system $O_{\varrho}$.
 Suppose that a net of Hamiltonians  $\bs H$ is adapted to the magnetic system $O_{\varrho}$.
  If $O_{\varrho}$ induces a magnetic order with a  spin density of  $s$, then each NMGS associated with $\bs H$ exhibits a magnetic order with a  spin density  of $s$.

 \end{Prop}
 \begin{proof}
 Let ${\bs \psi}=\{\psi_{\Lambda} : \Lambda\in F_{\mathbb{L}}\}$ be an NMGS associated with $\bs H$.
 Because $\bs H$ is adapted to $O_{\varrho}$, $H_{\Lambda} \in A_{\Lambda, [\varrho|\Lambda|]}(\M_{\Lambda}. \vphi_{\Lambda})$ holds for all $\Lambda\in F_{\mathbb{L}}$.
 Hence, by  Corollary \ref{StaGsFixH}, $\psi_{\Lambda}$ has total spin $S(\M_{\Lambda}, \vphi_{\Lambda})$.
 Recalling Definition \ref{DefOrderM},  we see that there exists  an increasing sequence  of  sets $\{\Lambda_n : n\in \BbbN\}\subset F_{\mathbb{L}}$ such that 
$\bigcup_{n=1}^{\infty} \Lambda_n=\mathbb{L}$ and  \eqref{LimitForm} holds.
Putting the above together, we find that $\bs \psi$ exhibits a magnetic order with a  spin density of  $s$.
 \end{proof}
 
 Suppose that we are given another net of IEE spaces   ${\sf H}_{\varrho}^{\prime}=\{\X_{\Lambda}^{\prime} \in \mathscr{I}(\h_{\Lambda, [\varrho |\Lambda|]}) : \Lambda\in  P(\mathbb{L}) \}$ and magnetic system $O_{\varrho}^{\prime}$ associated with  ${\sf H}^{\prime}_{\varrho}$. In addition, suppose that another net of Hamiltonians ${\bs H}'=\{H'_{\Lambda} : \mbox{$H'_{\Lambda}$ acts in $\mathfrak{X}_{\Lambda}$ for all $\Lambda\in F_{\mathbb{L}}$} \}$ is given.


\begin{Thm}\label{OtoO}
Suppose that magnetic systems $O_{\varrho}$ and $O_{\varrho}\rq{}$ are given.
 Suppose that nets of Hamiltonians ${\bs H}$ and ${\bs H}\rq{}$ are adapted to $O_{\varrho} $ and $O_{\varrho}\rq{}$, respectively.
 In addition, suppose that 
 \be
 O_{\varrho}\rq{} \LRA   O_{\varrho}.
 \ee
 Then we have the following:
\begin{itemize}
\item[\rm (i)] $S_{H_{\Lambda}}=S_{H^{\prime}_{\Lambda}}$ for all $\Lambda\in F_{\mathbb{L}}$, where $H_{\Lambda} \in {\bs H}$ and $H^{\prime}_{\Lambda}\in {\bs H}\rq{}$.

\item[\rm (ii)] If  the magnetic system $ O_{\varrho}$ induces a magnetic order with a  spin density of  $s$, then every NMGS  associated with either $\bs H$ or $\bs H'$ exhibits a magnetic order with a   spin density  of $s$ as well.
\end{itemize}
\end{Thm}
\begin{proof}
(i) Use Corollary \ref{StaGsBigA}. (ii) follows from (i) and Corollary \ref{MOrderC}.
\end{proof}

This theorem implies the stability of magnetic orders: let ${\bs H}$ be a given net of  Hamiltonians, and let 
${\bs H}\rq{}$ be another net of  Hamiltonians defined by adding  interaction with the environmental system to ${\bs H}$. If an NMGS associated with $\bs H$ exhibits a  magnetic order, then every NMGS associated with $\bs H\rq{}$ exhibits the similar magnetic order; see Sections \ref{Sect5} and \ref{Sect6} for examples.

By reconsidering Theorem \ref{OtoO} in terms of the stability class, we obtain the following theorem.

\begin{Coro}[Stability of magnetic orders ]\label{StaMagOrder}

Suppose that magnetic systems $O_{\varrho}$ and $O_{\varrho}\rq{}$ are given.
 Suppose that nets of Hamiltonians ${\bs H}$ and ${\bs H}\rq{}$ are adapted to $O_{\varrho} $ and $O_{\varrho}\rq{}$, respectively. 
Assume that $ O^{\prime}_{\varrho}\in \mathscr{C}( O_{\varrho})$. 
 If  the magnetic system $ O_{\varrho}$ induces a magnetic order with a  spin density of  $s$, then
every NMGS associated with either $\bs H$ or $\bs H'$ exhibits a magnetic order with a  spin density  of $s$ as well.
\end{Coro}

Corollary \ref{StaMagOrder} can be interpreted as follows:
suppose that we are given a stability class $\mathscr{ C}(O_{\varrho})$, represented by $O_{\varrho}$.  Consider any magnetic system $O'_{\varrho}$ that belongs to  $\mathscr{ C}(O_{\varrho})$.
The magnetic order produced by a net of   Hamiltonians ${\bs H}'$  adapted  to  $O_{\varrho}'$ has the same properties as the magnetic order produced by $\bs H$  adapted  to the representative system $O_{\varrho}$.
As we will see in the specific applications,  $\bs H$ has  a straightforward structure and is often easy to analyze. Suppose ${\bs H}'$ is a net  of Hamiltonians describing a complex interaction between a many-electron system and its environment. Even though it is difficult to analyze ${\bs H}'$,
the facts described here make it easy to examine the structure of the magnetic order produced by ${\bs H}'$.
 
\subsection{Magnetic structures in  crystal lattices }

We will deal with various models on crystal lattices in specific applications. In this subsection, we will introduce some terms related to crystal lattices for later use and
briefly summarize how the previous theory is applied to models on crystal lattices. 
For a mathematical description of the crystal lattice, see, for example, \cite{Kotani2000}.

Let $G=(\mathbb{L}, E)$ be an oriented, locally finite connected graph.
We assume that $G$ has no multiple edges and loops for simplicity.  
A set of spatial vectors $\Gamma \subset \BbbR^d$ that satisfies the following properties is called a {\it lattice group}:
there exist linearly independent vectors  ${\bs a}_1, {\bs a_2},\dots,  {\bs a}_d$ in  $\Gamma$ satisfying 
\be
\Gamma=\{k_1 {\bs a}_1+k_2{\bs a}_2+\dots +k_d {\bs a}_d : k_1, k_2,\dots,  k_d\in \BbbZ\}.
\ee
Each  vector belonging to $\Gamma$ is called a  {\it lattice vector}  and ${\bs a}_1, {\bs a}_2,\dots,  {\bs a}_d$ are called {\it $\BbbZ$-basis}.
In the present paper, $\Gamma$ always denotes a  lattice group.
  $G$ is called a ($\Gamma$-){\it crystal lattice} if there exists a lattice group $\Gamma$ acting on $G$ freely and its quotient $G_0=(\mathbb{L}_0, E_0) =\Gamma \setminus G $ is
 a finite graph. In other words, $G$ is an abelian covering graph of a finite graph $G_0$ for which the  covering transformation group is $\Gamma$.

A {\it realization} of $G$ into the  space $\BbbR^d$ is a map $  \pi : \mathbb{L}\to \BbbR^d$ satisfying $   \pi({\bs a} x)=\pi(x)+{\bs a}$ and $  \pi(\mathbb{L})+{\bs a}=  \pi(\mathbb{L})\   (x\in \mathbb{L}, \bs a\in \Gamma)$. 
 For $x\in \mathbb{L}$, $  \pi(x)$ is called the realized vertex, and for $e\in E$, the corresponding edge is determined  by the realized vertexes $  \pi(o(e)) $ and $  \pi(t(e))$, where, for a  given oriented edge $e \in E$, the origin and the terminus of $e$ are denoted by $o(e)$ and $t(e)$, respectively.
 Often, $  \pi$ is also referred to as a   periodic realization of $G$ with  periodic lattice group $\Gamma$.
 Two typical examples are depicted in Figure \ref{Crystal lattices}.
 \begin{figure}[h]
\begin{center}
\includegraphics[scale=0.8]{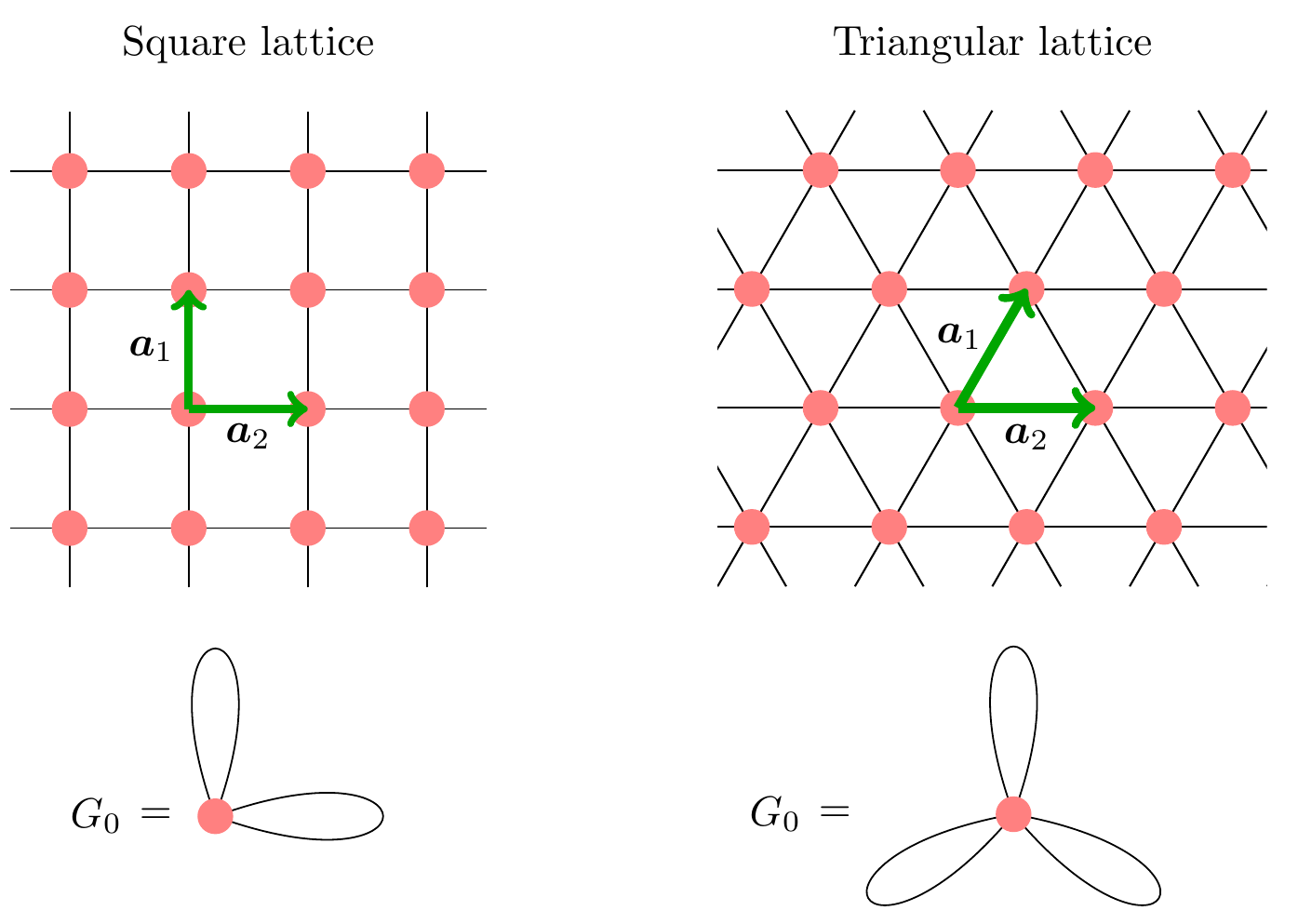}
\caption{Crystal lattices} 
    \label{Crystal lattices}
\end{center}
\end{figure}
In this paper, we always assume that 
 $  \pi$ is {\it non-degenerate}.\footnote{
Namely, $  \pi$ is injective,  and all directed edges belonging to $E_x$ are realized in such a way that they do not overlap, where $E_x=\{e\in E : o(e)=x\}$ with $o(e)$, the origin of $e$. 
}
Given a realization $  \pi$, 
we set $G^{  \pi}=(  \pi(\mathbb{L}), E^{  \pi})$. 
Often, $G^{\pi}$ is also called a realization of $G$ into $\BbbR^d$.

Consider the net of  sublattices $\{(\Lambda, E_{\Lambda}) : \Lambda\in F_{\mathbb{L}}\}$.
Given a realization, $\pi$,  of the crystal lattice $G$ in $\BbbR^d$,
we obtain the corresponding  net of the sublattices realized in $\BbbR^d$: $\{(\pi(\Lambda), E^{\pi}_{\Lambda}) : \Lambda\in F_{\mathbb{L}}\}$.

Given a $\varrho\in (0, 1)$, let ${\sf H}^G_{\varrho}=\{\X_{\Lambda} \in \mathscr{I}(\h_{\Lambda, [\varrho |\Lambda|]}) : \Lambda\in  F_{\mathbb{L}} \}$ be a net of IEE spaces on $G$.
Suppose that we are given a magnetic system $O^G_{\varrho}$ associated with ${\sf H}_{\varrho}^G$. 
By  Lemma \ref{Lemm1}, there exists a Hilbert space $\mathfrak{K}_{\pi(\Lambda), [\varrho|\Lambda|]}\in \mathscr{H}_{\pi(\Lambda), [\varrho|\Lambda|]}$ such that 
$\{ \h_{\Lambda, [\varrho|\Lambda|]}, \mathfrak{K}_{\pi(\Lambda), [\varrho|\Lambda|]}\}$ is an isomorphic pair.
Choose ${\sf H}_{\varrho}^{G^{\pi}}=\{\mathfrak{Y}_{\pi(\Lambda)} \in \mathscr{I}(\frak{K}_{\pi(\Lambda), [\varrho |\Lambda|]}) : \Lambda\in  F_{\mathbb{L}} \}$ to 
satisfy the following: for each $\Lambda\in F_{\mathbb{L}}$, there exists a unitary operator $U_{\Lambda}: \X_{\Lambda} \to \mathfrak{Y}_{\pi(\Lambda)}$ satisfying 
$
U_{\Lambda} S_x^{(i)} U_{\Lambda}^{-1}=S_{\phi(x)}^{\pi,  (i)}\ \  (i=1, 2, 3,\  x\in \Lambda),
$
where $S^{\pi,  (i)}_{y}$ stand for the spin  operators  on $\mathfrak{Y}_{\pi(\Lambda)}$. 
Note that such a choice is always possible from Lemma \ref{I=I}.

\begin{Def}\upshape
Suppose that  we are in the above setting.
Let $O_{\varrho}^{G^{\pi}}$ be a magnetic system associate with ${\sf H}_{\varrho}^{G^{\pi}}$. If $O_{\varrho}^{G^{\pi}}$ is equivalent to $O_{\varrho}^G$ in the sense of Definition \ref{MSEqui}, then $O_{\varrho}^{G^{\pi}}$
 is called a {\it realization of $O_{\varrho}^G$} on the realized crystal lattice $G^{\pi}$.
\end{Def}
It can be shown that there is at least one realization of $O_{\varrho}^G$ on the realized crystal lattice $G^{\pi}$.

 In much of the physics literature,
 we usually consider a magnetic system $O^{G^{\pi}}_{\varrho}$ on a concretely realized crystal lattice $G^{\pi}$ and a net of  Hamiltonians $\bs H^{\pi}$ adapted to  it. Needles to say, Theorem \ref{OtoO} and Corollary \ref{StaMagOrder} hold for $O^{G^{\pi}}_{\varrho}$ and $\bs H^{\pi}$.

By applying Proposition \ref{IsoStaCl}, we immediately obtain the following.
 \begin{Prop}
 For any realization, $\pi$, of the crystal lattie $G$ into $\BbbR^d$, 
 $\mathscr{C}(O^G_{\varrho})$ and $\mathscr{C}(O^{G^{\pi}}_{\varrho})$ are isomorphic  in the sense of Proposition \ref{IsoStaCl}.
 \end{Prop}
 This proposition asserts that the magnetic properties we are investigating do not depend on any concrete realization of the crystal lattice.
 For this reason, we consider magnetic systems on abstract crystal lattices rather than individual realizations.

  Finally, we note that it is necessary to fix a specific realization of the crystal lattice in the analysis of long-range order.

\subsection{On the description of macroscopic systems}
We will conclude this section  with a discussion of macroscopic magnetic systems.
We note at the outset that there is still much to be learned about macroscopic systems and that the descriptions in this section  are only partial results of what is currently known.

Let ${\sf H}_{\varrho}=\{\X_{\Lambda} \in \mathscr{I}(\h_{\Lambda, [\varrho |\Lambda|]}) : \Lambda\in  P(\mathbb{L}) \}$ be a net of IEE spaces. We are given a magnetic system 
$O_{\varrho}=\{\{\M_{\Lambda}, \vphi_{\Lambda}\} \in \mathscr{Q}(\X_{\Lambda}) : \X_{\Lambda}\in {\sf H}_{\varrho}, \  \Lambda\in  F_{\mathbb{L}} \}$
  associated with ${\sf H}_{\varrho}$. Suppose that a net of Hamiltonians ${\bs H}=\{H_{\Lambda} :\mbox{$\Lambda\in F_{\mathbb{L}}$} \}$ is adapted to $O_{\varrho}$.

Take $\Lambda, \Lambda', \Lambda''\in F_{\mathbb{L}}$ with $\Lambda\subset \Lambda'\subset \Lambda''$. 
Then the following chain holds:
\be
F(\M_{\Lambda}, \vphi_{\Lambda})\longleftarrow F(\M_{\Lambda'}, \vphi_{\Lambda'})\longleftarrow F(\M_{\Lambda''}, \vphi_{\Lambda''}).\label{ChainX}
\ee
We denote by $\vartheta_{\Lambda' \Lambda}$ the isometric linear mapping from $\X_{\Lambda}$ into $\X_{\Lambda'}$.
Considering the property \eqref{ChainX}, without loss of generality, we may assume the following:
\be
\vartheta_{\Lambda'' \Lambda'} \vartheta_{\Lambda' \Lambda}=\vartheta_{\Lambda'' \Lambda}. \label{ConsisTheta}
\ee

In the remainder of this subsection, we assume that $O_{\varrho}$ induces a magnetic order with a spin density of  $s$. 
  Hence, there exists  an increasing sequence  of  sets $\{\Lambda_n : n\in \BbbN\}\subset F_{\mathbb{L}}$ such that 
$\bigcup_{n=1}^{\infty} \Lambda_n=\mathbb{L}$ and  \eqref{LimitForm} holds.

By using the property \eqref{ConsisTheta}, we can define an inductive limit \cite{KaRi}: 
$
\ilim \X_{\Lambda_n}.$ 
For simplicity of symbol, we set 
\be
\X_{\mathbb{L}}=\ilim \X_{\Lambda_n}.
\ee
It is natural to think that this Hilbert space is the foundation for describing macroscopic systems.
By \cite{Bratteli1997}, we have the following:
\begin{Lemm}\label{PropIndLim}
For each $n\in \mathbb{N}$, there exists an isometric linear mapping $\vartheta_{\Lambda_n}$ from $\X_{\Lambda_n}$ into 
$\X_{\mathbb{L}}$ satisfying the following:
\begin{itemize}
\item[\rm (i)] If $n \le n\rq{}$, then $\vartheta_{\Lambda_{n'}}\vartheta_{\Lambda_{n\rq{}}\Lambda_n}=\vartheta_{\Lambda_n}$.
\item[\rm (ii)] $\bigcup_{n \in \mathbb{N}} \vartheta_{\Lambda_n}\X_{\Lambda_n}$ is dense in $\X_{\mathbb{L}}$.
\end{itemize}
Furthermore, the Hilbert space $\X_{\mathbb{L}}$ and the net of the isometric linear mappings $ \{\vartheta_{\Lambda_n} : \Lambda\in \mathbb{N}\}$    are uniquely determined,   up to unitary equivalence. 
\end{Lemm}

Let us recall that the magnetic properties of finite systems are expressed using the standard forms. In the following, we will construct  standard forms to describe the magnetic properties of macroscopic systems.
For this purpose, we first define the von Neumann algebra on $\X_{\mathbb{L}}$ by
\be
\M_{\mathbb{L}}=\Bigg(\bigcup_{n \in \BbbN} \M_{\Lambda_n}\Bigg)^{\prime\prime}. \label{MLDef}
\ee
We also define the unitary involution $J_{\mathbb{L}}$ on $\X_{\mathbb{L}}$ by
\be
J_{\mathbb{L}} \vartheta_{\Lambda_n}=\vartheta_{\Lambda_n} J_{\Lambda_n}. \label{JMacFi}
\ee
By Lemma \ref{PropIndLim}, we see that $J_{\mathbb{L}}$ is well defined.
Finally, a  convex cone $\Cone_{\mathbb{L}}$ in $\X_{\mathbb{L}}$ is given by 
\be
\Cone_{\mathbb{L}}=\overline{\bigcup_{n \in \BbbN} L^2(\M_{\Lambda_n}, \vphi_{\Lambda_n})_+},\label{PLDef}
\ee
where the bar stands for the closure in the strong topology.

To describe the results, we introduce the following definition:
\begin{Def}\upshape
Let $\X_1$ and $\X_2$ be complex separable Hilbert spaces with $\X_1\subseteq \X_2$.
Let $\M_1$ and $\M_2$ be von Neumann algebras on $\X_1$ and $\X_2$, respectively.
Suppose that two standard forms $F_1=\{\M_1,  \X_1, \Cone_1, J_1\}$ and $F_2=\{\M_2,  \X_2, \Cone_2, J_2\}$
are given. We say that $F_1$ and $F_2$ are {\it consistent} if the following (i), (ii) and (iii)
 are satisfied:
 \begin{itemize}
 \item[\rm (i)] $\M_1=P\M_2 P$, where $P$ is the orthogonal projection from $\X_2$ to $\X_1$.
 \item[\rm (ii)] $\Cone_1\subseteq \Cone_2$.
 \item[\rm (iii)] $\Cone_1=P\Cone_2$.
 \end{itemize}
 We write this as $F_2 \dasharrow F_1$.
 It should be noted that this definition of consistency is a generalization of Definition \ref{DefConForm}.
 \end{Def}

Based on the above definition, the relationship between the magnetic properties of the macroscopic system and the magnetic properties of the finite systems can be expressed as follows.

\begin{Thm}\label{MacroFinite}
The quadruple $\{\M_{\mathbb{L}}, \X_{\mathbb{L}},  \Cone_{\mathbb{L}}, J_{\mathbb{L}}\}$ is a standard form of $\M_{\mathbb{L}}$. We denote by $\ilim F(\M_{\Lambda_n}, \vphi_{\Lambda_n})$ this standard form.
For each $n, n\rq{}\in \BbbN$ with $n\le n\rq{}$, the following chain holds:
\be
F(\M_{\Lambda_n}, \vphi_{\Lambda_n}) \longleftarrow F(\M_{\Lambda_{n\rq{}}}, \vphi_{\Lambda_{n\rq{}}})
\dashleftarrow  \ilim F(\M_{\Lambda_n}, \vphi_{\Lambda_n}).
\ee
\end{Thm}
\begin{proof}
First, we will show that $\Cone_{\mathbb{L}}$ is a self-dual cone.
Let $\Cone^{\dagger}_{\mathbb{L}}$ be the dual cone of $\Cone_{\mathbb{L}}$: $\Cone_{\mathbb{L}}^{\dagger}=\{\eta\in \X_{\mathbb{L}} : \la \xi|\eta\ra\ge 0\, \forall \xi\in \Cone_{\mathbb{L}}\}$.
We readily confirm that $\Cone_{\mathbb{L}} \subseteq \Cone_{\mathbb{L}}^{\dagger}$. 
Let us prove the inverse inclusion relation. 
In what follows, we identify $\X_{\Lambda_n}$ with $\theta_{\Lambda_n}\X_{\Lambda_n}$.
Hence, $\X_{\Lambda_n}$ can be regarded as a closed subspace of $\X_{\mathbb{L}}$. In addition,
by (ii) of Lemma \ref{PropIndLim}, $\bigcup_n \X_{\Lambda_n}$ is dense in $\X_{\mathbb{L}}$.
For simplicity of presentation, let us set $\Cone_{\Lambda}=L^2(\M_{\Lambda}, \vphi_{\Lambda})_+$.
 Take $\eta\in \Cone_{\mathbb{L}}^{\dagger}$ arbitrarily. There exists a sequence $\{\eta_\ell\in \bigcup_n \X_{\Lambda_n} : \ell \in \BbbN\}$ such that $\eta_{\ell} \to \eta$ as $\ell\to \infty$.
For each $\ell\in \BbbN$, we can choose $n(\ell) \in\BbbN$ such that $\eta_{\ell}\in \X_{\Lambda_{n(\ell)}}$. Because $\Cone_{\Lambda_{n(\ell)}}$ is a self-dual cone in $\X_{\Lambda_{n(\ell)}}$, we can decompose 
$\eta_{\ell}$ as $\eta_{\ell}=\eta^{\rm R}_{\ell, +}-\eta^{\rm  R}_{\ell, -}+\im (\eta^{\rm I}_{\ell, +}-\eta^{\rm I}_{\ell, -})$, where $\eta^{\rm R}_{\ell, \pm}$ and $\eta^{\rm I}_{\ell, \pm}$ satisfy
$\eta^{\rm R}_{\ell, \pm}, \eta^{\rm I}_{\ell, \pm}\in \Cone_{\Lambda_{n(\ell)}}$ and 
$\la\eta^{\rm R}_{\ell, +}|\eta^{\rm R}_{\ell, -}\ra=0=\la\eta^{\rm I}_{\ell, +}|\eta^{\rm I}_{\ell, -}\ra $.
Using the fact $\Cone_{\Lambda_m} \subseteq \Cone_{\Lambda_n}$ provided that $m<n$, we find that 
$\{\eta^{\rm R}_{\ell, \pm} : \ell \in \BbbN\}$ and $\{\eta^{\rm I}_{\ell, \pm} : \ell\in \BbbN\}$
 are convergent sequences. Now we set 
 $\eta^{\rm R}_{\pm}=\lim_{\ell\to \infty}\eta^{\rm R}_{\ell, \pm}$ and $\eta^{\rm I}_{\pm}=\lim_{\ell\to \infty}\eta^{\rm I}_{\ell, \pm}$. 
We claim that $
 \eta^{\rm R}_{-}=\eta^{\rm I}_{+}=\eta^{\rm I}_{-}=0
$. To see this, assume first that $\eta^{\rm I}_{-}\neq 0$. Trivially, $\eta^{\rm I}_{-} \in \Cone_{\mathbb{L}}$, which implies that $\la\eta|\eta^{\rm I}_{ -}\ra\ge 0$.
On the other hand, we observe that 
\be
\la \eta^{\rm I}_{-}|\eta\ra=\lim_{\ell\to \infty} \la\eta^{\rm I}_{\ell, -}|\eta_{\ell}\ra
=\la \eta^{\rm I}_{-}| \eta^{\rm R}_{+}-\eta^{\rm R}_{-}\ra+(-\im) \|\eta^{\rm I}_-\|^2,
\ee
which  is a contradiction.  By  using similar arguments, we can also show that  $ \eta^{\rm R}_{-}=\eta^{\rm I}_{+}=0$. Hence, $\eta=\eta^{\rm R}_{+} \in \Cone_{\mathbb{L}}$, which implies that $\Cone_{\mathrm{L}}^{\dagger} \subseteq \Cone_{\mathbb{L}}$.

Because $\Cone_{\mathbb{L}}$ is a self-dual cone, we readily confirm that the quadruple $\{\M_{\mathbb{L}}, \X_{\mathbb{L}},  \Cone_{\mathbb{L}}, J_{\mathbb{L}}\}$ is a standard form of $\M_{\mathbb{L}}$.

The property $F(\M_{\Lambda_n}, \vphi_{\Lambda_n}) \longleftarrow F(\M_{\Lambda_{n\rq{}}}, \vphi_{\Lambda_{n\rq{}}})$ follows from \eqref{ChainX}.
Let $P_{\Lambda_{n'}\mathbb{L}}$ be the orthogonal projection from $\X_{\mathbb{L}}$ to $\X_{\Lambda_{n'}}$.
By the construction of $\M_{\mathbb{L}}$, we have $\M_{\Lambda_{n'}}=P_{\Lambda_{n'}\mathbb{L}} \M_{\mathbb{L}} P_{\Lambda_{n'}\mathbb{L}}$. Furthermore, we have $\Cone_{\Lambda_{n'}} \subset \Cone_{\mathbb{L}}$ and $\Cone_{\Lambda_{n'}} =P_{\Lambda_{n'}\mathbb{L}}\Cone_{\mathbb{L}}$.
Hence, we conclude that $F(\M_{\Lambda_{n\rq{}}}, \vphi_{\Lambda_{n\rq{}}})
\dashleftarrow  \ilim F(\M_{\Lambda_n}, \vphi_{\Lambda_n})$.
\end{proof}
The standard form $\ilim F(\M_{\Lambda_n}, \vphi_{\Lambda_n})$ is thought to describe the magnetic properties of the macroscopic system. Theorem \ref{MacroFinite} implies that the magnetic properties of the macroscopic system and those of the finite systems are consistently connected.

Given a $\Lambda\in F_{\mathbb{L}}$, let $\psi_{\Lambda}$ be the normalized ground state of $H_{\Lambda, M=0}$. 
By Lemma \ref{PropIndLim}, we can regard $\psi_{\Lambda_n}$ as a unit vector in $\X_{\mathbb{L}}$.
Hence, for each $n\in \BbbN$, we  can  define the state $\zeta_{\Lambda_n}$ on $\M_{\mathbb{L}}$ by $\zeta_{\Lambda_n}(x):=\la \psi_{\Lambda_n}| x \psi_{\Lambda_n}\ra\, (x\in \M_{\mathbb{L}})$.
Because the unit
ball of the dual of $\M_{\mathbb{L}}$ is compact in the weak-$*$ topology \cite[Theorem IV.21]{Reed1981}, there
is a convergent subsequence of $\{\zeta_{\Lambda_n} : n\in \BbbN\}$. The state $\zeta_{\mathbb{L}}$ is then defined as the weak-$*$
limit of a convergent subsequence and can be regard as  a ground state of  the infinite system.\footnote{
For mathematical details of infinite fermion systems, see, e.g.,  \cite{ARAKI2003,Matsui1996}. This paper considers more general systems than those treated in \cite{ARAKI2003,Matsui1996}, and a rigorous study of the ground-state properties of such systems has yet to be explored.
}
In what follows, we denote by $\{\zeta_{\Lambda_{n_k}} : n_k\in \BbbN\}$ an arbitrarily fixed subsequence   of $\{\zeta_n : n\in \BbbN\}$ defining $\zeta_{\mathbb{L}}$.

\begin{Prop}\label{Inherit}
We have the following:
\begin{itemize}
\item[\rm (i)] $\displaystyle \lim_{n_k\to \infty} \zeta_{\Lambda_{n_k}}(N_{\Lambda_{n_k}})/|\Lambda_{n_k}|=\varrho$, where
$N_{\Lambda}$ is the electron number operator: $N_{\Lambda}=\sum_{x\in \Lambda} n_x$. 
\item[\rm (ii)] $\displaystyle 
\lim_{n_k\to \infty }\zeta_{\Lambda_{n_k}}({\bs S}_{\Lambda_{n_k}}^2)/|\Lambda_{n_k}|^2=s^2
$.
\item[\rm (iii)] For all $x\in \M_{\mathbb{L}}$, $\zeta_{\mathbb{L}}(x J_{\mathbb{L}} xJ_{\mathbb{L}}) \ge 0$ holds.
\end{itemize}
\end{Prop}
\begin{proof}

The property (i) is a direct consequence of the setting that $\X_{\Lambda} \in \mathscr{I}(\h_{\Lambda, [\varrho |\Lambda|]})$ for all $ \Lambda\in  P(\mathbb{L})$. 
The property (ii) follows from Definition \ref{DefOrderM} and the fact that ${\bs S}^2_{\Lambda_n} \psi_{\Lambda_n}=S_{\Lambda_n}(S_{\Lambda_n}+1) \psi_{\Lambda_n}$ with $S_{\Lambda_n}=S(\M_{\Lambda_n}, \vphi_{\Lambda_n})$.

(iii) For any $x\in \bigcup_n \M_{\Lambda_n}$, we find that 
$\zeta_{\Lambda_{n_k}}(x_{n_k} J_{\Lambda_{n_k}} x J_{\Lambda_{n_k}}) \ge 0$ for sufficiently large $k$, where
$\{n_k\}_k$ is the subsequence to be chosen when defining $\zeta_{\mathbb{L}}$.
Taking the limit $k\to \infty$ and applying \eqref{JMacFi}, we obtain the property (iii).
\end{proof}

\begin{Rem}\upshape
\begin{itemize}
\item
Let us recall that the following  holds for finite systems:
 $\zeta_{\Lambda_n}(x J_{\Lambda_n} xJ_{\Lambda_n}) \ge 0\, (x\in \M_{\Lambda_n})$.
As we have repeatedly emphasized, the magnetic properties of the ground states are characterized by positivity. Therefore, the claim  (iii) of Proposition \ref{Inherit} asserts that the positivity describing the magnetic properties propagates from the finite systems to  the  macroscopic system consistently.

\item
Note that we also allow for the $s = 0$ cases in our definition of the magnetic orders. This is why we choose the ground states of the Hamiltonians restricted to the $M = 0$ subspace. The magnetic orders in the $s = 0$ cases may not be measurable as the long-range orders. However, as can be seen from the previous discussion, the ground states have positivity which characterizes the magnetic properties, and this justifies the use of the term  \lq\lq{}order\rq\rq{} even in the $s = 0$ cases.

\end{itemize}
\end{Rem}

Let ${\sf H}^{\prime}_{\varrho}=\{\X^{\prime}_{\Lambda} \in \mathscr{I}(\h_{\Lambda, [\varrho |\Lambda|]}) : \Lambda\in  P(\mathbb{L}) \}$ be another  net of IEE spaces. Assume that we are given a magnetic system 
$O^{\prime}_{\varrho}=\{\{\M^{\prime}_{\Lambda}, \vphi^{\prime}_{\Lambda}\} \in \mathscr{Q}(\X^{\prime}_{\Lambda}) : \X^{\prime}_{\Lambda}\in {\sf H}^{\prime}_{\varrho}, \  \Lambda\in  F_{\mathbb{L}} \}$
  associated with ${\sf H}^{\prime}_{\varrho}$.
  \begin{Thm}\label{MacroDiagram}
  Suppose that $O^{\prime}_{\varrho} \LRA O_{\varrho}$. Then
  we have $
  \ilim F(\M^{\prime}_{\Lambda_n}, \vphi^{\prime}_{\Lambda_n}) \dasharrow\ilim F(\M_{\Lambda_n}, \vphi_{\Lambda_n})
  $. Hence, if $m<n$, then the following diagram is commutative:
  \be
\begin{tikzcd}
F(\M_{\Lambda_m}^{\prime}, \vphi_{\Lambda_m}^{\prime}) \arrow[d]  &\arrow[l] F(\M_{\Lambda_{n}}^{\prime}, \vphi_{\Lambda_{n}}^{\prime}) \arrow[d]  &\arrow[l, dashed]    \ilim F(\M^{\prime}_{\Lambda_n}, \vphi^{\prime}_{\Lambda_n}) \arrow[d,dashed]\\
  F(\M_{\Lambda_m}, \vphi_{\Lambda_m})  &\arrow[l]F(\M_{\Lambda_{n}}, \vphi_{\Lambda_{n}})  &\arrow[l,dashed]  \ilim F(\M_{\Lambda_n}, \vphi_{\Lambda_n})
\end{tikzcd}\label{MacDiagComm}
\ee
  \end{Thm}
  \begin{proof}
  Because $O^{\prime}_{\varrho} \LRA O_{\varrho}$, we have the following:
  \begin{itemize}
  \item $\X_{\Lambda_n}$ is a closed subspace of $\X'_{\Lambda_n}$.
  \item $\M_{\Lambda_n}=P_n \M'_{\Lambda_n} P_n$, where $P_n$ is the orthogonal projection from 
  $\X'_{\Lambda_n}$ to $\X_{\Lambda_n}$.
  \item $\Cone_{\Lambda_n} \subseteq \Cone'_{\Lambda_n}$ and $\Cone_{\Lambda_n}=P_n\Cone'_{\Lambda_n}$, where $\Cone_{\Lambda}=L^2(\M_{\Lambda}, \vphi_{\Lambda})_+$ and 
  $\Cone_{\Lambda}^{\prime}=L^2(\M_{\Lambda}^{\prime}, \vphi_{\Lambda}^{\prime})_+.$
  \end{itemize}
  Combining these properties with the definition of the inductive limit of Hilbert spaces, we obtain the following:
  \begin{itemize}
  \item $\M_{\mathbb{L}}=P_{\mathbb{L}} \M'_{\mathbb{L}} P_{\mathbb{L}}$, where $P_{\mathbb{L}}$ is the orthogonal projection from $\ilim \X'_{\Lambda_n}$ to $\ilim \X_{\Lambda_n}$.
  \item $\Cone_{\mathbb{L}} \subseteq \Cone'_{\mathbb{L}}$ and $\Cone_{\mathbb{L}}=P_{\mathbb{L}} \Cone'_{\mathbb{L}}$.
  \end{itemize}
  Here, $\M_{\mathbb{L}}$ and $\M'_{\mathbb{L}}$ are defined through  \eqref{MLDef}, and $\Cone_{\mathbb{L}}$ and $\Cone'_{\mathbb{L}}$ are defined through  \eqref{PLDef}.
  Hence, we conclude that $
  \ilim F(\M^{\prime}_{\Lambda_n}, \vphi^{\prime}_{\Lambda_n}) \dasharrow\ilim F(\M_{\Lambda_n}, \vphi_{\Lambda_n})
  $.
  \end{proof}
  
  In this way, macroscopic systems can be naturally included in the picture of the network of positivity propagation that characterizes magnetic orders.

\section{The Marshall--Lieb--Mattis stability classes}\label{Sect5}

\subsection{Overview}
In this section, we introduce the Marshall--Lieb--Mattis (MLM) stability class. A set of theorems on magnetic stability derived from this class is a non-trivial extension of the MLM theorem for the Heisenberg model \cite{Marshall1955,Lieb1962} and Lieb's theorem for the Hubbard model \cite{Lieb1989}. It also gives a unified understanding of the magnetic properties of the ground states of a comprehensive class of models, including the Holstein--Hubbard model and the Kondo lattice model.

To better understand the results of this section, we will now overview the history surrounding the MLM theorem.
In early studies of the antiferromagnetic Heisenberg model, Marshall's finding on the ground state \cite{Marshall1955} is still essential in today's studies of this model.
Mattis and Lieb formulated this result in \cite{Lieb1962} in the best form known today: 
the ground state of the antiferromagnetic Heisenberg model on a connected bipartite lattice is unique\footnote{To be precise, the ground state is unique apart from the trivial $(2S+1)$-fold degeneracy.} and has total spin $S=\big||\Lambda_A|-|\Lambda_B|\big|/2$.
This result is now widely known as the MLM theorem.
In 1989, Lieb discovered influential theorems about the Hubbard model \cite{Lieb1989}. One of his discoveries can be summarized as follows: 
in the half-filled repulsive Hubbard model on a bipartite connected lattice, the ground state is unique\footnote{The same footnote as for the MLM theorem applies.} and  has total spin
 $S=\big||\Lambda_A|-|\Lambda_B|\big|/2$.
It is not a coincidence that this theorem and the MLM theorem are similar. Understanding the deeper reasons for this is one of the goals of this section.
A fundamental idea in the proof of Lieb's theorem is to apply the idea of reflection positivity, used in   axiomatic quantum  field  theory \cite{Osterwalder1973,Osterwalder1975}, to spin space. Today, this method  is called  {\it spin reflection positivity} (SRP).
Lieb's SRP has since been applied to various models of many-electron systems with great success. 
First, an application of spin reflection positivity to the Kondo lattice model was given by Yanagisawa and Shimoi \cite{Yanagisawa1995}. Later, this result was improved by Tsunetsugu \cite{Tsunetsugu1997}.
Then, Ueda {\it et al.} gave an adaptation of  SRP to the periodic Anderson model \cite{Ueda1992}.
Freericks and Lieb applied  SPR to systems with electron-phonon  interactions \cite{Freericks1995}. 
Later, the author reformulated  SPR using order-preserving operator inequalities, which enabled the analysis of a broader class of models, including many-electron systems interacting with phonons and quantized radiation fields \cite{Miyao2019}.

As mentioned above, Lieb's SRP has a wide range of applications.\footnote{For more examples that could not be given here, see \cite{Shen1998, Tian2004} and the references therein. } The purpose of this section is to provide a mathematical framework that can describe these many results in a unified way.
This goal will be achieved using the theory constructed in Sections \ref{Sect3} and \ref{Sect4}.

\subsection{General results}\label{GeneMLM}

Let $G=(\mathbb{L}, E)$ be an infinite  connected bipartite graph. We  denote the bipartite structure of $G$ by 
$\mathbb{L}=\mathbb{L}_A\sqcup_E \mathbb{L}_B$; this means that  $\mathbb{L}$ can be divided into two disjoint and independent sets $\mathbb{L}_A$ and $\mathbb{L}_B$ such that every edge connects a vertex in $\mathbb{L}_A$ to one in $\mathbb{L}_B$. 
Recalling the definition of $F_{\mathbb{L}}$, i.e., \eqref{DefF(L)}, we see that, for each $\Lambda\in F_{\mathbb{L}}$, 
the subgraph $G_{\Lambda}=(\Lambda, E_{\Lambda})$ of $G$  satisfies the following properties:
\begin{itemize}
\item $G_{\Lambda}$ is an induced subgraph of $G$, which is connected.
\item If we denote the bipartite structure of $G_{\Lambda}$ as $\Lambda=\Lambda_A\sqcup_{E_{\Lambda}} \Lambda_B$, then $\Lambda_A=\mathbb{L}_A \cap \Lambda$ and $\Lambda_B=\mathbb{L}_B\cap \Lambda$.
\end{itemize}

For later convenience, we set
\be
F_{\mathbb{L}}^{\rm (e)}=\{\Lambda\in F_{\mathbb{L}} : \mbox{$|\Lambda|$ is an even number}\},
\ee
where $|\Lambda|$ stands for the cardinality of $\Lambda$.
Given a $\Lambda\in P(\mathbb{L})$, define
\be
\h_{\Lambda}=Q_{\Lambda} \F_{\Lambda, N=|\Lambda|}, \label{HSingE}
\ee
where $Q_{\Lambda}=\prod_{x\in \Lambda} (n_{x\up}-n_{x\down})^2$ and $\F_{\Lambda, N}$ is defined by \eqref{AntiSymm}. Because $Q_{\Lambda}$ commutes with $S_x^{(i)}\, (i=1, 2, 3)$, $\h_{\Lambda}$ is reducible.
The Hilbert space $\h_{\Lambda}$ represents a  half-filling system with one electron at each site.
In this subsection, we regard $\h_{\Lambda}$ itself as an IEE space associated with $\h_{\Lambda}$.
Let  ${\sf H}_{1/2}^{\rm MLM}$ be the net of IEE spaces given by $
{\sf H}_{1/2}^{\rm MLM}=\{\h_{\Lambda} : \Lambda\in F^{\rm (e)}_{\mathbb{L}}\}
$.  In the following, we introduce a crucial magnetic system associated with ${\sf H}_{1/2}^{\rm MLM}$; see Definition \ref{DefMLMcl}.

Define
\be
|\omega\ra_{\Lambda}=(-1)^{|\Lambda_B|}\Bigg[\prod'_{x\in \Lambda} c_{x\down}^*\Bigg]|\varnothing\ra_{\Lambda},
\ee
where $|\varnothing\ra_{\Lambda}$ is the fermionic Fock vacuum in  $\F_{\Lambda}:=\bigoplus_N \F_{\Lambda, N}$, and $\prod'_{x\in \Lambda}$ indicates the product taken over all vertex in $\Lambda$ with an arbitrarily fixed order.
For  given $X, Y\subseteq \Lambda$, define the vector
\be
|X, Y\ra_{\Lambda}=(-1)^{|\overline{Y}\cap \Lambda_B|}\Bigg[\prod'_{x\in X}c_{x\up}^* \Bigg] \Bigg[\prod'_{y\in \overline{Y}}c_{y\down}\Bigg]|\omega\ra_{\Lambda}, \label{DefCONS}
\ee
where $\overline{Y}=\Lambda\setminus Y$.
The vector $|X, Y\ra_{\Lambda}$ represents a state in which the regions $X$ and $Y$ are occupied by electrons with up-spin and down-spin, respectively. This definition may seem a little odd, but it is convenient because it is consistent with the hole-particle transformation we will introduce later.
Then $\{|X, \overline{X}\ra_{\Lambda} : X\subseteq \Lambda\}$ is a complete orthonormal system (CONS) of $\h_{\Lambda}$.
Note that, for any $X\subseteq \Lambda$, 
the vector $|X, \overline{X}\ra_{\Lambda}$  represents a state in which each vertex in  $\Lambda$ is occupied by a single electron; in addition, if $|\Lambda|$ is {\it even}, then by choosing $X\subseteq \Lambda$ such that  $|X|=|\Lambda|/2$,  we see that  $|X, \overline{X}\ra_{\Lambda}$ belongs to the  $M=0$-subspace.
This fact is essential when using {\it spin reflection positivity}, explaining why we need to introduce  $F^{\rm (e)}_{\mathbb{L}}$.

Let $\M_{\Lambda}^{\rm MLM}$ be the abelian von Neumann algebra on $\h_{\Lambda}$ generated by 
diagonal operators associated with $\{|X, \overline{X}\ra_{\Lambda} : X\subseteq \Lambda\}$.\footnote{
Given a CONS $\{x_n : n\in \BbbN\}$ in a Hilbert space $\h$, we say that a linear operator $A$ on $\h$ is {\it diagonal} with respect to $\{x_n : n\in \BbbN\}$, if $A$ can be expressed as $
A=\sum_{n=1}^{\infty} a_n |x_n\ra\la x_n|. 
$}
Note that every element in $\M_{\Lambda}^{\rm MLM}$ commutes with $S_{\Lambda}^{(3)}$, that is, $\M_{\Lambda}^{\rm MLM}$ is a von Neumann subalgebra of $\A_{\Lambda}(\h_{\Lambda})$.
Here, recall that $\A_{\Lambda}(\h_{\Lambda})$ is defined by \eqref{DefAX}.
Define the vector $\xi_{\Lambda}\in \h_{\Lambda}$ by 
\be
\xi_{\Lambda}=\sum_{X\subseteq \Lambda}  |X, \overline{X}\ra_{\Lambda}. \label{TaneVec}
\ee
Corresponding to $\xi_{\Lambda}$, we set $\vphi^{\rm MLM}_{\Lambda}=\vphi_{\xi_{\Lambda}}$, where, for a given vector $\eta$, we set $\vphi_{\eta}(a)=\la \eta|a\eta\ra$.
Trivially, $\vphi^{\rm MLM}_{\Lambda}$ is a faithful semi-finite normal weight on $\M_{\Lambda}^{\rm MLM}$. 
Hence, the pair $\{\M_{\Lambda}^{\rm MLM}, \vphi^{\rm MLM}_{\Lambda}\}$ is an IEE system associated with $\h_{\Lambda}$.
The corresponding standard form $F(\M_{\Lambda}^{\rm MLM}, \vphi^{\rm MLM}_{\Lambda})$ satisfies the following:
\begin{itemize}
\item $\h_{\Lambda}=L^2(\M_{\Lambda}^{\rm MLM}, \vphi_{\Lambda}^{\rm MLM})$.
\item $\psi\in L^2(\M_{\Lambda}^{\rm MLM}, \vphi^{\rm MLM}_{\Lambda})_+$, if  and only if,
$\psi_X\ge 0$ for all $X\subseteq \Lambda$, where $\psi_X= {}_{\Lambda}\la X, \overline{X}|\psi\ra_{\Lambda}$.
\item For each $\psi=\sum_{X\subseteq \Lambda} \psi_X |X, \overline{X}\ra_{\Lambda}$, the action of the modular conjugation $J_{\Lambda}$
 is given by $J_{\Lambda} \psi=\sum_{X\subseteq \Lambda} \psi_X^* |X, \overline{X}\ra_{\Lambda}$, where, for a given $z\in \BbbC$, $z^*$ stands for the complex conjugate of $z$.
\end{itemize}

\begin{Lemm}\label{VolSys}
If $\Lambda, \Lambda'\in F^{\rm (e)}_{\mathbb{L}}$ satisfies $\Lambda\subseteq \Lambda'$, then 
$F(\M_{\Lambda\rq{}}^{\rm MLM}, \vphi^{\rm MLM}_{\Lambda\rq{}})\LRA F(\M_{\Lambda}^{\rm MLM}, \vphi^{\rm MLM}_{\Lambda})$ holds. Hence, the net $O_{1/2}^{\rm MLM}=\big\{\{\M_{\Lambda}^{\rm MLM}, \vphi^{\rm MLM}_{\Lambda}\} : \Lambda\in F^{\rm (e)}_{\mathbb{L}}
\big\}$ is a magnetic system.
\end{Lemm}
\begin{proof}
First, let us recall a natural identification  of the fermionic Fock spaces: given Hilbert spaces $\X$ and $\frak{Y}$, $\F_{\rm F}(\X\oplus \frak{Y})=\F_{\rm F}(\X)\otimes \F_{\rm F}(\frak{Y})$ holds, where
$\F_{\rm F}(\X)$ stands for the fermionic Fock space over $\X$: $\F_{\rm F}(\X)=\bigoplus_n \bigwedge^n \X$. To be precise, there exists a unitary operator $\tau : \F_{\rm F}(\X\oplus \frak{Y})\to \F_{\rm F}(\X)\otimes \F_{\rm F}(\frak{Y})$ satisfying 
$
\tau|\varnothing \ra_{\X\oplus \frak{Y}}=|\varnothing \ra_{\X} \otimes |\varnothing\ra_{\frak{Y}}
$ and
\be
\tau c_{\X\oplus \frak{Y}}(f\oplus g) \tau^{-1}=c_{\X}(f)\otimes 1+(-1)^{N_{\X}} \otimes c_{\frak{Y}}(g)\ \ (f\in \X, g\in \frak{Y}). \label{IdnAnni}
\ee
Here, $|\varnothing \ra_{\X}$,  $c_{\X}(f)$ and $N_{\X}$ stand for the Fock vacuum,  an  annihilation operator,  and the number operator  in $\F_{\rm F}(\X)$, respectively. Using this fact, we find that $\h_{\Lambda} \otimes \h_{\Lambda'\setminus \Lambda}$ is a subspace of $\h_{\Lambda'}$, provided that $\Lambda\subset \Lambda'$. Indeed, by the above identification, we get $\F_{\Lambda\rq{}}=\F_{\Lambda}\otimes \F_{\Lambda\rq{}\setminus \Lambda}$, where $\F_{\Lambda}$ is the fermionic Fock space over $\ell^2(\Lambda) \oplus \ell^2(\Lambda)$. 
By comparing the $|\Lambda\rq{}|$-electron subspaces on both sides of this equation, we can see that 
$\F_{\Lambda\rq{}, |\Lambda\rq{}|}=\bigoplus_{m+n=|\Lambda\rq{}|} \F_{\Lambda, m} \otimes \F_{\Lambda\rq{}\setminus \Lambda, n}$. Hence, $\F_{\Lambda, |\Lambda|} \otimes \F_{\Lambda\rq{}\setminus \Lambda, |\Lambda\rq{}\setminus \Lambda|}$ is a subspace of $\F_{\Lambda\rq{}, |\Lambda\rq{}|}$.  Because $Q_{\Lambda\rq{}}=Q_{\Lambda} \otimes Q_{\Lambda\rq{}\setminus \Lambda}$
 due to \eqref{IdnAnni}, we conclude the claim.

Set $\tilde{\xi}_{\Lambda}=2^{-|\Lambda|/2} \xi_{\Lambda}$.
Define the isometric linear mapping $\kappa : \h_{\Lambda} \to \h_{\Lambda'}$ by 
$
\kappa \eta=\eta\otimes \tilde{\xi}_{\Lambda'\setminus \Lambda}\, (\eta\in \h_{\Lambda})
$.
Hence, $\h_{\Lambda}$ can be regarded as a subspace of $\h_{\Lambda'}$ by identifying the image $\kappa \h_{\Lambda}$ with $\h_{\Lambda}$.
Furthemore, we can identify $\M^{\rm MLM}_{\Lambda}$ with $\M^{\rm MLM}_{\Lambda} \otimes \Pi_{\Lambda'\setminus \Lambda}$, where $\Pi_{\Lambda' \setminus \Lambda}=|\tilde{\xi}_{\Lambda'\setminus \Lambda}\ra\la \tilde{\xi}_{ \Lambda'\setminus \Lambda} |$.

We denote by $P_{\Lambda\Lambda'}$ the orthogonal projection from $\h_{\Lambda'}$ to $\h_{\Lambda}$.
 Note that $P_{\Lambda\Lambda'}$ can be written as 
\be
P_{\Lambda\Lambda'}=1\otimes \Pi_{\Lambda'\setminus \Lambda}.
\ee
Using this formula, we readily confirm the following:
\begin{itemize}
\item $\M_{\Lambda}^{\rm MLM}=\M_{\Lambda}^{\rm MLM}\otimes \Pi_{\Lambda'\setminus \Lambda}=P_{\Lambda\Lambda'}\M_{\Lambda'}^{\rm MLM}P_{\Lambda\Lambda'}$.
\item $\vphi_{\Lambda'}^{\rm MLM}\circ \mathscr{E}_{\Lambda\Lambda'}=\vphi_{\Lambda'}^{\rm MLM}$,
where $
\mathscr{E}_{\Lambda\Lambda'}(x)=P_{\Lambda\Lambda'}xP_{\Lambda\Lambda'}\ (x\in \M_{\Lambda\rq{}}^{\rm MLM})
$.
\item $P_{\Lambda\Lambda'} L^2(\M_{\Lambda'}^{\rm MLM}, \vphi^{\rm MLM}_{\Lambda'})_+=L^2(\M_{\Lambda}^{\rm MLM}, \vphi^{\rm MLM}_{\Lambda})_+$.
\end{itemize}
Therefore, we conclude the desired assertion in the lemma.
\end{proof}

Lemma \ref{VolSys} allows for the following definition.
\begin{Def}\label{DefMLMcl}
\upshape
We call the magnetic system $O^{\rm MLM}_{1/2}$ in Lemma \ref{VolSys} the {\it  Marshall--Lieb--Mattis(MLM) system}. The stability class, $\mathscr{C}(O^{\rm MLM}_{1/2})$,  
 is called the {\it MLM stability class}.
\end{Def}

The various theorems in this section are derived from the following fundamental theorem:
\begin{Thm}\label{BasicMLM}
We have the following:
\begin{itemize}
\item[\rm (i)]$\displaystyle S(\M^{\rm MLM}_{\Lambda}, \vphi^{\rm MLM}_{\Lambda})=\big||\Lambda_A|-|\Lambda_B|\big|\big/2$ for all $\Lambda\in F^{\rm (e)}_{\mathbb{L}}$. 
\item[\rm (ii)] The mapping $S_{\bullet} : H\in \mathscr{A}_{\Lambda, |\Lambda|}(\M_{\Lambda}^{\rm MLM}, \vphi_{\Lambda}^{\rm MLM})\mapsto S_{H}$ is constant and satisfies $S_{H}=\big||\Lambda_A|-|\Lambda_B|\big|\big/2$ for all $\Lambda\in F^{\rm (e)}_{\mathbb{L}}$. Here, recall the definition of $\mathscr{A}_{\Lambda, N}(\cdots)$, i.e., \eqref{DefALN}.
\item[\rm (iii)] Suppose that we are given a magnetic system $O_{ 1/2}$ in the MLM stability class $\mathscr{C}(O^{\rm MLM}_{1/2})$, i.e.,  $O_{ 1/2} \LRA O^{\rm MLM}_{1/2}$. In addition,  suppose a net of Hamiltonians ${\bs H}$ is adapted to $O_{ 1/2}$. Then,  for all $\Lambda\in F^{\rm (e)}_{\mathbb{L}}$, each  ground state of $H_{\Lambda}\in {\bs H}$ has   total spin $S_{H_{\Lambda}}=\big||\Lambda_{ A}|-|\Lambda_{B}|\big|\big/2$. If there exist a non-negative constant $s$ and an increasing sequence of sets $\{\Lambda_n : n\in \BbbN\}\subset  F^{\rm (e)}_{\mathbb{L}}$ such that $\bigcup_{n=1}^{\infty} \Lambda_n=\mathbb{L}$ and 
$\big||\Lambda_{n, A}|-|\Lambda_{n, B}|\big|=2s|\Lambda_n|+o(|\Lambda_n|)$ as $n\to \infty$, then each NMGS  associated with $\bs H$ exhibits a  magnetic order with a spin density  of $s$.
\item[\rm (iv)] Suppose we are in the same setting as in {\rm (iii)}.
Denote  $O_{1/2} = \big\{\{\M_{\Lambda}, \vphi_{\Lambda}\} : \Lambda\in F_{\mathbb{L}}\big\}$, and let $\ilim F(\M_{\Lambda_n}, \vphi_{\Lambda_n})$ and $
\ilim F(\M^{\rm MLM}_{\Lambda_n}, \vphi^{\rm MLM}_{\Lambda_n})
$ be the standard forms for macroscopic systems as defined in Theorem \ref{MacroFinite}. In this case, the following diagram is commutative for each $m, n\in \BbbN$ with $m<n$:
 \be
\begin{tikzcd}
F(\M_{\Lambda_m}, \vphi_{\Lambda_m}) \arrow[d]  &\arrow[l] F(\M_{\Lambda_{n}}, \vphi_{\Lambda_{n}}) \arrow[d]  &\arrow[l, dashed]    \ilim F(\M_{\Lambda_n}, \vphi_{\Lambda_n}) \arrow[d,dashed]\\
  F(\M^{\rm MLM}_{\Lambda_m}, \vphi^{\rm MLM}_{\Lambda_m})  &\arrow[l]F(\M^{\rm MLM}_{\Lambda_{n}}, \vphi^{\rm MLM}_{\Lambda_{n}})  &\arrow[l,dashed]  \ilim F(\M^{\rm MLM}_{\Lambda_n}, \vphi^{\rm MLM}_{\Lambda_n})
\end{tikzcd}
\ee

\end{itemize}
\end{Thm}
\begin{proof} (i)
Fix $\Lambda\in F^{\rm (e)}_{\mathbb{L}}$ arbitrarily.
The {\it MLM Hamiltonian} on $\Lambda$ is defined by
\be
H^{\rm MLM}_{\Lambda}={\bs S}_{\Lambda_A}\cdot {\bs S}_{\Lambda_B}, \label{MLMHami}
\ee
where ${\bs S}_{\Lambda_A}=(S_{\Lambda_A}^{(1)}, S_{\Lambda_A}^{(2)}, S_{\Lambda_A}^{(3)})$ denote the total spin operators on $\Lambda_A$ : $S_{\Lambda_A}^{(i)}=\sum_{x\in \Lambda_A}S_x^{(i)}\, (i=1, 2, 3)$. The Hamiltonian $\MLM$ acts in $\h_{\Lambda}$.
Although  $H^{\rm MLM}_{\Lambda}$ has a simple form, the magnetic properties of its ground states are fundamental.  By (i) of Theorem \ref{ManyHa}, for any $M\in \mathrm{spec}(S_{\Lambda}^{(3)} \restriction\h_{\Lambda})$, we have
\be
e^{-\beta H_{\Lambda, M}^{\rm MLM}} \rhd 0\ \ \mbox{w.r.t. $L^2(\M_{\Lambda}^{\rm MLM}[M], \vphi_{\Lambda, M}^{\rm MLM})_+$ $\forall \beta >0$,}
\ee
where  $ H_{\Lambda, M}^{\rm MLM}$ denotes the restriction of $ H_{\Lambda}^{\rm MLM}$ to the $M$-subspace.
Due to Theorem \ref{pff}, the ground state, $\psi_{\Lambda, 0}$, of $H_{\Lambda, 0}^{\rm MLM}$ is unique and strictly positive w.r.t. $L^2(\M_{\Lambda}^{\rm MLM}[0], \vphi_{\Lambda, M}^{\rm MLM})_+$. Because $H_{\Lambda, 0}^{\rm MLM}$ commutes with $S_{\Lambda}^{(i)}\, (i=1, 2, 3)$, we know that $\psi_{\Lambda, 0}$ has total spin $S=\big||\Lambda_A|-|\Lambda_B|\big|/2$. In addition, by using arguments similar to those of the proof
of Proposition \ref{GSinG}, the ground state of $H_{\Lambda, M}^{\rm MLM}$ is unique, can be chosen to be strictly positive  w.r.t. $L^2(\M_{\Lambda}^{\rm MLM}[M], \vphi_{\Lambda, M}^{\rm MLM})_+$ and has total spin $S$  for all $M\in\mathrm{spec}(S_{\Lambda}^{(3)} \restriction \h_{\Lambda})$ with $|M| \le S$.
We denote by $\psi_{\Lambda, M}$ the ground state of $H_{\Lambda, M}^{\rm MLM}$. Define the  vector $\psi_{\Lambda}$ by $\psi_{\Lambda}=\bigoplus_{M\in \mathrm{spec}(S_{\Lambda}^{(3)} \restriction \h_{\Lambda})} \tilde{\psi}_{\Lambda, M}$, where
\begin{align}
\tilde{\psi}_{\Lambda, M}=
\begin{cases}
\psi_{\Lambda, M} & \mbox{if $|M| \le S$}\\
0 & \mbox{otherwise}.
\end{cases} 
\end{align}
From the construction, we know that $\psi_{\Lambda}$ is a ground state of $H_{\Lambda}^{\rm MLM}$ and belongs to  $ G(\M^{\rm MLM}_{\Lambda}, \vphi^{\rm MLM}_{\Lambda})$. Hence, from Theorem \ref{BasicThm}, it follows that $S(\M^{\rm MLM}_{\Lambda}, \vphi^{\rm MLM}_{\Lambda})=S$.

(ii),  (iii) and (iv) follow from Corollary \ref{StaGsBigA},  Theorem \ref{OtoO} and Theorem \ref{MacroDiagram}, respectively.
\end{proof}

Recall the discussion in Section \ref{Sect4} about abstract crystal lattices and their realizations.
By combining Proposition \ref{IsoStaCl} and Theorem \ref{BasicMLM}, we obtain the following.
\begin{Thm}
Let $G$ be an abstract crystal lattice and $\pi$ be an arbitrary realization of $G$.
Then Theorem \ref{BasicMLM} holds for the realized crystal lattice $G^{\pi}$. We will write Theorem \ref{BasicMLM}(${\pi}$) for Theorem \ref{BasicMLM} that holds for $G^{\pi}$.
In this case, the values of the total spin and the spin density in Theorem \ref{BasicMLM}($\pi$) do not depend on the realization $\pi$. In other words, these are graph invariants.

Let $O_{1/2}^{\rm MLM, \pi}$ be the MLM system on $G^{\pi}$. Then 
$\mathscr{C}(O^{\rm MLM}_{\varrho})$ and $\mathscr{C}(O_{\varrho}^{\rm MLM, \pi})$ are isomorphic in the sense of Proposition \ref{IsoStaCl}.
\end{Thm}

With this theorem in mind, we will mainly consider magnetic systems on abstract graphs in the remainder of this section. In concrete examples, however, we will consider the specific realizations of crystal lattices that are visually comprehensible.

Here are some typical examples:

\begin{Exa}\upshape
The two-dimensional square lattice in Figure. \ref{Crystal lattices} is a bipartite connected graph.
If $G_{\Lambda_n}$ is a square lattice with one side of length $2n$, we know that 
 $
\big||\Lambda_{n ,A}|-|\Lambda_{n, B}|\big|=0
$.
A similar consideration holds for $d$-dimensional hypercubic lattices.
\end{Exa}

\begin{Exa}[Star]\upshape
Given a $k\in \BbbN\, (k\ge 2)$, a {\it star} $S_{k}$ is a tree with one internal vertex and $k$ leaves.
Figure \ref{S_8} illustrates $S_8$. 
\begin{figure}[h]
\begin{center}
\includegraphics[scale=0.7]{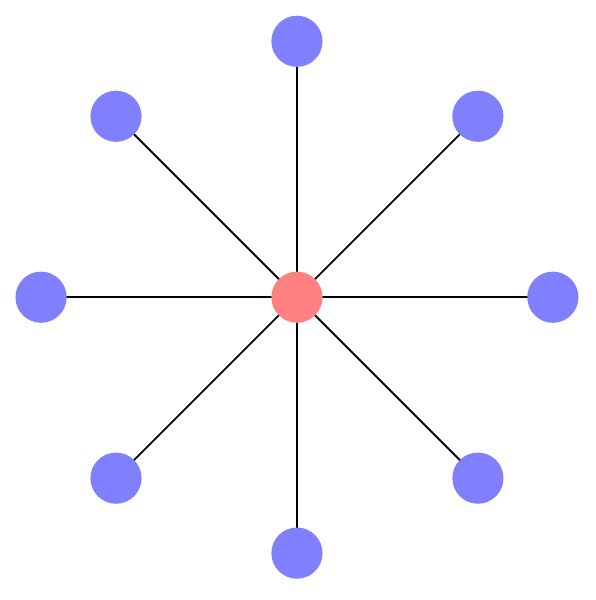}
\caption{The star $S_8$} 
    \label{S_8}
\end{center}
\end{figure}
 In this example, we  consider $S_{2n-1}\, (n\in \BbbN)$.
We denote the vertex at the center of the star as $x_0$ and the surrounding leaves as $L_n$.
$S_{2n-1}$ is a complete and bipartite graph. The bipartite structure of $S_{2n-1}=(\Lambda_n, E_n)$ is precisely given by 
$\Lambda_n=\Lambda_{ n, A} \sqcup_{E_n}\Lambda_{n, B} $ with 
$\Lambda_{n, A}=\{x_0\}$ and $\Lambda_{ n, B}=L_n$.
Hence, $\big||\Lambda_{n, A}|-|\Lambda_{n, B}|\big|=2n-2$ holds.

Next, we consider $S_{\infty}=(\mathbb{L}, E)$. We are interested in an increasing sequence
$\{S_{2n-1}=(\Lambda_n, E_n) : n\in \BbbN, \, n\ge 2\}$
 of subgraphs of $S_{\infty}$, where $S_{2n-1}$s  are chosen so that $S_{2n-1}\subset S_{2n+1}$ and $\bigcup_{n=2}^{\infty} \Lambda_n=\mathbb{L}$.
 Note that 
 \be
\frac{\big||\Lambda_{n, A}|-|\Lambda_{n, B}|\big|}{|\Lambda_n|}=\frac{2n-2}{2n}
\xrightarrow[n\to \infty]{} 1. \label{LimStar}
\ee
In the above setting, let us consider the MLM system $O_{1/2}^{\rm MLM}$ associated with this increasing sequence. Let $O_{1/2}$ be a magnetic system in the MLM stability class $\mathscr{C}(O_{1/2}^{\rm MLM})$.  Suppose we are given a net of Hamiltonians ${\bs H}$ adapted to $O_{1/2}$.
Every NMGS exhibits a magnetic order with a spin density of 
$1/2$ in the thermodynamic  limit.
\end{Exa}

\begin{Exa}[Regular tree]\upshape
A {\it regular tree}, $R^z$,  of degree $z$ is the infinite tree with $z$ edges at each vertex.
A regular tree is sometimes called a {\it Bethe lattice}. 
With one vertex, say $x_0$,  chosen as root, all other vertices are arranged in shells around this root vertex, which is then also called the origin of the graph. The number of vertices in the $k$-th shell, $H_k$, is given by
\be
|H_k|=z(z-1)^{k-1}\ (k\ge 1).
\ee
Figure \ref{R^3} illustrates $R^3$.
\begin{figure}[h]
\begin{center}
\includegraphics[scale=0.7]{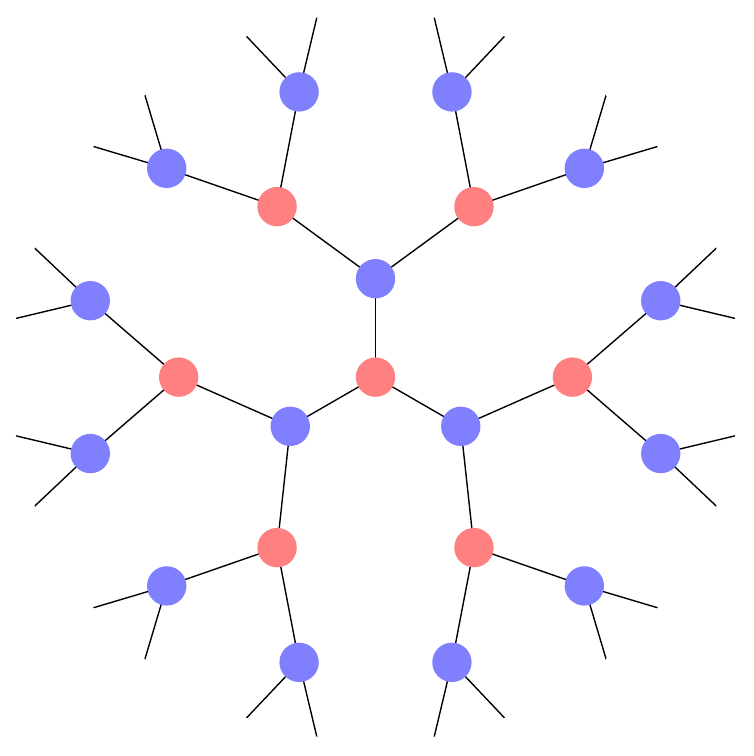}
\caption{A Behte lattice with $z=3$} 
    \label{R^3}
\end{center}
\end{figure}
Trivially, $R^z=(\mathbb{L}, E)$ is a  bipartite connected graph with the bipartite structure 
$\mathbb{L}_A=\{x_0\}\cup \bigcup_{k=1}^{\infty} H_{2k}$ and $\mathbb{L}_B=\bigcup_{k=1}^{\infty}H_{2k+1}$.
Below, let $z$ be an odd number. Given a natural number $n$, we set 
$\Lambda_n=\{x_0\}\cup \bigcup_{k=1}^n H_k$. 
In this case, we define 
$R^z_n=(\Lambda_n, E_n)$
 as the corresponding induced subgraph of $R^z$. We note that $|\Lambda_n|$  is indeed an even number.
 In this way,  we have constructed an increasing sequence $\{R^z_n : n\in \BbbN\}$ of subgraphs. A simple calculation shows that 
\be
\frac{\big||\Lambda_{n, A}|-|\Lambda_{n, B}|\big|}{|\Lambda_n|}
\xrightarrow[n\to \infty]{} \frac{z-2}{z}.
\ee
Under the above settings, let $O_{1/2}^{\rm MLM}$ be the MLM system associated with this increasing sequence.
Now, consider a magnetic system $O_{1/2}$ belonging to the MLM stability class $\mathscr{C}(O_{1/2}^{\rm MLM})$. We also consider a net of Hamiltonians ${\bs H}$ that is adapted to $O_{1/2}$.
Every NMGS associated with $\bs H$ exhibits a magnetic order with a  spin density of 
$\frac{1}{2}-\frac{1}{z}$
in the thermodynamic  limit.
By varying $z$, we can construct a magnetic order with spin density $s$ for various values of $s\in (0, 1/2)$.
\end{Exa}

\begin{Exa}\upshape
Consider the decorated chain given in Figure. \ref{1D}. If we denote by $G_{\Lambda_n}$ the subgraph with $2n$ edge lengths, we can see that $\big||\Lambda_{n, A}|-|\Lambda_{n, B}|\big|/|\Lambda_n|=1/2$.
Note that we can consider various other decorations of the chain than the one considered here.

\begin{figure}[h]
\begin{center}
\includegraphics[scale=0.7]{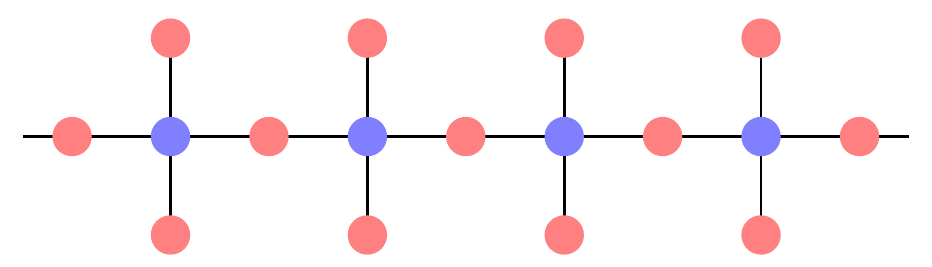}
\caption{A decorated  chain}
\label{1D}
\end{center}
\end{figure}
\end{Exa}

\begin{Exa}[The Lieb lattice]\upshape
Let us consider the two-dimensional Lieb lattice given in Figure \ref{Lieb-lattice}. As seen from this figure, the Lieb lattice is a bipartite connected graph, and if we denote by $G_{\Lambda_n}$ the Lieb lattice with $2n$ edge lengths, we can easily verify that $\big||\Lambda_{n, A}|-|\Lambda_{n, B}|\big|/|\Lambda_n|=1/3$.

\begin{figure}[h]
\begin{center}
\includegraphics[scale=0.7]{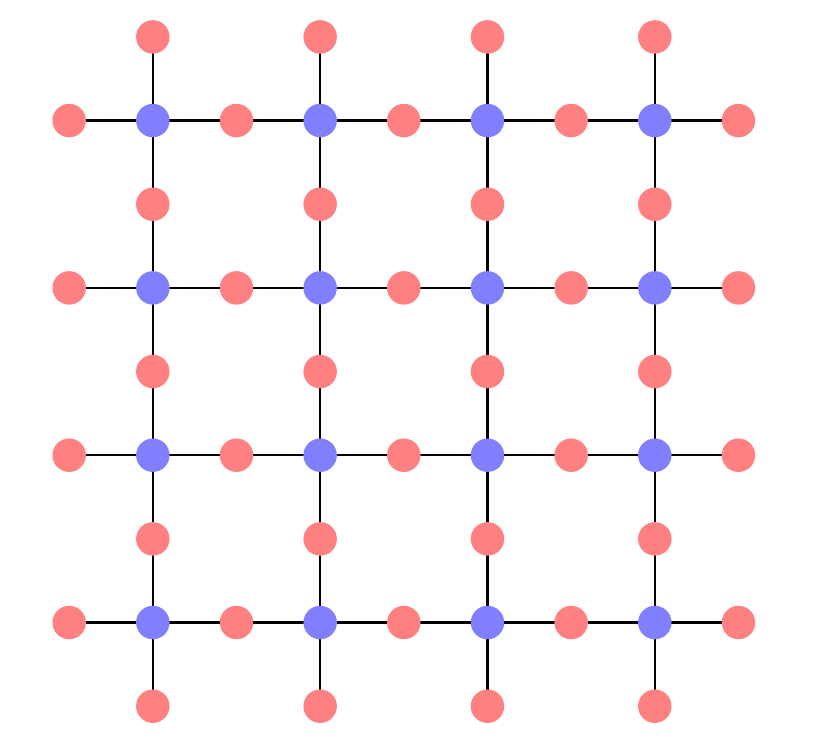}
\caption{The 2D Lieb lattice} 
    \label{Lieb-lattice}
\end{center}
\end{figure}
Recently, finite Lieb lattices have been realized in the laboratory, and their electronic properties have been clarified; see \cite{Slot2017}.
\end{Exa}

\subsection{The Heisenberg model}\label{SubsecHei}
As an application of Theorem \ref{BasicMLM}, let us first investigate the antiferromagnetic Heisenberg model.
The Heisenberg Hamiltonian on $\Lambda$ is given by 
\be
H^{\rm Hei}_{\Lambda} =\sum_{x, y\in \Lambda} J_{x y} {\bs S}_x\cdot {\bs S}_y.
\ee
The operator $H^{\rm Hei}_{\Lambda}$ acts in $\h_{\Lambda}$.
Set $E^J=\{\{x, y\} : J_{xy}\neq 0\}$ and $E_{\Lambda}^J=\{\{x, y\} : J_{xy}\neq 0, x,y\in \Lambda\}$ $(\Lambda\in F_{\mathbb{L}})$. Correspondingly, we have the graphs $G^J=(\mathbb{L}, E^J)$ and $G^J_{\Lambda}=(\Lambda, E_{\Lambda}^J)$;
$G^J$ and $G^J_{\Lambda}$ are called the {\it  graphs generated  by the interaction} $\{J_{xy} : x, y\in \mathbb{L}\}$.

We will use $G^J$ to describe the conditions that the interaction $J$ must satisfy. To do so, we need some preparation in graph theory.
First, let us recall the following well-known result:
every connected graph contains a normal spanning tree, with any specified vertex as its root. 
For proof of this claim, see \cite[Theorem 3]{Soukup}.
In the following, we write $T(G)=(\mathbb{L}, E^{\rm NST})$ for the arbitrarily fixed  normal spanning tree that $G$ contains.
Note that if $G$ is bipartite and connected, then $T(G)$ is also bipartite and $\mathbb{L}_A\sqcup_E \mathbb{L}_B=\mathbb{L}_A\sqcup_{E^{\rm NST}} \mathbb{L}_B$ holds.
Given a $\Lambda\in F_{\mathbb{L}}$,  we denote by $T(G)_{\Lambda}=(\Lambda, E_{\Lambda}^{\rm NST})$ the induced subgraph of $T(G)$.

In this section, we assume the following:
\begin{Assum}\label{CoupAss}\upshape
\indent
\begin{itemize}
\item[(i)] $J_{x y} =J_{yx}\ge 0$ for all $x, y\in \mathbb{L}$.
\item[(ii)] $T(G)$ is a subgraph of $G^J$ such that $\mathbb{L}_A\sqcup_E \mathbb{L}_B=\mathbb{L}_A\sqcup_{E^J} \mathbb{L}_B=\mathbb{L}$.
\item[(iii)] For each $\Lambda\in F_{\mathbb{L}}$, $T(G)_{\Lambda}$ is a subgraph of $G_{\Lambda}^J$ such that $\Lambda_A\sqcup_{E_{\Lambda}} \Lambda_B=\Lambda_A\sqcup_{E_{\Lambda}^J} \Lambda_B=\Lambda$.
\end{itemize}
\end{Assum}

\begin{Exa}\upshape
Let $G$ be an abstract crystal lattice and let $\pi$ be its realization. In general, if an interaction $J=\{J_{xy}\}$ satisfies $G^J=G^{\pi}$, then
$J$ is said to describe nearest-neighbor interactions. In this paper, we will extend this terminology to say that $J$ also describes nearest-neighbor interactions when $G^J=G$.\footnote{
Since we are considering an abstract graph, $\mathbb{L}$ is not a metric space. Therefore, the term  ``nearest neighbor" is not accurate, but for convenience, we will also call $J$, in this case, a nearest neighbor interaction.
}
\end{Exa}

\begin{Exa}\upshape
Let $\overline{G}=(\mathbb{L}, \overline{E})$ and $\overline{G}_{\Lambda}=(\Lambda, \overline{E}_{\Lambda})$ be the completed graphs such that  $\mathbb{L}_A\sqcup_{\overline{E}} \mathbb{L}_B=\mathbb{L}$ and 
  $\Lambda_A\sqcup_{\overline{E}_{\Lambda}} \Lambda_B=\Lambda$.
  Define the interaction $J^{\rm MLM}_{xy}$ by 
  \begin{align}
  J^{\rm MLM}_{xy}
  =\begin{cases}
  1, & \mbox{if $x\in \mathbb{L}_A$ and $y\in \mathbb{L}_B$, or $x\in \mathbb{L}_B$ and $y\in \mathbb{L}_A$  }\\
  0, & \mbox{otherwise}.
  \end{cases}
  \end{align}
  Hence, $E^{J^{\rm MLM}}=\overline{E}$ and $E_{\Lambda}^{J^{\rm MLM}}=\overline{E}_{\Lambda}\ (\Lambda\in F_{\mathbb{L}})$.
 We see that  the Heisenberg Hamiltonian associated with the interaction $\{J_{xy}^{\rm MLM} : x, y\in \mathbb{L}\}$ is equal to the MLM Hamiltonian given by \eqref{MLMHami}. From this, we can consider the MLM Hamiltonian  to be a particular case of the Heisenberg  Hamiltonian.
 
 The readers may think that the MLM Hamiltonian is artificial and physically unimportant. However, as we have seen, our theory describes magnetic properties independent of the details of the Hamiltonians. From this standpoint, the MLM Hamiltonian is worth analyzing. In fact, as we have seen in the proof of Theorem \ref{BasicMLM}, the MLM Hamiltonian is the fundamental Hamiltonian that characterizes the MLM stability class, and it has a remarkable feature not found in other Heisenberg models: the ground state energy and the total spin of the ground states  can be easily calculated.  This example highlights how the theory constructed in this paper is applied in practice: first to elucidate the magnetic properties of straightforward and solvable models and then build the corresponding stability classes before analyzing more complex models.
\end{Exa}

\begin{Thm}\label{ThmMLMHei}
Assume Condition \ref{CoupAss}.
We have the following:
\begin{itemize}
\item[{\rm (i)}] $H_{\Lambda}^{\rm Hei} \in A_{\Lambda, |\Lambda|}(\M_{\Lambda}^{\rm MLM}, \vphi_{\Lambda}^{\rm MLM})$ for all $\Lambda\in F^{\rm (e)}_{\mathbb{L}}$. In addition,  we have
\be
S_{H_{\Lambda}^{\rm Hei}}=\big||\Lambda_A|-|\Lambda_B|\big|\big/2.
\ee
\item[\rm (ii)] Set
${\bs H}^{\rm Hei}=\{H_{\Lambda}^{\rm Hei} :  \Lambda\in F^{\rm (e)}_{\mathbb{L}}\}$.  Then 
 ${\bs H}^{\rm Hei}$ are adapted to $O_{1/2}^{\rm MLM}$. Hence, 
 if there exist a non-negative constant $s$ and an increasing sequence of sets $\{\Lambda_n : n\in \BbbN\}\subset  F^{\rm (e)}_{\mathbb{L}}$ such that $\bigcup_{n=1}^{\infty} \Lambda_n=\mathbb{L}$ and 
$\big||\Lambda_{n, A}|-|\Lambda_{n, B}|\big|=2s|\Lambda_n|+o(|\Lambda_n|)$ as $n\to \infty$, then every NMGS  associated with ${\bs H}^{\rm Hei}$  exhibits a  magnetic order with a spin density of $s$.
 \end{itemize}
\end{Thm}
\begin{proof}
Fix $\Lambda\in F^{\rm (e)}_{\mathbb{L}}$ arbitrarily. From (ii) of Theorem \ref{ManyHa}, it follows that 
\be
e^{-\beta H_{\Lambda, M}^{\rm MLM}}\rhd 0
\ \mbox{ w.r.t. $L^2(\M_{\Lambda}^{\rm MLM}[M], \vphi_{\Lambda, M}^{\rm MLM})_+$}
\ee
 for all $\beta >0$ and $M\in \mathrm{spec}(S_{\Lambda}^{(3)} \restriction \h_{\Lambda})$.
Hence, $H^{\rm Hei}_{\Lambda} \in A_{\Lambda, |\Lambda|}(\M_{\Lambda}^{\rm MLM}, \vphi_{\Lambda}^{\rm MLM})$ for all $\Lambda\in F^{\rm (e)}_{\mathbb{L}}$, which implies that 
${\bs H}^{\rm Hei}$ is adapted to $O_{1/2}^{\rm MLM}$.  By applying Theorem \ref{BasicMLM}, we obtain the desired results in Theorem \ref{ThmMLMHei}. 
\end{proof}

\begin{Rem} \upshape
This theorem contains the following important message. Consider the Heisenberg Hamiltonians defined by the couplings $J_1=\{J_{1, xy}\}$ and $J_2=\{J_{2, xy}\}$. 
Let us say that the forms of $J_1$ and $J_2$ are very different. Even in such a case,  we can conclude from Theorem \ref{ThmMLMHei} that the magnetic structures of the ground states of the two Hamiltonians  will be the same as long as $J_1$ and $J_2$ satisfy Condition \ref{CoupAss}.
For example, the Hamiltonians described by $J^{\rm MLM}$ and $J$ satisfying $G^J=T(G)$ have the same magnetic properties in the ground states in a sense claimed in Theorem \ref{ThmMLMHei}.
Notice that $J^{\rm MLM}$ is maximal in the sense that $G^{J^{\rm MLM}}$ is  a complete graph, while $J$ is minimal in the sense that $G^J$ is equal to  $T(G)$. From this observation, it can be seen that Theorem \ref{ThmMLMHei} describes a stability concerning the magnetic properties of the ground states of quantum spin systems described by the Heisenberg models.
\end{Rem}

\subsection{The Hubbard model}\label{MLM-H}
In this subsection, we examine the Hubbard model:
\be
H_{\Lambda}^{\rm H}=\sum_{x, y\in \Lambda} \sum_{\sigma=\up, \down}t_{xy}c_{x\sigma}^*c_{y\sigma}
+\sum_{x, y\in \Lambda} \frac{U_{xy}}{2}(n_x-1)(n_y-1),  \label{DefHubbard}
\ee
where $t_{xy}$ is the hopping matrix, and $U_{xy}$ is the energy of the Coulomb interaction.
The operator $H_{\Lambda}^{\rm H}$ acts in $\h_{\Lambda}^{\rm H} :=\F_{\Lambda, N=|\Lambda|}$. Note here that the Hilbert space in which $H_{\Lambda}^{\rm H}$ acts is different from that of $H_{\Lambda}^{\rm Hei}$.
In this paper, we assume that $\{t_{xy}\}$ and $\{U_{xy}\}$ are $|\Lambda|\times |\Lambda|$ real symmetric matrices.

Set \be
\mathfrak{R}_{\Lambda}=\big\{c_{x_1\up}^*\cdots c_{x_n\up}^* c_{y_1\up} \cdots c_{y_n\up} : x_1, \dots, x_n, y_1, \dots, y_n\in \Lambda, n\in \{1, \dots, |\Lambda|\}\big\}. \label{DefRalg}
\ee
We denote by $\M_{\Lambda}^{\rm H}$ the von Neumann algebra generated by $\mathfrak{R}_{\Lambda}$.
Define the state on $\M_{\Lambda}^{\rm H}$ by $\vphi_{\Lambda}^{\rm H}(x)=\vphi_{\xi_{\Lambda}}(x)\, (x\in \M_{\Lambda}^{\rm H})$, where $\xi_{\Lambda}$ is defined by \eqref{TaneVec}.
Due to \eqref{HSingE}, $Q_{\Lambda}$ is 
 the orthogonal projection from $\h_{\Lambda}^{\rm H}$ to $\h_{\Lambda}$.
 In addition, 
  we have $Q_{\Lambda}\xi_{\Lambda}=\xi_{\Lambda}$, which implies that
\be
\vphi_{\Lambda}^{\rm H} \circ \mathscr{E}^{\rm H}_{\Lambda}=\vphi_{\Lambda}^{\rm H}, \label{EHtoH}
\ee 
where $\mathscr{E}^{\rm H}_{\Lambda}(x)=Q_{\Lambda} x Q_{\Lambda}$.

\begin{Lemm}\label{ConHubb}
We have the following:
\begin{itemize}
\item[\rm (i)] If $\Lambda, \Lambda'\in F^{\rm (e)}_{\mathbb{L}}$ satisfies $\Lambda\subseteq \Lambda'$, then 
$F(\M^{\rm H}_{\Lambda\rq{}}, \vphi^{\rm H}_{\Lambda\rq{}})\longrightarrow F(\M^{\rm H}_{\Lambda}, \vphi^{\rm H}_{\Lambda})$.
\item[\rm (ii)] For each $\Lambda 
\in F^{\rm (e)}_{\mathbb{L}}
, $
$F(\M^{\rm H}_{\Lambda}, \vphi^{\rm H}_{\Lambda})\longrightarrow F(\M^{\rm MLM}_{\Lambda}, \vphi^{\rm MLM}_{\Lambda})$.
\end{itemize}
The above results can be summarized as the following commutative diagram:
\be
\begin{tikzcd}
  F(\M^{\rm H}_{\Lambda}, \vphi^{\rm H}_{\Lambda})   \arrow[d] &
  \ar[l]   F(\M^{\rm H}_{\Lambda\rq{}}, \vphi^{\rm H}_{\Lambda\rq{}})
  \arrow[d]  \\
  F(\M^{\rm MLM}_{\Lambda}, \vphi^{\rm MLM}_{\Lambda}) &\arrow[l] F(\M^{\rm MLM}_{\Lambda\rq{}}, \vphi^{\rm MLM}_{\Lambda\rq{}})
\end{tikzcd}
\ee
In addition, the diagram in {\rm (iv)} of Theorem \ref{BasicMLM} with $\{\M_{\Lambda}, \vphi_{\Lambda}\}=\{\M^{\rm H}_{\Lambda}, \vphi^{\rm H}_{\Lambda}\}$ becomes commutative.
\end{Lemm}
\begin{proof}
(i) During this proof, we will also use the identification between fermionic Fock spaces used in the proof of Lemma \ref{VolSys}. 
Take $\Lambda, \Lambda' \in F_{\mathbb{L}}^{\rm (e)}$ such that  $\Lambda\subseteq \Lambda'$.
In the proof of Lemma \ref{VolSys}, we have  already proved  that $\h_{\Lambda}^{\rm H} \otimes \h_{\Lambda'\setminus \Lambda}^{\rm H}$can be regarded as  a subspace of $\h^{\rm H}_{\Lambda'}$. 
 
Given $\Lambda, \Lambda' \in F_{\mathbb{L}}^{\rm (e)}$ with $\Lambda\subseteq \Lambda'$, 
define the isometric linear mapping $\kappa^{\rm H} : \h_{\Lambda}^{\rm H} \to \h_{\Lambda'}^{\rm H}$ by 
$\kappa^{\rm H} \eta=\eta \otimes \tilde{\xi}_{\Lambda'\setminus \Lambda}\, (\eta \in \h^{\rm H}_{\Lambda})$, where $\tilde{\xi}_{\Lambda}$ is given in the proof of Lemma \ref{VolSys}. By identifying $\kappa^{\rm H} \h_{\Lambda}^{\rm H} $ with $\h_{\Lambda}^{\rm H}$, 
we can regard $\h_{\Lambda}^{\rm H}$ as a subspace of $\h_{\Lambda'}^{\rm H}$.
In addition, $\M_{\Lambda}^{\rm H}$ can be identified  with $\M_{\Lambda}^{\rm H}\otimes \Pi_{\Lambda'\setminus \Lambda}$, where 
$\Pi_{\Lambda' \setminus \Lambda}=|\tilde{\xi}_{\Lambda'\setminus \Lambda}\ra\la \tilde{\xi}_{ \Lambda'\setminus \Lambda} |$.
We denote by $P_{\Lambda\Lambda'}^{\rm H}$ the orthogonal projection from $\h^{\rm H}_{\Lambda'}$
to $\h_{\Lambda}^{\rm H}$. Trivially, $P_{\Lambda\Lambda'}^{\rm H}=1\otimes \Pi_{\Lambda'\setminus \Lambda}$.
We readily confirm the following:
\begin{itemize}
\item $\M^{\rm H}_{\Lambda}=\M_{\Lambda}^{\rm H}\otimes \Pi_{\Lambda'\setminus \Lambda}=
P_{\Lambda\Lambda'}^{\rm H}\M_{\Lambda'}^{\rm H}P_{\Lambda\Lambda'}^{\rm H}$.
\item $\vphi_{\Lambda'}^{\rm H} \circ \mathscr{E}_{\Lambda\Lambda'}^{\rm H}
=\vphi^{\rm H}_{\Lambda'}$, where $\mathscr{E}^{\rm H}_{\Lambda\Lambda'}(x)
=P_{\Lambda\Lambda'}^{\rm H} x P_{\Lambda\Lambda'}^{\rm H}\, (x\in \M_{\Lambda'}^{\rm H})$.
\end{itemize}
To complete the proof, we need only show the following:
\be
P_{\Lambda\Lambda'}^{\rm H}L^2(\M_{\Lambda'}^{\rm H}, \vphi_{\Lambda'}^{\rm H})_+=L^2(\M_{\Lambda}^{\rm H}, \vphi_{\Lambda}^{\rm H})_+. \label{LasBoss}
\ee
Recall that $\vphi_{\Lambda}^{\rm H}$ is defined by $\vphi_{\Lambda}^{\rm H}=\vphi_{\xi_{\Lambda}}$.
Because $\xi_{\Lambda}$ is cyclic and separating, $L^2(\M_{\Lambda}^{\rm H}, \vphi_{\Lambda}^{\rm H})_+$
can be expressed as
\be
L^2(\M_{\Lambda}^{\rm H}, \vphi_{\Lambda}^{\rm H})_+=\{x J_{\Lambda}x J_{\Lambda} \xi_{\Lambda} : x\in \M^{\rm H}_{\Lambda}\}.
\ee  
Note that $J_{\Lambda} c_{x_1\up}^*\cdots c_{x_n\up}^* c_{y_1\up} \cdots c_{y_n\up} J_{\Lambda}
=(-1)^{|X\cap \Lambda_B|+|Y\cap \Lambda_B|}c_{x_1\down}\cdots c_{x_n\down} c_{y_1\down}^* \cdots c_{y_n\down}^* $.
Hence, if we express $\eta\in L^2(\M_{\Lambda'}^{\rm H}, \vphi_{\Lambda'}^{\rm H})$ as 
\be
\eta=\sum_{{X, Y\subseteq \Lambda'}\atop{|X|=|Y|}} \eta_{XY} |X, \overline{Y}\ra_{\Lambda'},\label{ExpEta}
\ee
then $\eta$ belongs to $ L^2(\M_{\Lambda'}^{\rm H}, \vphi_{\Lambda'}^{\rm H})_+$, if and only if, $\{\eta'_{X, Y} \}_{X, Y}$ is a  positive semi-definite matrix, where $\eta'_{XY} $ is defined by
$\eta'_{X, Y}=\eta_{X, Y}$ if $|X|=|Y|$, $\eta_{X, Y}'=0$ if $|X|\neq |Y|$.
Note that we have used $ |X, \overline{Y}\ra_{\Lambda'}$ instead of $ |X, Y\ra_{\Lambda'}$ in the expression \eqref{ExpEta}.
Suppose that $\eta \in L^2(\M_{\Lambda'}^{\rm H}, \vphi_{\Lambda'}^{\rm H})_+$.
The action of the orthogonal projection  $P_{\Lambda\Lambda'}^{\rm H}$ on the vector $\eta$ is as follows: 
\be
P_{\Lambda\Lambda'}^{\rm H}\eta=\sum_{{X_{\Lambda}, Y_{\Lambda} \subseteq \Lambda}\atop{|X_{\Lambda}|=|Y_{\Lambda}|}} \tilde{\eta}_{X_{\Lambda}, Y_{\Lambda}} |X_{\Lambda}, \overline{Y}_{\Lambda}\ra_{\Lambda} \otimes \tilde{\xi}_{\Lambda'\setminus \Lambda},
\ee
where
\be
\tilde{\eta}_{X_{\Lambda}, Y_{\Lambda}} =\sum_{X_{\Lambda'\setminus \Lambda}\subseteq \Lambda'\setminus \Lambda}
\eta_{X_{\Lambda} \sqcup X_{\Lambda'\setminus \Lambda}, Y_{\Lambda} \sqcup X_{\Lambda'\setminus \Lambda}}.
\ee
Set $\tilde{\eta}'_{X_{\Lambda}, Y_{\Lambda}}=\tilde{\eta}_{X_{\Lambda}, Y_{\Lambda}}$ if $|X_{\Lambda}|=|Y_{\Lambda}|$, $\tilde{\eta}'_{X_{\Lambda}, Y_{\Lambda}} =0$ if $|X_{\Lambda}|\neq |Y_{\Lambda}|$.
Since $\{\eta'_{XY}\}$ is positive semi-definite, so is $\{\tilde{\eta}'_{X_{\Lambda}, Y_{\Lambda}}\}$. This means 
that $P_{\Lambda\Lambda'}^{\rm H}\eta \in L^2(\M_{\Lambda}^{\rm H}, \vphi_{\Lambda}^{\rm H})_+$.
Hence, we conclude that $P_{\Lambda\Lambda'}^{\rm H}L^2(\M_{\Lambda'}^{\rm H}, \vphi_{\Lambda'}^{\rm H})_+\subseteq L^2(\M_{\Lambda}^{\rm H}, \vphi_{\Lambda}^{\rm H})_+$. 
The inverse inclusion relation can be easily shown. 
This completes the proof of  \eqref{LasBoss}.

(ii) First, recall the definition of $Q_{\Lambda}$:  $Q_{\Lambda}=\prod_{x\in \Lambda} (n_{x\up}-n_{x\down})^2$. From the definition of $\h_{\Lambda}$, i.e., Eq. \eqref{HSingE}, we get 
\be
Q_{\Lambda} L^2(\M_{\Lambda}^{\rm H}, \vphi_{\Lambda}^{\rm H})=L^2(\M_{\Lambda}^{\rm MLM}, \vphi_{\Lambda}^{\rm MLM}).
\ee

Note that $Q_{\Lambda}  c_{x_1\up}^*\cdots c_{x_n\up}^* c_{y_1\up} \cdots c_{y_n\up} Q_{\Lambda}$ is non-zero,
 if and only if, $\{x_1, \dots, x_n\}=\{y_1, \dots, y_n\}$.
 In this case, $Q_{\Lambda}  c_{x_1\up}^*\cdots c_{x_n\up}^* c_{y_1\up} \cdots c_{y_n\up} Q_{\Lambda}$ is identical to $n_{x_1\up}\cdots n_{x_n\up}$, apart from the sign.
 From this fact, it follows that  $Q_{\Lambda} \M^{\rm H}_{\Lambda} Q_{\Lambda}\subseteq \M_{\Lambda}^{\rm MLM}$. Since it is easy to show the inverse inclusion relation, the equality 
 $Q_{\Lambda} \M^{\rm H}_{\Lambda} Q_{\Lambda}= \M_{\Lambda}^{\rm MLM}$
  is eventually established.
  
  Next, let us prove $Q_{\Lambda} L^2(\M_{\Lambda}^{\rm H}, \vphi_{\Lambda}^{\rm H})_+=L^2(\M_{\Lambda}^{\rm MLM}, \vphi_{\Lambda}^{\rm MLM})_+$.
 Recall that if we denote a vector $\eta\in L^2(\M_{\Lambda}^{\rm H}, \vphi_{\Lambda}^{\rm H})_+$ as 
 $\eta=\sum_{{X, Y\subseteq \Lambda}\atop{|X|=|Y|}} \eta_{XY} |X, \overline{Y}\ra_{\Lambda}$, then 
 $\{\eta'_{X, Y}\}$ is a positive definite matrix.
 In particular, we can see that the diagonal components are all positive: $\eta'_{X, X} \ge 0$.
 Since the action of $Q_{\Lambda}$ on $\eta$ is 
 $
 Q_{\Lambda}\eta=\sum_{X\subseteq \Lambda} \eta_{X, X}|X, \overline{X}\ra_{\Lambda}
 $, we know that $Q_{\Lambda}\eta\in L^2(\M_{\Lambda}^{\rm MLM}, \vphi_{\Lambda}^{\rm MLM})_+$.
 Hence, $ Q_{\Lambda} L^2(\M_{\Lambda}^{\rm H}, \vphi_{\Lambda}^{\rm H})_+\subseteq L^2(\M_{\Lambda}^{\rm MLM}, \vphi_{\Lambda}^{\rm MLM})_+$ holds. It is straightforward to show the inverse inclusion relation.
 
 Combining the above arguments with \eqref{EHtoH}, we can conclude that (ii) is valid.
\end{proof}

In what follows, we examine properties of  the magnetic system $O_{1/2}^{\rm H}
:=\big\{
\{\M_{\Lambda}^{\rm H}, \vphi^{\rm H}_{\Lambda}\} : \Lambda\in F_{\mathbb{L}}^{(\rm e)}
\big\}
$.

Let $G^t$ and $G^t_{\Lambda}\, (\Lambda\in F_{\mathbb{L}})$ be the graphs generated  by the hopping matrix
 $\{t_{xy}\}$; see  the definition immediately preceding Condition \ref{CoupAss}. 
The following conditions are imposed on $H_{\Lambda}^{\rm H}$.

\begin{Assum}\label{AssHubbard}\upshape
Let $T(G)$ be the normal spanning tree introduced in Condition \ref{CoupAss}.

\begin{itemize}
\item[(ii)] $T(G)$ is a subgraph of $G^t$ such that $\mathbb{L}_A\sqcup_E \mathbb{L}_B=\mathbb{L}_A\sqcup_{E^t} \mathbb{L}_B=\mathbb{L}$.
\item[(iii)] For each $\Lambda\in F_{\mathbb{L}}$, $T(G)_{\Lambda}$ is a subgraph of $G_{\Lambda}^t$ such that $\Lambda_A\sqcup_{E_{\Lambda}} \Lambda_B=\Lambda_A\sqcup_{E_{\Lambda}^t} \Lambda_B=\Lambda$.
\end{itemize}
\end{Assum}

Using the theory constructed in the previous sections,
Lieb's theorem is extended as follows.
\begin{Thm}\label{MLMHubbThm}
Assume Condition \ref{AssHubbard}. 
Assume that,
for each  $\Lambda\in F^{\rm (e)}_{\mathbb{L}}$, $\{U_{xy}\}_{x, y\in \Lambda}$ is a positive definite matrix:
\be
\sum_{x, y\in\Lambda} z_x^* z_y U_{xy} >0\ \ \ \forall \{z_x\}_{x\in \Lambda} \in \BbbC^{\Lambda} \setminus \{0\}.
\ee
We have the following:
\begin{itemize}
\item[\rm (i)] For each $\Lambda\in F^{\rm (e)}_{\mathbb{L}}$, $H_{\Lambda}^{\rm H}\in \mathscr{A}_{\Lambda, |\Lambda|}(\M_{\Lambda}^{\rm MLM}, \vphi_{\Lambda}^{\rm MLM})$ holds. Hence,  $S_{H_{\Lambda}^{\rm H}}=\big||\Lambda_A|-|\Lambda_B|\big|\big/2$.
 \item[\rm (ii)] $O_{1/2}^{\rm H} \LRA O_{1/2}^{\rm MLM}$. Hence, $O_{1/2}^{\rm H} \in \mathscr{C}(O_{1/2}^{\rm MLM})$.
 \item[\rm (iii)] Set ${\bs H}^{\rm H}=\{H_{\Lambda} ^{\rm H} : \Lambda \in F^{\rm (e)}_{\mathbb{L}}\}$. Then ${\bs H}^{\rm H}$ is adapted to $O_{1/2}^{\rm H}$.  Hence,
 if there exist a non-negative  constant $s$ and an increasing sequence of sets $\{\Lambda_n : n\in \BbbN\}\subset  F^{\rm (e)}_{\mathbb{L}}$ such that $\bigcup_{n=1}^{\infty} \Lambda_n=\mathbb{L}$ and 
$\big||\Lambda_{n, A}|-|\Lambda_{n, B}|\big|=2s|\Lambda_n|+o(|\Lambda_n|)$ as $n\to \infty$, then each NMGS  associated with ${\bs H}^{\rm H}$ exhibits a  magnetic order with a spin density of  $s$.
\end{itemize}
\end{Thm}
\begin{proof}
Fix $\Lambda\in F^{\rm (e)}_{\mathbb{L}}$. By applying (iii) of Theorem \ref{ManyHa}, we know that
$
\{e^{-\beta H^{\rm H}_{\Lambda, M}} \}_{\beta\ge 0}$ is ergodic  w.r.t. $L^2(\M^{\rm H}_{\Lambda}[M], \vphi^{\rm H}_{\Lambda, M})_+$ for all $ M\in \mathrm{spec}(S_{\Lambda}^{(3)} \restriction \h_{\Lambda})$.
Hence, $H^{\rm H}_{\Lambda}\in A_{\Lambda, |\Lambda|}(\M_{\Lambda}^{\rm H}, \vphi_{\Lambda}^{\rm H})$
for all $\Lambda\in F^{\rm (e)}_{\mathbb{L}}$, which means that ${\bs H}^{\rm H}$ is adapted to $O_{1/2}^{\rm H}$.
By using Theorem \ref{BasicMLM} and Lemma \ref{ConHubb}, we conclude the desired results in Theorem \ref{MLMHubbThm}.
\end{proof}

In this way, it turns out that we can unify the MLM theorem and Lieb's theorem.

\subsection{The Holstein--Hubbard model} \label{MLM-HH}
Let us consider the Holstein--Hubbard model of interacting electrons coupled to dispersionless phonons of frequency $\omega > 0$. The Hamiltonian is
\be
H^{\rm HH}_{\Lambda} = H^{\rm H}_{\Lambda} +
\sum_{x, y\in \Lambda}
g_{xy} (n_x -1)(b_y^*
 + b_y ) +\sum_{x\in \Lambda}
\omega b_x^*b_x,
\ee
where $H_{\Lambda}^{\rm H}$ is the Hubbard Hamiltonian given by \eqref{DefHubbard};  $b_x^*$
 and $b_x$ are bosonic creation- and annihilation operators at site $x\in \Lambda$, respectively. The operators $b_x^*$
 and $b_x$ satisfy the canonical commutation
relations:
\be
[b_x, b_y^*]=\delta_{xy},\ \ [b_x, b_y]=0.
\ee
$g_{xy}$ is the strength of the electron-phonon interaction.
We assume that the matrix $\{g_{xy}\}$ is real and symmetric. The operator $H_{\Lambda}^{\rm HH}$ acts in 
$\h^{\rm HH}_{\Lambda}=\h^{\rm H}_{\Lambda} \otimes \F^{\rm ph}_{\Lambda}$, where $\F^{\rm ph}_{\Lambda}$ is the bosonic Fock space over $\ell^2(\Lambda)$: $\F^{\rm ph}_{\Lambda}=\bigoplus_{n=0}^{\infty} \otimes^n_{\rm s} \ell^2(\Lambda)$. By the Kato--Rellich theorem \cite[Theorem X.12]{Reed1975}, $H_{\Lambda}^{\rm HH}$ is self-adjoint on $
\D(N^{\rm ph}_{\Lambda})$, bounded from below, where $N_{\Lambda}^{\rm ph}=\sum_{x\in \Lambda} b_x^*b_x$.

Recall the following identification:
\be
\F^{\rm ph}_{\Lambda}=L^2(\BbbR^{|\Lambda|}). \label{IdnLF}
\ee
Using this fact, we define the abelian von Neumann algebra on $\F^{\rm ph}_{\Lambda}$ by 
\be
\M_{\Lambda}^{\rm ph}=L^{\infty}(\BbbR^{|\Lambda|}).
\ee
The bosonic Fock vacuum in $\F^{\rm ph}_{\Lambda}$ is denoted by $\Omega_{\Lambda}^{\rm ph}$. Now, define
\be
\M^{\rm HH}_{\Lambda}=\M^{\rm H}_{\Lambda} \otimes \M^{\rm ph}_{\Lambda}\label{vNTensor}
\ee
and $\vphi^{\rm HH}_{\Lambda}=\vphi_{\xi_{\Lambda}^{\rm HH}}$, where $\xi_{\Lambda}^{\rm HH}=\xi_{\Lambda} \otimes \Omega_{\Lambda}^{\rm ph}$. Here, recall that $\xi_{\Lambda}$ is defined by \eqref{TaneVec}. Under the identification \eqref{IdnLF}, we can identify 
$
\Omega_{\Lambda}^{\rm ph}$ with the function $(\omega/|\Lambda|)^{|\Lambda|/4} e^{-\omega {\bs q}^2/2},\ \ {\bs q}\in \BbbR^{|\Lambda|}
$. Hence, we readily confirm that $\vphi^{\rm HH}_{\Lambda}$ is a faithful semi-finite normal weight on $\M^{\rm HH}_{\Lambda}$. The   isometric linear  mapping $\kappa^{\rm HH}_{\Lambda}$ from $\h^{\rm H}_{\Lambda}$ into $\h^{\rm HH}_{\Lambda}$ is given by 
$
\kappa_{\Lambda}^{\rm HH} \eta=\eta\otimes \Omega^{\rm ph}_{\Lambda}\  (\eta\in \h^{\rm H}_{\Lambda}).
$ Hence, $\h^{\rm HH}_{\Lambda}$ is an IEE space associated with $\h_{\Lambda}^{\rm H}$.
Now let $Q_{\Lambda}^{\rm HH}$ be the orthogonal projection from $\h^{\rm HH}_{\Lambda}$ to $\h^{\rm H}_{\Lambda}$.  Trivially, $Q^{\rm HH}_{\Lambda}\xi_{\Lambda}^{\rm HH}=\xi_{\Lambda}^{\rm HH}$ holds, which implies that 
\be
\vphi^{\rm HH}_{\Lambda} \circ \mathscr{E}_{\Lambda}^{\rm HH}=\vphi^{\rm H}_{\Lambda}, 
\ee
where $
\mathscr{E}_{\Lambda}^{\rm HH}(x)=Q^{\rm HH}_{\Lambda} x Q^{\rm HH}_{\Lambda}\ (x\in \M_{\Lambda}^{\rm HH})
$. 

\begin{Lemm}\label{ConHoHu}
We have the following:
\begin{itemize}
\item[\rm (i)] If $\Lambda, \Lambda'\in F^{\rm (e)}_{\mathbb{L}}$ satisfies $\Lambda\subseteq \Lambda'$, then 
$F(\M^{\rm HH}_{\Lambda\rq{}}, \vphi^{\rm HH}_{\Lambda\rq{}})\longrightarrow F(\M^{\rm HH}_{\Lambda}, \vphi^{\rm HH}_{\Lambda})$.
\item[\rm (ii)] For each $\Lambda 
\in F^{\rm (e)}_{\mathbb{L}}
, $
$F(\M^{\rm HH}_{\Lambda}, \vphi^{\rm HH}_{\Lambda})\longrightarrow F(\M^{\rm H}_{\Lambda}, \vphi^{\rm H}_{\Lambda})$.
\end{itemize}
The above results can be summarized as the following commutative diagram:
\be
\begin{tikzcd}
  F(\M^{\rm HH}_{\Lambda}, \vphi^{\rm HH}_{\Lambda})   \arrow[d] &
  \ar[l]   F(\M^{\rm HH}_{\Lambda\rq{}}, \vphi^{\rm HH}_{\Lambda\rq{}})
  \arrow[d]  \\
  F(\M^{\rm H}_{\Lambda}, \vphi^{\rm H}_{\Lambda}) &\arrow[l] F(\M^{\rm H}_{\Lambda\rq{}}, \vphi^{\rm H}_{\Lambda\rq{}})
\end{tikzcd}
\ee
Furthermore, when $\{\M_{\Lambda}, \vphi_{\Lambda}\}$  is replaced by $\{\M^{\rm H}_{\Lambda}, \vphi^{\rm H}_{\Lambda}\}$ and $\{\M^{\prime}_{\Lambda}, \vphi^{\prime}_{\Lambda}\}$ is replaced by $\{\M^{\rm HH}_{\Lambda}, \vphi^{\rm HH}_{\Lambda}\}$ in Theorem \ref{MacroDiagram}, the diagram \eqref{MacDiagComm} becomes commutative.
\end{Lemm}
\begin{proof}
(i)
The basic idea of the proof is the same as in  the proof of Lemma \ref{ConHubb}.
In this proof, we will use the identification between fermionic Fock spaces used in the proof of Lemma \ref{VolSys}. In addition, the following identification between bosonic Fock spaces is also helpful:
 given Hilbert spaces $\X$ and $\frak{Y}$, $\F_{\rm B}(\X\oplus \frak{Y})=\F_{\rm B}(\X)\otimes \F_{\rm B}(\frak{Y})$ holds, where $\F_{\rm B}(\X)$ denotes the bosonic Fock space over $\X$: $\F_{\rm B}(\X)=\bigoplus_n \otimes_{\rm s}^n \X$. To be precise, there exists a unitary operator $\tau : \F_{\rm B}(\X\oplus \frak{Y})\to \F_{\rm B}(\X)\otimes \F_{\rm B}(\frak{Y})$ satisfying 
$
\tau|\varnothing \ra^{\rm B}_{\X\oplus \frak{Y}}=|\varnothing \ra^{\rm B}_{\X} \otimes |\varnothing\ra^{\rm B}_{\frak{Y}}
$ and
\be
\tau b_{\X\oplus \frak{Y}}(f\oplus g) \tau^{-1}=b_{\X}(f)\otimes 1+ 1 \otimes b_{\frak{Y}}(g)\ \ (f\in \X, g\in \frak{Y}).
\ee
Here, $|\varnothing \ra^{\rm B}_{\X}$ and   $b_{\X}(f)$  stand for the Fock vacuum and  an  annihilation operator  in $\F_{\rm B}(\X)$, respectively. Using the above observation, we obtain the practical identification: $\F^{\rm ph}_{\Lambda\rq{}}=\F^{\rm ph}_{\Lambda} \otimes \F^{\rm ph}_{\Lambda\rq{}\setminus \Lambda}$ provided   that $\Lambda\subset \Lambda\rq{}$.

Take $\Lambda, \Lambda'\in F^{\rm (e)}_{\mathbb{L}}$ with $\Lambda\subseteq \Lambda'$. Using the above identifications, we can regard $\h_{\Lambda}^{\rm HH} \otimes \h_{\Lambda'\setminus \Lambda}^{\rm HH}$ as as closed subspace of $\h_{\Lambda'}^{\rm HH}$.\footnote{ Indeed, by the identifications,
 we get 
 $
 \F_{\Lambda'} \otimes \F^{\rm ph}_{\Lambda'} =
 (\F_{\Lambda}\otimes \F_{\Lambda'\setminus \Lambda}) \otimes (\F^{\rm ph}_{\Lambda}\otimes \F^{\rm ph}_{\Lambda'\setminus \Lambda})
 \simeq
  (\F_{\Lambda} \otimes \F^{\rm ph}_{\Lambda}) \otimes (\F_{\Lambda'\setminus \Lambda} \otimes \F^{\rm ph}_{\Lambda'\setminus \Lambda}) 
 $. Considering the $|\Lambda'|$-electron subspaces on both sides of this equality, we get the following: 
 $
 \h_{\Lambda'}^{\rm HH}=\F_{\Lambda', |\Lambda'|} \otimes \F^{\rm ph}_{\Lambda'}
 =\bigoplus_{m+n=|\Lambda'|} (\F_{\Lambda, m} \otimes \F^{\rm ph}_{\Lambda}) \otimes (\F_{\Lambda'\setminus \Lambda, n} \otimes \F^{\rm ph}_{\Lambda'\setminus \Lambda})
 $.
  Thus, from this equality, we obtain the desired assertion.
 }
Define the isometric linear mapping $\kappa^{\rm HH} : \h^{\rm HH}_{\Lambda} \to \h_{\Lambda\rq{}}^{\rm HH}$
 by $\kappa^{\rm HH} \eta=\eta\otimes \tilde{\xi}^{\rm HH}_{\Lambda'\setminus \Lambda}\, (\eta\in \h_{\Lambda}^{\rm HH})$, where  $\tilde{\xi}^{\rm HH}_{\Lambda}=\tilde{\xi}_{\Lambda} \otimes \Omega^{\rm ph}_{\Lambda}$. Here, recall that 
 $\tilde{\xi}_{\Lambda}$ is given in the proof of Lemma \ref{VolSys}
  and $\Omega^{\rm ph}_{\Lambda}$
  is the Fock vacuum in $\F^{\rm ph}_{\Lambda}$. By identifying $\kappa^{\rm HH} \h^{\rm HH}_{\Lambda}$
 with $\h^{\rm HH}_{\Lambda}$, we can regard $\h^{\rm HH}_{\Lambda}$ as a closed subspace of $\h^{\rm HH}_{\Lambda'}$. In addition, we have the identification: $\M^{\rm HH}_{\Lambda}=\M^{\rm HH}_{\Lambda} \otimes \Pi^{\rm HH}_{\Lambda'\setminus \Lambda}$, where $\Pi_{\Lambda}^{\rm HH}=|\tilde{\xi}^{\rm HH}_{\Lambda}\ra\la \tilde{\xi}^{\rm HH}_{\Lambda}|$.
 
 Set $P^{\rm HH}_{\Lambda\Lambda'}=1\otimes\Pi^{\rm HH}_{\Lambda'\setminus \Lambda}$. Then $P^{\rm HH}_{\Lambda\Lambda'}$ is the orthogonal projection from $\h^{\rm HH}_{\Lambda'}$ to $\h^{\rm HH}_{\Lambda}$.
It is easy to show the following properties: 
\begin{itemize}
\item $\M^{\rm HH}_{\Lambda}=
P_{\Lambda\Lambda'}^{\rm HH}\M_{\Lambda'}^{\rm HH}P_{\Lambda\Lambda'}^{\rm HH}$.
\item $\vphi_{\Lambda'}^{\rm HH} \circ \mathscr{E}_{\Lambda\Lambda'}^{\rm HH}
=\vphi^{\rm HH}_{\Lambda'}$, where $\mathscr{E}^{\rm HH}_{\Lambda\Lambda'}(x)
=P_{\Lambda\Lambda'}^{\rm HH} x P_{\Lambda\Lambda'}^{\rm HH}\, (x\in \M_{\Lambda'}^{\rm HH})$.
\end{itemize}
All that remains is to show the following:
\be
P_{\Lambda\Lambda'}^{\rm HH}L^2(\M_{\Lambda'}^{\rm HH}, \vphi_{\Lambda'}^{\rm HH})_+=L^2(\M_{\Lambda}^{\rm HH}, \vphi_{\Lambda}^{\rm HH})_+. \label{LasBossHH}
\ee
Before starting on this proof, we need to do a little preparation.
In general, each  element $\eta$ in $L^2(\M_{\Lambda}^{\rm HH}, \vphi_{\Lambda}^{\rm HH})$ can be represented as 
\be
\eta=\sum_{{X, Y\subseteq \Lambda}\atop{|X|=|Y|}}  |X, \overline{Y}\ra_{\Lambda} \otimes\eta_{XY},\ \ 
\eta_{XY} \in \F^{\rm ph}_{\Lambda}.
\ee
Using the identification \eqref{IdnLF}, we can regard $\eta_{XY}$ as an element in $L^2(\BbbR^{|\Lambda|})$.
Therefore, $\eta\in L^2(\M_{\Lambda}^{\rm HH}, \vphi_{\Lambda}^{\rm HH})_+$, if and only if, $\{\eta_{XY}'({\bs q})\}$ is positive semi-definite for  a.e. $\bs q$,
where $\eta'_{XY}({\bs q})=\eta_{XY}({\bs q})$ if $|X|=|Y|$,   $\eta'_{XY}({\bs q})=0$ if $|X|\neq |Y|$.

Now suppose that $\eta\in L^2(\M_{\Lambda\rq{}}^{\rm HH}, \vphi_{\Lambda\rq{}}^{\rm HH})_+$.
The action of   $P_{\Lambda\Lambda'}^{\rm HH}$ on the vector $\eta$ is as follows: 
\be
P_{\Lambda\Lambda'}^{\rm H}\eta=\Bigg[ \sum_{{X_{\Lambda}, Y_{\Lambda} \subseteq \Lambda}\atop{|X_{\Lambda}|=|Y_{\Lambda}|}} 
 |X_{\Lambda}, \overline{Y}_{\Lambda}\ra_{\Lambda}\otimes \tilde{\eta}_{X_{\Lambda}, Y_{\Lambda}} \Bigg] \otimes \tilde{\xi}^{\rm HH}_{\Lambda'\setminus \Lambda},
\ee
where
\be
\tilde{\eta}_{X_{\Lambda}, Y_{\Lambda}}({\bs q}_{\Lambda}) =\sum_{X_{\Lambda'\setminus \Lambda}\subseteq \Lambda'\setminus \Lambda} \int_{\BbbR^{\Lambda'\setminus \Lambda}} d{\bs q}_{\Lambda' \setminus \Lambda}
\eta_{X_{\Lambda} \sqcup X_{\Lambda'\setminus \Lambda}, Y_{\Lambda} \sqcup X_{\Lambda'\setminus \Lambda}}({\bs q}_{\Lambda}, {\bs q}_{\Lambda'\setminus \Lambda}) \Omega_{\Lambda\rq{}\setminus \Lambda}^{\rm ph}({\bs q}_{\Lambda\rq{}\setminus \Lambda})^2.
\ee
Here, recall that under the identification \eqref{IdnLF}, we already know that $\Omega^{\rm ph}_{\Lambda} ({\bs q})=(\omega/|\Lambda|)^{|\Lambda|/4} e^{-\omega {\bs q}^2/2}$.
Define $\tilde{\eta}'_{X_{\Lambda}, Y_{\Lambda}}({\bs q})$ in the same way as before.
Since $\Omega^{\rm ph}_{\Lambda}$ is a Gaussian function, $\Omega_{\Lambda}^{\rm ph}({\bs q})$ is strictly positive almost everywhere. In addition, $\{\eta'_{XY}({\bs q})\}$ is  positive semi-definite for a.e. $\bs q$.
Hence, $\{\tilde{\eta}'_{X_{\Lambda}, Y_{\Lambda}}({\bs q}_{\Lambda})\}$ is  positive semi-definite for a.e. ${\bs q}_{\Lambda}$, which implies that $P_{\Lambda\Lambda'}^{\rm H}\eta \in L^2(\M_{\Lambda}^{\rm HH}, \vphi_{\Lambda}^{\rm HH})_+$. Thus, we conclude that $P_{\Lambda\Lambda'}^{\rm HH}L^2(\M_{\Lambda'}^{\rm HH}, \vphi_{\Lambda'}^{\rm HH})_+\subseteq L^2(\M_{\Lambda}^{\rm HH}, \vphi_{\Lambda}^{\rm HH})_+$.
The inverse inclusion relation can be easily shown. 

(ii) 
Note that $ Q_{\Lambda}^{\rm HH}$ can be espressed as $Q_{\Lambda}^{\rm HH}=1\otimes |\Omega^{\rm ph}_{\Lambda}\ra\la\Omega^{\rm ph}_{\Lambda}|$.
We readily confirm the following properties:
\begin{itemize}
\item $Q_{\Lambda}^{\rm HH} L^2(\M_{\Lambda}^{\rm HH}, \vphi_{\Lambda}^{\rm HH})=L^2(\M_{\Lambda}^{\rm H}, \vphi_{\Lambda}^{\rm H})$.
\item $
Q_{\Lambda}^{\rm HH} \M_{\Lambda}^{\rm HH} Q_{\Lambda}^{\rm HH}=\M^{\rm H}_{\Lambda}.
$
\item $Q_{\Lambda}^{\rm HH} L^2(\M_{\Lambda}^{\rm HH}, \vphi_{\Lambda}^{\rm HH})_+=L^2(\M_{\Lambda}^{\rm H}, \vphi_{\Lambda}^{\rm H})_+$
\end{itemize}
The third property can be shown using the fact that $\Omega_{\Lambda}^{\rm ph}({\bs q})$ is strictly  positive for all ${\bs q} \in \BbbR^{\Lambda}$.
\end{proof}

Lemma \ref{ConHoHu} allows us to define the magnetic system $O^{\rm HH}_{1/2}$
 by $O^{\rm HH}_{1/2}=\big\{
\{\M_{\Lambda}^{\rm HH}, \vphi^{\rm HH}_{\Lambda}\} : \Lambda\in F_{\mathbb{L}}
\big\}$.

The following condition is necessary to state the results:
\begin{Assum}\label{AssG}\upshape
For each $\Lambda\in F_{\mathbb{L}}^{\rm (e)}$, $\sum_{y\in \Lambda}g_{xy}$ is independent of $x\in \Lambda$.\footnote{This sum can be dependent on $\Lambda$.}
\end{Assum}

We are now ready to present our results on the Holstein--Hubbard model.

\begin{Thm}\label{MLMHHThm} Assume Conditions \ref{AssHubbard} and \ref{AssG}. 
Let 
\be
U_{{\rm eff}, xy}=U_{xy}-\frac{2}{\omega} \sum_{z\in \Lambda} g_{xz} g_{yz}. \label{DefUeff}
\ee
Assume that, for each  $\Lambda\in F^{\rm (e)}_{\mathbb{L}}$, $\{U_{{\rm eff}, xy}\}_{x, y\in \Lambda}$ is a positive definite matrix.

 We have the following:
\begin{itemize}
\item[\rm (i)] For each $\Lambda\in F^{\rm (e)}_{\mathbb{L}}$, $H_{\Lambda}^{\rm HH}\in \mathscr{A}_{\Lambda, |\Lambda|}(\M_{\Lambda}^{\rm MLM}, \vphi_{\Lambda}^{\rm MLM})$ holds. Hence,  $S_{H_{\Lambda}^{\rm HH}}=\big||\Lambda_A|-|\Lambda_B|\big|\big/2$.
 \item[\rm (ii)] $O_{1/2}^{\rm HH} \LRA O_{1/2}^{\rm H} \LRA O_{1/2}^{\rm MLM}$. Hence, $O_{1/2}^{\rm HH} \in \mathscr{C}(O_{1/2}^{\rm MLM})$.
 \item[\rm (iii)] Set ${\bs H}^{\rm HH}=\{H_{\Lambda} ^{\rm HH} : \Lambda \in F^{\rm (e)}_{\mathbb{L}}\}$. Then ${\bs H}^{\rm HH}$ is adapted to $O_{1/2}^{\rm HH}$.  Hence, 
 if there exists a non-negative constant $s$ and an increasing sequence of sets $\{\Lambda_n : n\in \BbbN\}\subset  F^{\rm (e)}_{\mathbb{L}}$ such that $\bigcup_{n=1}^{\infty} \Lambda_n=\mathbb{L}$ and 
$\big||\Lambda_{n, A}|-|\Lambda_{n, B}|\big|=2s|\Lambda_n|+o(|\Lambda_n|)$ as $n\to \infty$, then each NMGS  associated with $H_{\Lambda}^{\rm HH}$ exhibits a  magnetic order with a  spin density of $s$.
\end{itemize}
\end{Thm}
\begin{proof}
Choose   $\Lambda\in F^{\rm (e)}_{\mathbb{L}}$ arbitrarily. 
By  (ii) of Lemma \ref{ConHoHu}, we readily confirm that $O^{\rm HH}_{1/2} \LRA O^{\rm H}_{1/2}$. Combining this with (ii) of Theorem \ref{MLMHubbThm}, we get (ii) of Theorem \ref{MLMHHThm}.
By  applying (iv) of Theorem \ref{ManyHa},  we see that 
$\{e^{-\beta H_{\Lambda, M}^{\rm HH}}\}_{\beta \ge 0} $ is ergodic  w.r.t. $L^2(\M_{\Lambda}^{\rm HH}[M], \vphi^{\rm HH}_{\Lambda, M})_+$ for all  $ M\in \mathrm{spec}(S_{\Lambda}^{(3)} \restriction \h_{\Lambda})$. 
Hence, $H^{\rm HH}_{\Lambda} \in A_{\Lambda, |\Lambda|}(\M^{\rm HH}_{\Lambda}, \vphi^{\rm HH}_{\Lambda})$, which implies that ${\bs H}^{\rm HH}$
is adapted to $O^{\rm HH}_{1/2}$. By applying (ii) of Lemma \ref{ConHoHu}, we find that $H^{\rm HH}_{\Lambda}\in \mathscr{A}_{\Lambda, |\Lambda|}(\M^{\rm MLM}_{\Lambda}, \vphi^{\rm MLM}_{\Lambda})$.
Finally, using Theorem \ref{BasicMLM}, we conclude (iii) of Theorem \ref{MLMHHThm}.
\end{proof}

\begin{Exa}\upshape
A typical example of satisfying Condition \ref{AssG} is the case where $g_{xy}=g\delta_{xy}$.
If the energy of the  Coulomb interaction is of the form  $U_{xy}=U\delta_{xy}$,  then the energy of the effective Coulomb interaction is given  by $U_{{\rm eff}, xy}=(U-2g^2/\omega)\delta_{xy}$.
In this case, $\{U_{{\rm eff}, xy}\}$ is positive definite, provided that  $|g|<\sqrt{\omega U/2}$.
For more examples of $g_{xy}$ satisfying Condition \ref{AssG}, see \cite{Miyao2016}.
\end{Exa}

According to Theorem \ref{MLMHHThm}, if the electron-phonon interaction is not too strong compared to the Coulomb interaction, the Holstein--Hubbard model belongs to
$ \mathscr{A}_{\Lambda, |\Lambda|}(\M_{\Lambda}^{\rm MLM}, \vphi_{\Lambda}^{\rm MLM})$, and consequently, the magnetic properties of the ground states are described by the MLM stability class $\mathscr{C}(O_{1/2}^{\rm MLM})$.
Recalling the fact proved in Theorem \ref{MLMHubbThm}, that the MLM stability class also describes the magnetic properties of the ground states of the Hubbard model, Theorem \ref{MLMHHThm} shows that the magnetic properties of the ground states of the Hubbard model are stable even under the electron-phonon interaction.

\subsection{The Kondo lattice model}

In this subsection, we discuss  the magnetic properties of the ground states of the Kondo lattice model. The stability of the magnetic properties of the ground states of this system can be explained using a deformed MLM class. As an example, the Kondo  lattice system interacting with lattice vibrations will be examined in detail.
The results on the stability of the magnetic properties of the Kondo system are, of course, significant in themselves, but more importantly, a variant of the MLM class naturally appears in the analysis of strongly correlated electron systems. This fact implies that the discovery of various variants of the MLM class is essential for the understanding of strongly correlated electron systems.

As before, let $G=(\mathbb{L}, E)$ be a connected bipartite infinite graph. Fix $\Lambda \in F_{\mathbb{L}}$,
arbitrarily.
The Kondo lattice model (KLM) on $\Lambda$ is given by 
\be
H_{\Lambda}^{\rm K}=\sum_{x, y\in\Lambda} \sum_{\sigma=\up, \down}t_{xy} c_{x\sigma}^*c_{y\sigma}+J\sum_{x\in \Lambda}{\bs S}^c_x\cdot {\bs S}_x^f+\sum_{x, y\in \Lambda} \frac{U_{xy}}{2} (n_x^c-1)(n_y^c-1),
\ee
where ${\bs S}^c_x=(S_x^{c, (1)}, S_x^{c, (2)}, S_x^{c, (3)})$ are the spin operators at vertex $x$ associated with the conduction electrons, namely, defined by \eqref{DefSpinC}; ${\bs S}^f_x=(S_x^{f, (1)}, S_x^{f, (2)}, S_x^{f, (3)})$ are the spin operators at vertex $x$ associated with the localized spins:
\be
S_x^{f, (j)}=\frac{1}{2} \sum_{\sigma, \sigma\rq{}=\up, \down}f_{x\sigma}^* \big(s^{(j)}\big)_{\sigma \sigma\rq{}} f_{x\sigma\rq{}},\ \ j=1, 2, 3.
\ee
Here, $f_{x\sigma}$ denotes the annihilation operator for the localized spins which obeys the following anti-commutation relations:
\be
\{f_{x\sigma}, f_{x\rq{}\sigma\rq{}}^*
\}=\delta_{xx\rq{}} \delta_{\sigma\sigma\rq{}},\ \ \{f_{x\sigma}, f_{x\rq{}\sigma\rq{}}
\}=0.
\ee
 Remark that the operators $c_{x\sigma}$ and $f_{x\sigma}$ satisfy
\be
\{c_{x\sigma}, f_{x\rq{}\sigma\rq{}}\}=0=\{c_{x\sigma}, f_{x\rq{}\sigma\rq{}}^*\}.
\ee
The Hamiltonian $H_{\Lambda}^{\rm K}$ acts in $\h_{\Lambda}^{\rm K}
=P_{0, \Lambda}\bigwedge^{2|\Lambda|} \big(
\ell^2(\Lambda \sqcup \Lambda)\oplus \ell^2(\Lambda \sqcup \Lambda) 
\big),$ where $\Lambda\sqcup \Lambda$ stands for the discriminated union of $\Lambda$ and $\Lambda$, and 
\be
P_{0, \Lambda} =\prod_{x\in \Lambda} (
n_{x\up}^f -n_{x\down}^f
)^2.
\ee
Here, $n_{x\sigma}^f$ stands for the number operator for the localized spins: $n_{x\sigma}^f=f_{x\sigma}^*f_{x\sigma}$.
Note that $n_{x\up}^f+n_{x\down}^f=1$ holds on $\h^{\rm K}_{\Lambda}$.
The total spin operators are given by 
\be
S_{\rm tot, \Lambda}^{(j)}=S_{\Lambda}^{c, (j)}+S_{\Lambda}^{f, (j)},\ \ j=1, 2, 3,
\ee
where $S_{\Lambda}^{f, (j)}=\sum_{x\in \Lambda} S_x^{f, (j)}$.
In addition, we set
${\bs S}_{\rm tot, \Lambda}^2=
(S_{\rm tot, \Lambda}^{(1)})^2+(S_{\rm tot, \Lambda}^{(2)})^2+(S_{\rm tot, \Lambda}^{(3)})^2
$.

To describe the properties of the ground states of the KLM, we introduce some definitions.
Let $\mathbb{L}^c$ and $\mathbb{L}^f$ be two copies of $\mathbb{L}$. We denote by $\BbbL^c\sqcup \BbbL^f$ be the discriminated union of $\BbbL^c$ and $\BbbL^f$:
$
\BbbL^c\sqcup \BbbL^f=\{(x, c), (x, f) : x\in \BbbL\} 
$.
In the Kondo lattice system, the properties of the ground states differ greatly depending on the sign of the coupling constant $J$.
In order to express this mathematically, we define two graphs with $\BbbL^c\sqcup \BbbL^f$ as the vertex set.
Let us start with the definition of the graph $G^{\rm AF}$, which describes the properties of the ground states when the coupling constant is antiferromagnetic, i.e., $J>0$.
The edge set,  $\mathsf{E}^{\rm AF}$,  of $G^{\rm AF}$ is defined by
\begin{align}
\{(x, \tau), (x\rq{}, \tau\rq{})\} \in \mathsf{E}^{\rm AF} 
\Longleftrightarrow\begin{cases}
\{x, x\rq{}\} \in E & \mbox{if $\tau=\tau\rq{}=c$}\\
x=x\rq{} & \mbox{if $\tau \neq \tau\rq{}$}. \label{DefEF}
\end{cases}
\end{align}

Next, let us define the graph $G^{\rm F}$ to be used when the coupling constant is ferromagnetic, i.e., $J<0$.
The edge set $\mathsf{E}^{\rm F}$ of $G^{\rm F}$ can be characterized as follows: 
\begin{align}
\{(x, \tau), (x\rq{}, \tau\rq{})\} \in \mathsf{E}^{\rm F}
\Longleftrightarrow
 \mbox{$\{x, x\rq{}\}\in E$ if $\tau=\tau\rq{}=c$ or $\tau \neq \tau\rq{}$}. \label{DefEAF}
\end{align}
It is easy to see that both $G^{\rm AF}$ and $G^{\rm F}$ are connected graphs.
Note that each of these graphs is given the following bipartite structure: 
\begin{align}
{\rm AF} :\   & \BbbL^c\sqcup \BbbL^f=(\BbbL_A^c\sqcup \BbbL_B^f) \sqcup (\BbbL_B^c\sqcup \BbbL^f_A)\label{AFBI},\\
{\rm F} : \  & \BbbL^c\sqcup \BbbL^f=(\BbbL_A^c\sqcup \BbbL_A^f) \sqcup (\BbbL_B^c\sqcup \BbbL^f_B).\label{FBI}
\end{align}

For a given $\Lambda\in F_{\mathbb{L}}^{\rm (e)}$, we can define the subgraphs $G_{\Lambda}^{\sharp}=(\Lambda\sqcup \Lambda, \mathsf{E}_{\Lambda}^{\sharp})\, (\sharp={\rm AF, F})$ of the graphs $G^{\sharp}$, where the edge sets, $\mathsf{E}_{\Lambda}^{\sharp}$, are defined by replacing $E$ in   \eqref{DefEF} and \eqref{DefEAF}  by $E_{\Lambda}$.  $G_{\Lambda}^{\sharp}\, (\sharp={\rm AF, F})$ are also  connected bipartite graphs, and their bipartite structures are given by the following:
\begin{align}
{\rm AF} :\   & \Lambda^c\sqcup \Lambda^f=(\Lambda_A^c\sqcup \Lambda_B^f) \sqcup (\Lambda_B^c\sqcup \Lambda^f_A) =: ( \Lambda^c\sqcup \Lambda^f)_{1, {\rm AF}}\sqcup ( \Lambda^c\sqcup \Lambda^f)_{2, {\rm AF}},\label{BipA}\\
{\rm F} : \  & \Lambda^c\sqcup \Lambda^f=(\Lambda_A^c\sqcup \Lambda_A^f) \sqcup (\Lambda_B^c\sqcup \Lambda^f_B)=: ( \Lambda^c\sqcup \Lambda^f)_{1, {\rm F}}\sqcup ( \Lambda^c\sqcup \Lambda^f)_{2, {\rm F}},\label{BipB}
\end{align}
where $\Lambda^c$ and $\Lambda^f$ are copies of $\Lambda$.

The bipartite structures give rise to  the abelian von Neumann algebra 
$\N_{\Lambda}^{\rm MLM}$ and the   faithful semi-finite normal weights $\vphi^{\rm MLM}_{\Lambda, \sharp}\ (\sharp={\rm AF, F})$ as follows.
For simplicity of notation, we set 
\be
a_{u\sigma}=\begin{cases}
c_{x\sigma} & \mbox{for $u=(x, c) \in \BbbL^c\sqcup\BbbL^f$}\\
f_{x\sigma} & \mbox{for $u=(x, f) \in \BbbL^c\sqcup\BbbL^f$. }
\end{cases}
\ee
Set 
\be
|\Omega\ra_{\Lambda}=\Bigg[\prod'_{u\in \Lambda^c\sqcup \Lambda^f} a_{u\down}^*\Bigg]|\varnothing\ra_{\Lambda^c\sqcup \Lambda^f},
\ee
where $|\varnothing\ra_{\Lambda^c\sqcup \Lambda^f}$ is the fermionic Fock vacuum in  $\F_{\Lambda^c\sqcup \Lambda^f }$, and $\prod'_{u\in \Lambda^c\sqcup \Lambda^f}$ indicates the product taken over all vertex in $\Lambda^c\sqcup \Lambda^f$ with an arbitrarily fixed order.
Given $U, V\in \Lambda^c\sqcup \Lambda^f$, define
\be
|U, V\ra_{\sharp}=(-1)^{|V\cap (\Lambda^c\sqcup \Lambda^f)_{2, \sharp}|}\Bigg[\prod_{u\in U}a_{u\up}^* \Bigg] \Bigg[\prod_{v\in \overline{V}}a_{v\down} \Bigg]|\Omega\ra,\ \ \sharp={\rm AF, F}.
\ee
Then $\{|U, \overline{U}\ra_{\sharp} : U\subseteq \Lambda^c\sqcup \Lambda^f\}$ are CONSs of $\h_{\Lambda^c\sqcup \Lambda^f}$, where $\h_{\Lambda^c\sqcup \Lambda^f}$ is defined by \eqref{HSingE} with $\Lambda$ replaced by $\Lambda^c\sqcup \Lambda^f$. Define the  vectors $\tilde{\xi}_{\Lambda, \sharp}\in \h_{\Lambda^c\sqcup \Lambda^f}$ by 
\be
\tilde{\xi}_{\Lambda, \sharp}=\sum_{U\subseteq \Lambda^c\sqcup \Lambda^f}  |U, \overline{U}\ra_{\sharp}\ \ (\sharp={\rm AF, F}). \label{TaneVec2}
\ee
As in Subsection \ref{GeneMLM}, we set $\vphi_{\Lambda, \sharp}^{\rm MLM}=\vphi_{\tilde{\xi}_{\Lambda, \sharp}}$. We denote by $\N_{\Lambda}^{\rm MLM}$ the abelian von Neuman algebra generated by diagonal operators associated with  $\{|U, \overline{U}\ra_{\sharp} : U\subseteq \Lambda^c\sqcup \Lambda^f\}$.
Note that $\N_{\Lambda}^{\rm MLM}$ does not depend on $\sharp$, since the definition of $|U, \overline{U}\ra_{\sharp}$ is only different in sign when $\sharp=\rm AF$ and when $\sharp=\rm F$.
The  modified MLM systems:
\be 
O_{1/2, \sharp}^{\rm MLM}:=\big\{
\{\N_{\Lambda}^{\rm MLM}, \vphi_{\Lambda, \sharp}^{\rm MLM}\} : 
\Lambda \in  F^{\rm (e)}_{\mathbb{L}}
\big\}\ (\sharp={\rm AF}, {\rm F})
\ee
 play  essential roles to examine magnetic properties of the Kondo lattice systems.

\begin{Lemm}\label{VolSysKLM}
If $\Lambda, \Lambda'\in  F^{\rm (e)}_{\mathbb{L}}$ satisfies $\Lambda\subseteq \Lambda'$, then 
$F(\N_{\Lambda\rq{}}^{\rm MLM}, \vphi^{\rm MLM}_{\Lambda\rq{}, \sharp})\LRA F(\N_{\Lambda}^{\rm MLM}, \vphi^{\rm MLM}_{\Lambda, \sharp})$ holds for $\sharp ={\rm AF}, {\rm F}$. Hence, the nets $O_{1/2, \sharp}^{\rm MLM}$ are magnetic systems.
\end{Lemm}
\begin{proof}
Note that using the symbols from Subsection \ref{GeneMLM}, we can express $\N^{\rm MLM}_{\Lambda}=\M^{\rm MLM}_{\Lambda\sqcup \Lambda}$.
Therefore, if we change the considered graph in the proof of Lemma \ref{VolSys} from $G$ to $G^{\sharp}$, we can derive the assertion of the lemma.
\end{proof}

Corresponding to \eqref{DefRalg}, we set 
\be
\mathfrak{R}_{\Lambda^c\sqcup \Lambda^f}=
\big\{
a_{u_1\up}^*\cdots a_{u_n\up}^*a_{v_1\down}\cdots a_{v_n\down} : u_1, \dots, u_n, v_1,\dots, v_n\in \Lambda^c\sqcup \Lambda^f,\ n\in \{1, \dots, 2|\Lambda|\}
\big\}.
\ee
Define $\M_{\Lambda}^{\rm K}=P_{0, \Lambda} (\mathfrak{R}_{\Lambda^c\sqcup \Lambda^f})^{\prime\prime}P_{0, \Lambda}$. 
Note that unlike $\M^{\rm H}_{\Lambda}$ defined in Subsection \ref{MLM-H}, the orthogonal projection operator $P_{0, \Lambda}$ is included in the definition of $\M_{\Lambda}^{\rm K}$. This difference will become important later.
The states $\vphi_{\Lambda, \sharp}^{\rm MLM}\, (\sharp={\rm AF, F})$ can be naturally extended to $\M_{\Lambda}^{\rm K}$. We denote  by $\vphi_{\Lambda, \sharp}^{\rm K}$ the extensions.
The following nets of the IEE systems:
\be
O^{\rm K}_{1/2, \sharp}:=\big\{
\{ \M_{\Lambda}^{\rm K}, \vphi_{\Lambda, \sharp}^{\rm K}\} : 
\Lambda\in   F^{\rm (e)}_{\mathbb{L}}
\big\}\ (\sharp ={\rm AF}, {\rm F})
\ee
 are   also useful for characterizing  the magnetic properties of the ground states  of  the Kondo lattice model.

\begin{Lemm}\label{ConKLM}
We have the following:
\begin{itemize}
\item[\rm (i)] If $\Lambda, \Lambda'\in F^{\rm (e)}_{\mathbb{L}}$ satisfies $\Lambda\subseteq \Lambda'$, then 
$F(\M^{\rm K}_{\Lambda\rq{}}, \vphi^{\rm K}_{\Lambda\rq{}, \sharp})\longrightarrow F(\M^{\rm K}_{\Lambda}, \vphi^{\rm K}_{\Lambda, \sharp})$ for $\sharp={\rm AF}, {\rm F}$.
Hence, $O^{\rm K}_{1/2, \sharp}$ are magnetic systems.
\item[\rm (ii)] For each 
 $\sharp={\rm AF}, {\rm F}$ and $\Lambda 
\in F^{\rm (e)}_{\mathbb{L}}
$,  $F(\M^{\rm K}_{\Lambda}, \vphi^{\rm K}_{\Lambda, \sharp})\longrightarrow F(\N^{\rm MLM}_{\Lambda}, \vphi^{\rm MLM}_{\Lambda, \sharp})$.
\end{itemize}
The above results can be summarized as the following commutative diagram:
\be
\begin{tikzcd}
  F(\M^{\rm K}_{\Lambda}, \vphi^{\rm K}_{\Lambda, \sharp})   \arrow[d] &
  \ar[l]   F(\M^{\rm K}_{\Lambda\rq{}}, \vphi^{\rm K}_{\Lambda\rq{}, \sharp})
  \arrow[d]  \\
  F(\N^{\rm MLM}_{\Lambda}, \vphi^{\rm MLM}_{\Lambda, \sharp}) &\arrow[l] F(\N^{\rm MLM}_{\Lambda\rq{}}, \vphi^{\rm MLM}_{\Lambda\rq{}, \sharp})
\end{tikzcd}
\ee
Furthermore, when $\{\M_{\Lambda}, \vphi_{\Lambda}\}$  is replaced by $\{\N^{\rm MLM}_{\Lambda}, \vphi^{\rm MLM }_{\Lambda, \sharp}\}$ and $\{\M^{\prime}_{\Lambda}, \vphi^{\prime}_{\Lambda}\}$ is replaced by $\{\M^{\rm K}_{\Lambda}, \vphi^{\rm K}_{\Lambda, \sharp}\}$ in Theorem \ref{MacroDiagram}, the diagram \eqref{MacDiagComm} is commutative.
\end{Lemm}
\begin{proof}
The proof is almost the same as that of Lemma \ref{ConHubb}. However, for the convenience of the reader, a few remarks are in order.

First, note the following relations:
  \be
  \h_{\Lambda}^{\rm K}=P_{0, \Lambda}\h^{\rm H}_{\Lambda^c\sqcup \Lambda^f},\ \ \M^{\rm K}_{\Lambda}=P_{0, \Lambda} \M^{\rm H}_{\Lambda^c\sqcup \Lambda^f}P_{0, \Lambda}, 
\, \ \ L^2(\M^{\rm K}_{\Lambda}, \vphi^{\rm K}_{\Lambda, \sharp})_+=P_{0, \Lambda} L^2(\M^{\rm H}_{\Lambda^c\sqcup \Lambda^f}, \vphi_{\xi_{\Lambda, \sharp}})_+.  
  \ee
  Take $\Lambda, \Lambda'\in F_{\mathbb{L}}$ with $\Lambda\subset \Lambda'$.
  We set ${\bs P}_{\Lambda\Lambda'}:=P^{\rm H}_{\Lambda\sqcup \Lambda, \Lambda'\sqcup \Lambda'}$.
  Then $\bs P_{\Lambda\Lambda'}$ is the orthogonal projection from $L^2(\M_{\Lambda'\sqcup \Lambda'}^{\rm H}, \vphi_{\xi_{\Lambda'}, \sharp})$ to $L^2(\M_{\Lambda\sqcup \Lambda}^{\rm H}, \vphi_{\xi_{\Lambda}, \sharp})$. 
  Using arguments similar to those in the proof of (i) of Lemma \ref{ConHubb}, we know  that 
  \be
  {\bs P}_{\Lambda\Lambda'}L^2(\M_{\Lambda'\sqcup \Lambda'}^{\rm H}, \vphi_{\xi_{\Lambda'}, \sharp})_+=
  L^2(\M_{\Lambda\sqcup \Lambda}^{\rm H}, \vphi_{\xi_{\Lambda}, \sharp})_+. \label{KL2Diff}
  \ee
  Let $P^{\rm K}_{\Lambda\Lambda'}$ be the orthogonal projection from $L^2(\M^{\rm K}_{\Lambda'}, \vphi^{\rm K}_{\Lambda', \sharp})$ to $L^2(\M^{\rm K}_{\Lambda}, \vphi^{\rm K}_{\Lambda, \sharp})$.
 We readily confirm that 
 $P^{\rm K}_{\Lambda\Lambda'}={\bs P}_{\Lambda\Lambda'}\restriction \h^{\rm K}_{\Lambda'}$ and $
 P_{0, \Lambda} {\bs P}_{\Lambda\Lambda'}={\bs P}_{\Lambda\Lambda'}P_{0, \Lambda'}
 $.
 Combining these with \eqref{KL2Diff}, we obtain 
 \be
 P^{\rm K}_{\Lambda\Lambda'} L^2(\M^{\rm K}_{\Lambda'}, \vphi^{\rm K}_{\Lambda', \sharp})_+=
 L^2(\M^{\rm K}_{\Lambda}, \vphi^{\rm K}_{\Lambda, \sharp})_+.
 \ee 
 This completes the proof of (i) of Lemma \ref{ConKLM}.

Next, 
let $Q_{\Lambda}^{\rm K}$ be the orthogonal projection from $\h_{\Lambda}^{\rm K}$ to $\h_{\Lambda^c\sqcup \Lambda^f}$: $Q_{\Lambda}^{\rm K}=\prod_{x\in \Lambda} (n_{x\up}^c-n_{x\down}^c)^2$. Then we readily confirm that 
\be
Q_{\Lambda\sqcup \Lambda}=Q_{\Lambda}^{\rm K} P_{0, \Lambda}.
\ee
From the proof of (ii) of Lemma \ref{ConHubb}, we find
\begin{align}
Q_{\Lambda}^{\rm K} \M_{\Lambda}^{\rm K} Q_{\Lambda}^{\rm K}
=Q_{\Lambda^c\sqcup \Lambda^f}  (\mathfrak{R}_{\Lambda^c\sqcup \Lambda^f})^{\prime\prime}Q_{\Lambda^c\sqcup \Lambda^f} =\M^{\rm MLM}_{\Lambda^c\sqcup \Lambda^f}=\N^{\rm MLM}_{\Lambda}.
\end{align}
and 
\be
Q_{\Lambda\sqcup \Lambda}L^2(\M_{\Lambda\sqcup \Lambda}, \vphi_{\xi_{\Lambda}, \sharp})_+
=L^2(\M^{\rm MLM}_{\Lambda\sqcup \Lambda}, \vphi_{\xi_{\Lambda}, \sharp})_+.\label{QL2K}
\ee
Combining these with the property $L^2(\M_{\Lambda}^{\rm K}, \vphi^{\rm K}_{\Lambda, \sharp})_+=P_{0, \Lambda} L^2(\M^{\rm H}_{\Lambda\sqcup \Lambda}, \vphi_{\xi_{\Lambda}, \sharp})_+$,  we obtain
\begin{align}
Q^{\rm K}_{\Lambda} L^2(\M_{\Lambda}^{\rm K}, \vphi^{\rm K}_{\Lambda, \sharp})_+&=P_{0, \Lambda} Q^{\rm K}_{\Lambda} L^2(\M^{\rm H}_{\Lambda\sqcup \Lambda}, \vphi_{\xi_{\Lambda}, \sharp})_+\no
&=Q_{\Lambda\sqcup \Lambda} L^2(\M^{\rm H}_{\Lambda\sqcup \Lambda}, \vphi_{\xi_{\Lambda}, \sharp})_+
=L^2(\N^{\rm MLM}_{\Lambda}, \vphi^{\rm MLM}_{\Lambda, \sharp})_+.
\end{align}
This completes the proof of (ii) of Lemma \ref{ConKLM}.
\end{proof}

The basic results for the KLM are summarized in the following theorem.
\begin{Thm} \label{KondoThm1}
Assume Condition \ref{AssHubbard}.  Assume that, for each $\Lambda \in F_{\mathbb{L}}^{({\rm e})}$. the matrix $\{U_{xy}\}_{x, y\in \Lambda}$ is positive   semi-definite. Then we have the following:
\begin{itemize}
\item[\rm (i)] If $J>0$, the antiferromagnetic coupling, then 
$H_{\Lambda}^{\rm K} \in \mathscr{A}_{\Lambda^c\sqcup \Lambda^f, 2|\Lambda|}(\N^{\rm MLM}_{\Lambda}, \vphi^{\rm MLM}_{\Lambda, {\rm AF}})$ holds.
Hence, $S_{H_{\Lambda}^{\rm K}}=0$.
\item[\rm (ii)] If $J<0$, the ferromagnetic coupling, then 
$H_{\Lambda}^{\rm K} \in \mathscr{A}_{\Lambda^c\sqcup \Lambda^f, 2|\Lambda|}(\N^{\rm MLM}_{\Lambda}, \vphi^{\rm MLM}_{\Lambda, {\rm F}})$ holds.
Hence, $S_{H_{\Lambda}^{\rm K}}=\big||\Lambda_A|-|\Lambda_B|\big|$.
\item[\rm (iii)] $O_{1/2, \sharp}^{\rm K} \LRA O_{1/2, \sharp}^{\rm MLM}\ (\sharp={\rm AF, F})$.
\item[\rm (iv)]  Set ${\bs H}^{\rm K}=\{H_{\Lambda} ^{\rm K} : \Lambda \in F_{\mathbb{L}}\}$.
If $J>0$, then ${\bs H}^{\rm K}$ is adapted to $O^{\rm K}_{1/2, {\rm AF}}$. 
For any  increasing sequence of sets $\{\Lambda_n : n\in \BbbN\}\subset  F_{\mathbb{L}}$ such that $\bigcup_{n=1}^{\infty} \Lambda_n=\mathbb{L}$,  each NMGS  associated with ${\bs H}^{\rm K}$ $\{\psi^{\rm K}_{\Lambda_n} : n\in \BbbN \}$ exhibits a  magnetic order with a spin density of $0$.
In contrast to this, if $J<0$, then 
${\bs H}^{\rm K}$ is adapted to $O^{\rm K}_{1/2, {\rm F}}$.
Hence, 
 if there exist a non-negative constant $s$ and an increasing sequence of sets $\{\Lambda_n : n\in \BbbN\}\subset  F_{\mathbb{L}}$ such that $\bigcup_{n=1}^{\infty} \Lambda_n=\mathbb{L}$ and 
$\big||\Lambda_{n, A}|-|\Lambda_{n, B}|\big|=s|\Lambda_n|+o(|\Lambda_n|)$ as $n\to \infty$, then each NMGS   associated with ${\bs H}^{\rm K}$  exhibits a  magnetic order with a spin density of  $s$.
\end{itemize}

\end{Thm}
\begin{proof}

We give only a brief outline of the proof.
When emphasizing that the sign of the coupling constant $J$ in $H_{\Lambda}^{\rm K}$ is positive (resp. negative), we will denote the Hamiltonian $H_{\Lambda}^{\rm K}$ as $H_{\Lambda, \rm AF}^{\rm K}$ (resp. $H_{\Lambda, \rm F}^{\rm K}$).
Considering the bipartite structures \eqref{BipA} and \eqref{BipB}  of the  graphs $G_{\Lambda}^{\sharp}$, we see by Theorem \ref{BasicMLM} that: 
\begin{align}
S(\N^{\rm MLM}_{\Lambda}, \vphi^{\rm MLM}_{\Lambda, \sharp})=
\begin{cases}
0 & \mbox{if $\sharp={\rm AF}$}\\
\big|
|\Lambda_A|-|\Lambda_B|
\big| & \mbox{if $\sharp={\rm F}$}.
\end{cases}
\end{align}
On the other hand, we see that $\{e^{-\beta H_{\Lambda, \sharp, M}^{\rm K}}\}_{\beta \ge 0} $ is ergodic 
w.r.t. $L^2(\N^{\rm MLM}_{\Lambda}[M], \vphi^{\rm MLM}_{\Lambda, \sharp, M})_+$ for all   $M\in \mathrm{spec}(S^{(3)}_{\rm tot, \Lambda} \restriction \h_{\Lambda^c\sqcup \Lambda^f})$ and $\sharp={\rm AF, F}$  by  (v) of Theorem \ref{ManyHa}, where $H_{\Lambda, \sharp, M}^{\rm K}$ is the restriction of $H_{\Lambda, \sharp}^{\rm K}$ to the $M$-subspace. 
 By this fact and Lemma \ref{ConKLM}, $H_{\Lambda, \sharp}^{\rm K}\in \mathscr{A}_{\Lambda^c\sqcup \Lambda^f, 2|\Lambda|}(\N^{\rm MLM}_{\Lambda}, \vphi^{\rm MLM}_{\Lambda, \sharp}) $ follows for $\sharp={\rm AF, F}$. Therefore, by changing the graph from $G$ to $G^{\sharp}$ in Theorem \ref{BasicMLM}, the claims of Theorem \ref{KondoThm1} follow.
\end{proof}

Theorem \ref{KondoThm1} is a generalization of the results in  \cite{Tsunetsugu1997,Yanagisawa1995} by the theory we have constructed in the previous sections.
In contrast to Theorem \ref{MLMHubbThm}, we only assume that $\{U_{xy}\}$ is positive semi-definite in Theorem \ref{KondoThm1}. Accordingly, the case where $U_{xy}=0$ is also included in the theorem.

Next, we discuss the results on the stability of the magnetic properties of the ground states of the Kondo lattice model.
For this purpose, we consider a model in which conduction electrons interact with lattice vibrations:
\be
H^{\rm KH}_{\Lambda}=H_{\Lambda}^{\rm K}+
\sum_{x, y\in \Lambda}
g_{xy} (n_x -1)(b_y^*
 + b_y ) +\sum_{x\in \Lambda}
\omega b_x^*b_x,
\ee
where the bosonic annihilation  operator $b_x$ is given  in Subsection \ref{MLM-HH}. 
As in the case of the Holstein--Hubbard model, $g_{xy}$ represents the strength of the electron-phonon interaction.
We assume that the matrix $\{g_{xy}\}$ is real and symmetric.
The Hamiltonian $H_{\Lambda}^{\rm KH}$ acts in $\h_{\Lambda}^{\rm KH}:=\h_{\Lambda}^{\rm K} \otimes \F_{\Lambda}^{\rm ph}$, where $\F_{\Lambda}^{\rm ph}$ is the bosonic Fock space given in Subsection \ref{MLM-HH}.
By applying the Kato--Rellich theorem \cite[Theorem X.12]{Reed1975}, we can prove that $H_{\Lambda}^{\rm KH}$ is self-adjoint and bounded from below.
 Define 
\be
\M^{\rm KH}_{\Lambda}=\M_{\Lambda}^{\rm K} \otimes \M_{\Lambda}^{\rm ph},\ \ \vphi^{\rm KH}_{\Lambda, \sharp}=\vphi_{\xi_{\Lambda, \sharp}^{\rm KH}},
\ee
where $\xi_{\Lambda, \sharp}^{\rm KH}=\tilde{\xi}_{\Lambda, \sharp}\otimes \Omega_{\Lambda}^{\rm ph}$
with $
\tilde{\xi}_{\Lambda, \sharp}
$ defined by \eqref{TaneVec2} and $\Omega_{\Lambda}^{\rm ph}$ the bosonic Fock vacuum in $\F^{\rm ph}_{\Lambda}$. 
Let us consider the following nets of the IEE systems:
 \be
O_{1/2, \sharp}^{\rm KH}
:=\big\{
\{\M_{\Lambda}^{\rm KH}, \vphi_{\Lambda, \sharp}^{\rm KH}\} : \Lambda\in F_{\BbbL}^{\rm (e)} 
\big\}\ \ (\sharp={\rm AF, F}).
\ee

\begin{Lemm}\label{ConKoPh}
We have the following:
\begin{itemize}
\item[\rm (i)] If $\Lambda, \Lambda'\in F^{\rm (e)}_{\mathbb{L}}$ satisfies $\Lambda\subseteq \Lambda'$, then 
$F(\M^{\rm KH}_{\Lambda\rq{}}, \vphi^{\rm KH}_{\Lambda\rq{}, \sharp})\longrightarrow F(\M^{\rm KH}_{\Lambda}, \vphi^{\rm KH}_{\Lambda, \sharp})$ for $\sharp={\rm AF, F}$. Hence, $O^{\rm KH}_{1/2, \sharp}\, (\sharp={\rm AF, F})$ are magnetic systems.
\item[\rm (ii)] For each $\sharp ={\rm AF, F}$  and $\Lambda 
\in F^{\rm (e)}_{\mathbb{L}}
, $
$F(\M^{\rm KH}_{\Lambda}, \vphi^{\rm KH}_{\Lambda, \sharp})\longrightarrow F(\M^{\rm K}_{\Lambda}, \vphi^{\rm K}_{\Lambda, \sharp})$.
\end{itemize}
The above results can be summarized as the following commutative diagram:
\be
\begin{tikzcd}
  F(\M^{\rm KH}_{\Lambda}, \vphi^{\rm KH}_{\Lambda, \sharp})   \arrow[d] &
  \ar[l]   F(\M^{\rm KH}_{\Lambda\rq{}}, \vphi^{\rm KH}_{\Lambda\rq{}, \sharp})
  \arrow[d]  \\
  F(\M^{\rm K}_{\Lambda}, \vphi^{\rm K}_{\Lambda, \sharp}) &\arrow[l] F(\M^{\rm K}_{\Lambda\rq{}}, \vphi^{\rm K}_{\Lambda\rq{}, \sharp})
\end{tikzcd}
\ee
Furthermore, when $\{\M_{\Lambda}, \vphi_{\Lambda}\}$  is replaced by $\{\M^{\rm K}_{\Lambda}, \vphi^{\rm K }_{\Lambda, \sharp}\}$ and $\{\M^{\prime}_{\Lambda}, \vphi^{\prime}_{\Lambda}\}$ is replaced by $\{\M^{\rm KH}_{\Lambda}, \vphi^{\rm KH}_{\Lambda, \sharp}\}$ in Theorem \ref{MacroDiagram}, the diagram \eqref{MacDiagComm} is commutative.
\end{Lemm}
\begin{proof}
The proof is almost the same as that of Lemma \ref{ConHoHu}.
Note that as in the proof of Lemma \ref{ConKLM}, the orthogonal projection $P_{0, \Lambda}$ must be taken into account adequately.
\end{proof}

With this lemma and Theorem \ref{KondoThm1}, we have the following theorem:
\begin{Thm}\label{ThmKondo}
Assume Conditions \ref{AssHubbard} and \ref{AssG}.  Assume that, for  each $\Lambda\in F_{\mathbb{L}}^{\rm (e)}$. the matrix $\{U_{{\rm eff}, xy}\}_{x, y\in \Lambda}$ is positive semi-definite, where
$U_{{\rm eff}, xy}$ is given by \eqref{DefUeff}. Then we have the following:
\begin{itemize}
\item[\rm (i)] If $J>0$,  then 
$H_{\Lambda}^{\rm KH} \in \mathscr{A}_{\Lambda^c\sqcup \Lambda^f, 2|\Lambda|}(\N^{\rm MLM}_{\Lambda}, \vphi^{\rm MLM}_{\Lambda, {\rm AF}})$ holds.
Hence, $S_{H_{\Lambda}^{\rm KH}}=0$.
\item[\rm (ii)] If $J<0$,  then 
$H_{\Lambda}^{\rm KH} \in \mathscr{A}_{\Lambda^c\sqcup \Lambda^f, 2|\Lambda|}(\N^{\rm MLM}_{\Lambda}, \vphi^{\rm MLM}_{\Lambda, {\rm F}})$ holds.
Hence, $S_{H_{\Lambda}^{\rm KH}}=\big||\Lambda_A|-|\Lambda_B|\big|$.
\item[\rm (iii)] $O_{1/2, \sharp}^{\rm KH} \longrightarrow O_{1/2, \sharp}^{\rm K} \longrightarrow O_{1/2, \sharp}^{\rm MLM}\ (\sharp={\rm AF, F})$. Hence, $O_{1/2, \sharp}^{\rm KH}\in \mathscr{C}(O_{1/2, \sharp}^{\rm MLM})$.
\item[\rm (iv)]  Set ${\bs H}^{\rm KH}=\{H_{\Lambda} ^{\rm KH} : \Lambda \in F_{\mathbb{L}}\}$.
If $J>0$, then ${\bs H}^{\rm KH}$ is adapted to $O^{\rm KH}_{1/2, {\rm AF}}$. 
For any  increasing sequence of sets $\{\Lambda_n : n\in \BbbN\}\subset  F_{\mathbb{L}}$ such that $\bigcup_{n=1}^{\infty} \Lambda_n=\mathbb{L}$,  each NMGS associated with ${\bs H}^{\rm KH}$ exhibits a  magnetic order with a spin density of  $0$.
In contrast to this, if $J<0$, then 
${\bs H}^{\rm KH}$ is adapted to $O^{\rm KH}_{1/2, {\rm F}}$. 
Hence, 
 if there exists a non-negative constant $s$ and an increasing sequence of sets $\{\Lambda_n : n\in \BbbN\}\subset  F_{\mathbb{L}}$ such that $\bigcup_{n=1}^{\infty} \Lambda_n=\mathbb{L}$ and 
$\big||\Lambda_{n, A}|-|\Lambda_{n, B}|\big|=s|\Lambda_n|+o(|\Lambda_n|)$ as $n\to \infty$, then each NMGS associated with ${\bs H}^{\rm KH}$ exhibits a  magnetic order with a  spin density of  $s$.
\end{itemize}

\end{Thm}
\begin{proof} The sketch of the proof is as follows.
We will use the similar notation as in the proof of Theorem  \ref{KondoThm1}.  By (vi) of Theorem  \ref{ManyHa},
we know that  $\{e^{-\beta H_{\Lambda, \sharp, M}^{\rm KH}} \}_{\beta \ge 0}$ is ergodic w.r.t. $L^2(\M^{\rm KH}_{\Lambda}[M], \vphi^{\rm KH}_{\Lambda, \sharp, M})_+$ for all  $M\in \mathrm{spec}(S^{(3)}_{\rm tot, \Lambda} \restriction \h_{\Lambda^c\sqcup \Lambda^f})$ and $\sharp ={\rm AF, F}$, where $H_{\Lambda, \sharp, M}^{\rm KH}$ is the restriction of $H_{\Lambda, \sharp}^{\rm KH}$ to the $M$-subspace. Hence, we have $H_{\Lambda}^{\rm KH} \in A_{\Lambda^c\sqcup \Lambda^f, 2|\Lambda|}(\M^{\rm KH}_{\Lambda}, \vphi^{\rm MLM}_{\Lambda, \sharp})$, which implies that ${\bs H}^{\rm HK}$ is adapted to $O^{\rm KH}_{1/2, {\rm AF}}$ (resp. $O^{\rm KH}_{1/2, {\rm F}}$ ) if $J>0$ (resp. $J<0$).
  By this and Lemma \ref{ConKoPh}, we find  
  $O^{\rm KH}_{1/2, \sharp} \LRA O^{\rm K}_{1/2, \sharp}$ and 
  $H_{\Lambda, \sharp}^{\rm KH} \in \mathscr{A}_{\Lambda^c\sqcup \Lambda^f, 2|\Lambda|}(\N^{\rm MLM}_{\Lambda}, \vphi^{\rm MLM}_{\Lambda, \sharp}) \, (\sharp={\rm AF}, {\rm F})$. Furthermore, by using  (ii) of  Theorem \ref{KondoThm1}, the property (iii) follows.
The rest of the assertions of the theorem can be easily proved by changing the graph considered in Theorem \ref{BasicMLM} from $G$ to $G^{\sharp}$.
\end{proof}

We note that, in contrast to Theorem \ref{MLMHHThm}, we only assume that $\{U_{{\rm eff}, xy}\}$ is positive semi-definite in Theorem \ref{ThmKondo}. Therefore, the case of $U_{{\rm eff}, xy}=0$ is also included in the theorem.

In Theorem \ref{KondoThm1}, we showed that the deformed MLM stability classes $\mathscr{C}(O_{1/2, \sharp}^{\rm K})$ characterize the magnetic properties of the ground states of the Kondo lattice model. On the other hand, Theorem \ref{ThmKondo} implies that the ground states of the Kondo lattice system, in which conduction electrons and phonons are interacting, is also characterized by  $\mathscr{C}(O_{1/2, \sharp}^{\rm K})$. Therefore, we can conclude that the magnetic properties of the ground states of the Kondo lattice model are robust to the interaction of conduction electrons and phonons.

\begin{Rem}\upshape
In \cite{Miyao2021-2}, we consider the Kondo lattice system in a more generalized setting, where localized spins and conduction electrons live on different crystal lattices.
Furthermore, it is discussed  that the same result as Theorem \ref{ThmKondo} holds for the model describing the interaction between the Kondo lattice system and the quantized radiation field.
\end{Rem}

\subsection{On some more stability theorems related to the MLM stability class}

The MLM stability class can characterize the magnetic properties of the ground states of many more models than those discussed in this section.
We will close this section with a short overview of these.

In  \cite{Miyao2019}, the author considers a model of a many-electron system interacting with a quantized radiation field and shows that this model belongs to the MLM stability class. In this sense, the MLM magnetic properties are stable even under the influence of quantized radiation fields.
The analysis of the same system by Giuliani {\it et al.} in \cite{Giuliani2012} using the rigorous renormalization group also suggests the stability of various physical quantities, which is consistent with the results of this paper.

Ueda {\it et al.}  have discussed the magnetic properties of the periodic Anderson model (PAM) in the ground states in \cite{Ueda1992}.
Their results, expressed in the language of the theory developed in this paper, are as follows: the net of the  PAMs is adapted to a magnetic system belonging to the MLM stability class.
Furthermore, as in a deformed KLM case, it can be shown that these magnetic properties of the ground states  of the PAM are stable even when considering the interactions of electrons with lattice vibrations or with the quantized radiation field.
The proof is highly technical, and we will not go into it here. Instead, the details of the proof are given in \cite{Miyao2022-2}.

In \cite{Freericks1995}, Freericks and Lieb analyzed the magnetic properties of the ground states of the Su--Schrieffer--Heeger (SSH) Hamiltonian of polyacetylene.
They showed that among the ground states of the SSH Hamiltonian, there exists one with total spin $S=0$.
Later, in \cite{Miyao2012}, the author proved that the ground state of the SSH Hamiltonian  is unique and  has total spin $S=0$ when considering the Coulomb repulsion between electrons.
This result can be interpreted as the net of  the SSH Hamiltonians is adapted to a particular magnetic system belonging to the MLM stability class.

In \cite{Kubo1990}, Kubo and Kishi obtained the following result on the  Hubbard model.
In the half-filled repulsive model on a bipartite lattice, the charge and the on-site pairing susceptibilities are bounded above by $U$, where $U$ is the strength of the  on-site interaction potential.
From this result, we can conclude that no charge density wave emerges.
  This result can be regarded as an extension of Lieb's theorem discussed in Subsection \ref{MLM-H} to the finite temperature case.
In \cite{Miyao2015}, the author extends the Kubo--Kishi results to systems where electrons interact with lattice vibrations or quantized radiation fields.
In that paper, the author used a probabilistic approach, which at first glance appears to be very different from the approach in this paper.
 However, if we translate the method into the context of the standard forms, we can see that a structure similar to the theory in this paper emerges.
 This point will be discussed in detail in another paper.

\section{The Nagaoka--Thouless stability classes}\label{Sect6}

\subsection{Overview}
In this section, we describe the Nagaoka--Thouless (NT) stability class in detail.
Together with the MLM stability class, this stability class is fundamental in the stability theory of magnetism and is essential for understanding the basic ideas of the theory.
Compared to the MLM stability class, somewhat fewer examples are currently known in the NT stability class. On the other hand, the von Neumann algebras considered in this section are abelian, and it is a good introductory system because the mathematical treatment is gentler than in the previous section.

In 1965, Nagaoka examined a many-electron system with only one hole and found rigorously that the ground state is ferromagnetic when the Coulomb interaction between electrons is very strong.
As Thouless obtained the same result around the same time, this theorem is now called the Nagaoka--Thouless (NT) theorem.
The NT theorem is the first rigorous result on metallic ferromagnetism. 
For a deeper understanding of the results in this section, we will briefly review the history surrounding the NT theorem.
Tasaki reformulated  the NT theorem as it is known today in \cite{Tasaki1989}.
Koller {\it et al.}  analyzed a more realistic many-electron Hamiltonian and showed that the NT theorem still holds for this Hamiltonian \cite{Kollar1996}.
In \cite{Aizenman1990}, Aizenman and Lieb extended the NT theorem to the finite temperature case.
The author proved that NT ferromagnetism is stable under electron-phonon and electron-photon interactions and extended these results to finite temperatures \cite{Miyao2017,Miyao2020-2}.
Another interesting development direction is extending the NT theorem to the ${\rm SU}(n)$ Hubbard model by Katsura and Tanaka \cite{Katsura2013}.
The progress of experimental techniques has been remarkable in recent years, and the NT theorem has been confirmed experimentally \cite{Dehollain2020}.
Therefore, the importance of the NT theorem in the rigorous study of ferromagnetism is becoming higher and higher.

In this way, the NT theorem is the first rigorous result of metallic ferromagnetism and one of the few mathematical results in many electron systems that have been confirmed experimentally.
The main goal of this section is to clarify the mathematical structure behind this  theorem and how the stability theory is constructed from it.

\subsection{General results}\label{GeneNT}
Suppose that we are given an infinite connected graph $G=(\mathbb{L}, E)$.
Take $\Lambda\in F_{\mathbb{L}}$ arbitrarily.
In this subsection, we consider a system in which a single hole moves around on $\Lambda$. We will, therefore, first construct a Hilbert space to describe the system; to describe the situation where there is one hole in $\Lambda$, we prepare the set of spin configurations: 
\be
\mathcal{S}_{\Lambda}=\bigg\{{\bs \sigma}=\{\sigma_x\}_{x\in \Lambda}\in \{-1, 0, 1\}^{\Lambda} : \sum_{x\in \Lambda}|\sigma_x|=|\Lambda|-1\bigg\}.
\ee
For ${\bs \sigma}\in \mathcal{S}_{\Lambda}$, there is only one $x_0\in \Lambda$, which satisfies $\sigma_{x_0}=0$. This $x_0$ represents the position of the hole; note that all vertexes except $x_0$ are occupied by a single electron. In what follows, we denote by $x({\bs \sigma})$ the position of the hole in ${\bs \sigma}\in \mathcal{S}_{\Lambda}$.
The vector state  corresponding to this spin configuration $\bs \sigma$ is defined  by
\be
|{\bs \sigma}\ra_{\Lambda}=c_{x({\bs \sigma}), \sigma_{x({\bs \sigma})}}\prod'_{x\in \Lambda}c_{x\sigma_x}^{*} |\varnothing\ra_{\Lambda}\in \F_{\Lambda, |\Lambda|-1},
\ee
 where $\prod\rq{}_{x\in \Lambda}$ indicates the ordered product according to an arbitrarily fixed order in $\Lambda$.
We denote by $\h_{\Lambda}^{\rm NT}$ the subspace of $\F_{\Lambda, |\Lambda|-1}$ spanned by $\{|\bs \sigma\ra_{\Lambda} : {\bs \sigma}\in \mathcal{S}_{\Lambda}\}$.
We note that this Hilbert space can also be represented as 
\be
\h_{\Lambda}^{\rm NT}=P_{\Lambda}^{\rm G} \F_{\Lambda, |\Lambda|-1},
\ee
where 
$P_{\Lambda}^{\rm G}$ is the Gutzwiller projection:
\be
P_{\Lambda}^{\rm G}:=\prod_{x\in \Lambda}(1-n_{x\up}n_{x\down}).
\ee

Let $\M_{\Lambda}^{\rm NT}$ be the abelian von Neumann algebra on $\h^{\rm NT}_{\Lambda}$ generated by 
diagonal operators associated with $\{|{\bs \sigma}\ra_{\Lambda} : {\bs \sigma} \in \mathcal{S}_{\Lambda}\}$.
For the precise definition of the diagonal operators, see the footnote to the definition of $\M^{\rm MLM}_{\Lambda}$ in Subsection \ref{GeneMLM}.
Let
\be
\zeta_{\Lambda}=\sum_{{\bs \sigma} \in \mathcal{S}_{\Lambda}} |{\bs \sigma}\ra_{\Lambda}\in \h_{\Lambda}^{\rm NT}.
\ee
We define the faithful semi-finite normal weight on $\M^{\rm NT}_{\Lambda}$ by 
$\vphi_{\Lambda}^{\rm NT}:=\vphi_{\zeta_{\Lambda}}$.\footnote{Here, recall that, for a given vector $\eta$, $\vphi_{\eta}$ is defined by  $\vphi_{\eta}(a)=\la \eta|a\eta\ra$. }
In this way, we obtain the IEE system $\{\M_{\Lambda}^{\rm NT}, \vphi_{\Lambda}^{\rm NT}\}$. The following properties follow immediately from the definition: 
\begin{itemize}
\item $\h^{\rm NT}_{\Lambda}=L^2(\M_{\Lambda}^{\rm NT}, \vphi_{\Lambda}^{\rm NT})$.
\item $\psi\in L^2(\M_{\Lambda}^{\rm NT}, \vphi^{\rm NT}_{\Lambda})_+$, if   and only if,
$\psi_{\bs \sigma}\ge 0$ for all ${\bs \sigma}\in \mathcal{S}_{\Lambda}$, where $\psi_{\bs \sigma}= {}_{\Lambda}\la {\bs \sigma}|\psi\ra_{\Lambda}$.
\item For each $\psi=\sum_{{\bs \sigma}\in \mathcal{S}_{\Lambda}} \psi_{\bs \sigma} |{\bs \sigma}\ra_{\Lambda}\in \h_{\Lambda}^{\rm NT}$, the action of the modular conjugation $J_{\Lambda}$
 is given by $J_{\Lambda} \psi=\sum_{{\bs \sigma}\in \mathcal{S}_{\Lambda} }\psi_{\bs \sigma}^* |{\bs \sigma}\ra_{\Lambda}$.
\end{itemize}

By arguments similar to those in the proof of Lemma \ref{VolSys}, we obtain the following:
\begin{Lemm}\label{VolSys2}
 If $\Lambda, \Lambda'\in F_{\mathbb{L}}$ satisfies $\Lambda\subseteq \Lambda'$, then 
$F(\M^{\rm NT}_{\Lambda\rq{}}, \vphi^{\rm NT}_{\Lambda\rq{}})\longrightarrow F(\M^{\rm NT}_{\Lambda}, \vphi^{\rm NT}_{\Lambda})$. Hence, the net $O_{1/2}^{\rm NT}=\big\{\{\M_{\Lambda}^{\rm NT}, \vphi^{\rm NT}_{\Lambda}\} : \Lambda\in F_{\mathbb{L}}
\big\}$ is a magnetic system.
\end{Lemm}
\begin{proof}
The proof is almost the same as that of Lemma \ref{VolSys}. However, since there are some points to be noted, we give an outline of the proof. 

Take $\Lambda, \Lambda'\in F_{\mathbb{L}}$ with $\Lambda\subset \Lambda'$.
Using the identification of fermionic Fock spaces in the proof of Lemma \ref{VolSys}, 
we can regard $\h^{\rm NT}_{\Lambda} \otimes \h_{\Lambda'\setminus \Lambda}$ as a subspace of $\h_{\Lambda'}^{\rm NT}$, where $\h_{\Lambda}$ is given by \eqref{HSingE}.\footnote{
Here is a brief explanation of why $\h_{\Lambda'\setminus \Lambda}$ appeared.
Recall that the Hilbert space $\h^{\rm NT}_{\Lambda'}$ (resp. $\h^{\rm NT}_{\Lambda}$ ) is the set of  vectors that represent the situation where there is only one hole on the lattice $\Lambda'$ (resp. $\Lambda$).
Since $\h_{\Lambda'\setminus \Lambda}$ is a set of  vectors representing the situation where a single electron occupies each site in $\Lambda'\setminus \Lambda$, we can see that the tensor product 
$\h^{\rm NT}_{\Lambda} \otimes \h_{\Lambda'\setminus \Lambda}$
 is a set of vectors representing the situation where there is precisely one hole in $\Lambda'$ as a whole.}
 Let $\tilde{\xi}_{\Lambda}\in \h_{\Lambda'\setminus \Lambda}$ be the vector given in the proof of Lemma \ref{VolSys}.  Define the isometric linear mapping $\kappa^{\rm NT} : \h_{\Lambda}^{\rm NT} \to  \h_{\Lambda'}^{\rm NT}$ by  $\kappa^{\rm NT} \eta=\eta\otimes \tilde{\xi}_{\Lambda'\setminus \Lambda}\, (\eta\in \h^{\rm NT}_{\Lambda})$.
 By identifying the image $\kappa^{\rm NT} \h_{\Lambda}^{\rm NT}$ with $\h_{\Lambda}^{\rm NT}$,
 we can regard $\h_{\Lambda}^{\rm NT}$ as a subspace of $\h_{\Lambda'}^{\rm NT}$.
 Let $P^{\rm NT}_{\Lambda\Lambda'}=1\otimes \Pi_{\Lambda'\setminus \Lambda}\, (\Pi_{\Lambda' \setminus \Lambda}=|\tilde{\xi}_{\Lambda'\setminus \Lambda}\ra\la \tilde{\xi}_{ \Lambda'\setminus \Lambda} |)$. Trivially, $P^{\rm NT}_{\Lambda\Lambda'}$ is the orthogonal projection from $\h_{\Lambda'}^{\rm NT}$ to $\h_{\Lambda}^{\rm NT}$. 
 Under these settings, it is straightforward to show the following:
 \begin{itemize}
\item $\M_{\Lambda}^{\rm NT}=\M_{\Lambda}^{\rm NT}\otimes \Pi_{\Lambda'\setminus \Lambda}=P^{\rm NT}_{\Lambda\Lambda'}\M_{\Lambda'}^{\rm NT}P^{\rm NT}_{\Lambda\Lambda'}$.
\item $\vphi_{\Lambda'}^{\rm NT}\circ \mathscr{E}^{\rm NT}_{\Lambda\Lambda'}=\vphi_{\Lambda'}^{\rm NT}$,
where $
\mathscr{E}^{\rm NT}_{\Lambda\Lambda'}(x)=P^{\rm NT}_{\Lambda\Lambda'}xP^{\rm NT}_{\Lambda\Lambda'}\ (x\in \M_{\Lambda^{\prime}}^{\rm NT})
$.
\item $P^{\rm NT}_{\Lambda\Lambda'} L^2(\M_{\Lambda'}^{\rm NT}, \vphi^{\rm NT}_{\Lambda'})_+=L^2(\M_{\Lambda}^{\rm NT}, \vphi^{\rm NT}_{\Lambda})_+$.
\end{itemize}
Therefore, we conclude the desired assertion in the lemma.
\end{proof}

The lemma allows the following definition.

\begin{Def}\label{DefNTcl}
\upshape
We call the magnetic system $O^{\rm NT}_{1/2}$ in Lemma \ref{VolSys2} the {\it  Nagaoka--Thouless (NT) system}. The stability class, $\mathscr{C}(O^{\rm NT}_{1/2})$,  
 is called the {\it NT stability class}.
\end{Def}

The following theorem is the prototype from which the various results in this section are derived.

\begin{Thm}\label{BasicNT}
We have the following:
\begin{itemize}
\item[\rm (i)]$\displaystyle S(\M^{\rm NT}_{\Lambda}, \vphi^{\rm NT}_{\Lambda})=(|\Lambda|-1)\big/2$ for all $\Lambda\in F_{\mathbb{L}}$. 
\item[\rm (ii)] The mapping $S_{\bullet} : H\in \mathscr{A}_{\Lambda, |\Lambda|-1}(\M_{\Lambda}^{\rm NT}, \vphi_{\Lambda}^{\rm NT})\mapsto S_{H}$ is constant and satisfies $S_{H}=(|\Lambda|-1)\big/2$ for all $\Lambda\in F_{\mathbb{L}}$. Here, recall the definition of $\mathscr{A}_{\Lambda, N}(\cdots)$, i.e., \eqref{DefALN}.
\item[\rm (iii)] Suppose that we are given a magnetic system $O_{ 1/2}$ in the NT stability class $\mathscr{C}(O^{\rm NT}_{1/2})$, i.e.,  $O_{ 1/2} \LRA O^{\rm NT}_{1/2}$. In addition,  suppose a net of Hamiltonians ${\bs H}$ is adapted to $O_{ 1/2}$. Then,  for all $\Lambda\in F_{\mathbb{L}}$, each  ground state of $H_{\Lambda}\in {\bs H}$ has the  total spin $S_{H_{\Lambda}}=(|\Lambda|-1)\big/2$. For any  increasing sequence of sets $\{\Lambda_n : n\in \BbbN\}\subset  F_{\mathbb{L}}$ such that $\bigcup_{n=1}^{\infty} \Lambda_n=\mathbb{L}$, then each NMGS associated with $\bs H$  exhibits a strict magnetic order with a spin density of $1/2$.
\item[\rm (iv)] Suppose we are in the same setting as in {\rm (iii)}.
Denote  $O_{1/2} = \big\{\{\M_{\Lambda}, \vphi_{\Lambda}\} : \Lambda\in F_{\mathbb{L}}\big\}$, and let $\ilim F(\M_{\Lambda_n}, \vphi_{\Lambda_n})$ and $
\ilim F(\M^{\rm NT}_{\Lambda_n}, \vphi^{\rm NT}_{\Lambda_n})
$ be the standard forms for macroscopic systems as defined in Theorem \ref{MacroFinite}. 
In this case, the following diagram is commutative for each $m, n\in \BbbN$ with $m<n$:
 \be
\begin{tikzcd}
F(\M_{\Lambda_m}, \vphi_{\Lambda_m}) \arrow[d]  &\arrow[l] F(\M_{\Lambda_{n}}, \vphi_{\Lambda_{n}}) \arrow[d]  &\arrow[l, dashed]    \ilim F(\M_{\Lambda_n}, \vphi_{\Lambda_n}) \arrow[d,dashed]\\
  F(\M^{\rm NT}_{\Lambda_m}, \vphi^{\rm NT}_{\Lambda_m})  &\arrow[l]F(\M^{\rm NT}_{\Lambda_{n}}, \vphi^{\rm NT}_{\Lambda_{n}})  &\arrow[l,dashed]  \ilim F(\M^{\rm NT}_{\Lambda_n}, \vphi^{\rm NT}_{\Lambda_n})
\end{tikzcd}
\ee
\end{itemize}
\end{Thm}
\begin{proof}
(i)
For each $M\in \mathrm{spec}(S^{(3)}_{ \Lambda} \restriction \h_{\Lambda}^{\rm NT})$, we set
$
\mathcal{S}_{\Lambda}(M)=\{{\bs \sigma}\in \mathcal{S}_{\Lambda} : \sum_{x\in \Lambda} \sigma_x=2M\}.
$
We readily confirm that 
\be
\mathcal{S}_{\Lambda}=\bigsqcup_{M\in \mathrm{spec}(S^{(3)}_{ \Lambda} \restriction \h_{\Lambda}^{\rm NT})}
\mathcal{S}_{\Lambda}(M).
\ee
Using this, we find that 
\be
\zeta_{\Lambda}=\sum_{M\in \mathrm{spec}(S^{(3)}_{ \Lambda} \restriction \h_{\Lambda}^{\rm NT})} \zeta_{\Lambda}(M),
\ee
where $\zeta_{\Lambda}(M)=\sum_{{\bs \sigma} \in \mathcal{S}_{\Lambda}(M)}|{\bs \sigma}\ra_{\Lambda}$.
Note that $\zeta_{\Lambda}(M) >0$ w.r.t. $L^2(\M_{\Lambda}^{\rm NT}[M], \vphi_{\Lambda, M}^{\rm NT})_+$. 
Set $S^{\rm NT}=(|\Lambda|-1)/2$.
By direct calculation, we can see that 
$\zeta_{\Lambda}(S^{\rm NT})$ has total spin $S^{\rm NT}$. In addition, it holds that 
$
(S^{(-)}_{\Lambda})^m \zeta_{\Lambda}(S^{\rm NT})\propto \zeta_{\Lambda}(S^{\rm NT}-m)
$ holds, where $S^{(-)}_{\Lambda}=S^{(1)}_{\Lambda}-\im  S_{\Lambda}^{(2)}$.
Because ${\bs S}_{\Lambda}^2$ commutes with $S^{(-)}_{\Lambda}$, each $\zeta_{\Lambda}(M)$ has total spin 
$S^{\rm NT}$. Hence, $\zeta_{\Lambda}$ has total spin $S^{\rm NT}$ as well. 
From the above discussion, we know that $\zeta_{\Lambda}$ is a magnetic vector associated with $\{\M^{\rm NT}_{\Lambda}, \vphi^{\rm NT}_{\Lambda}\}$ and $S(
\M^{\rm NT}_{\Lambda}, \vphi^{\rm NT}_{\Lambda}
)=S^{\rm NT}$ holds.

(ii),  (iii) and (iv) follow from Corollary \ref{StaGsBigA},  Theorem \ref{OtoO} and Theorem \ref{MacroDiagram}, respectively.
\end{proof}

\begin{Rem}\upshape
Suppose we are in the setting of Theorem \ref{BasicNT}.
In the theorem, the structure of the graph $G$ had little effect on the discussion. However, in practical applications,
a strong restriction on  $G$ is imposed when proving that a given  net of Hamiltonians ${\bs H}$ is adapted to $O_{1/2}$; see the next subsection for details.
\end{Rem}

\subsection{The Hubbard model}

Here we consider a system of $|\Lambda|-1$ electrons moving around on $\Lambda$ described by the Hubbard model.
We are interested in the case where the Coulomb interaction between the electrons is very large.
Let us first derive an effective Hamiltonian to describe such a system.
To do so, let us first recall that the Hubbard model is given by
\be
H_{\Lambda}^{\rm H}=\sum_{x, y\in \Lambda} \sum_{\sigma=\up, \down}t_{xy}c_{x\sigma}^*c_{y\sigma}
+\sum_{x, y\in \Lambda} \frac{U_{xy}}{2}(n_x-1)(n_y-1).
\ee
The operator $H_{\Lambda}^{\rm H}$ acts in 
$\F_{\Lambda, |\Lambda|-1}$.
$\{t_{xy}\}$ and $\{U_{xy}\}$ are $|\Lambda|\times |\Lambda|$ real symmetric matrices.
For simplicity, we assume that $U_{xx}$ is constant for all $x$:
$U_{xx}=U\, (x\in \Lambda)$.

\begin{Lemm}
Define the effective Hamiltonian $H_{\Lambda}^{\rm H, \infty}$ by $H_{\Lambda}^{\rm H, \infty}=P_{\Lambda}^{\rm G} H_{\Lambda, U=0}^{\rm H} P_{\Lambda}^{\rm G}$, where
$H_{\Lambda, U=0}^{\rm H}$ is the Hamiltonian $H_{\Lambda}^{\rm H}$ with $U=0$. 
For all $z\in \BbbC\setminus \BbbR$, we obtain
\be
\lim_{U\to \infty}(H_{\Lambda}^{\rm H}-z)^{-1}=(H_{\Lambda}^{\rm H, \infty}-z)^{-1}P_{\Lambda}^{\rm G}
\ee
 in the operator norm topology.
\end{Lemm}
\begin{proof}
See \cite[Theorem 2.1]{Miyao2017}.
\end{proof}

Thus, we were able to derive the desired effective Hamiltonian. In what follows, we regard $H_{\Lambda}^{\rm H, \infty}$ as an operator acting on $\h^{\rm NT}_{\Lambda}$.
In order to describe our results, we need to prepare some more.

Given a $y\in \Lambda$, we define the map $S_y : \mathcal{S}_{\Lambda}(M) \to \mathcal{S}_{\Lambda}(M)$ by $S_y({\bs \sigma})={\bs \sigma}\rq{}$, where $\bs \sigma\rq{}=\{\sigma_z\rq{}\}_{z\in \Lambda}$ is given by 
\begin{align}
\sigma_z\rq{}=\begin{cases}
\sigma_y & \mbox{if $z=x({\bs \sigma})$}\\
0 & \mbox{if $z=y$}\\
\sigma_z & \mbox{otherwise}.
\end{cases}
\end{align}
In other words, $S_y$ is a map that transfers the hole in $\bs\sigma$ to the vertex $y$ and transfers the spin $\sigma_y$ at $y$ to the position of the original hole.

Take $\Lambda\in F_{\mathbb{L}}$ and $M\in \mathrm{spec}(S^{(3)}_{ \Lambda} \restriction \h_{\Lambda}^{\rm NT})$ arbitrarily. 
Let us construct a natural graph from $\mathcal{S}_{\Lambda}(M)$. We say that $\{{\bs \sigma}, {\bs \sigma}\rq{}\}\, ({\bs \sigma}, {\bs \sigma}\rq{}\in \mathcal{S}_{\Lambda}(M))$ is an edge if it satisfies $S_{x({\bs \sigma}\rq{})}({\bs \sigma})={\bs \sigma}\rq{}$, i.e.,
if we swap the hole in the spin configuration $\bs \sigma$ with the spin, $\sigma_{x({\bs\sigma}\rq{})}$,  at the position of the hole in $\bs\sigma\rq{}$, the resulting spin configuration, $S_{x({\bs \sigma}\rq{})}({\bs \sigma})$,  will match the spin configuration $\bs \sigma\rq{}$.
If we denote the set of all edges by $\mathcal{E}_{\Lambda}(M)$, then we thus get the graph $
\mathcal{G}_{\Lambda}(M)=(\mathcal{S}_{\Lambda}(M), \mathcal{E}_{\Lambda}(M))
$.
 In what follows, we set 
 \be
 F_{\mathbb{L}}^{\rm Conn}=\{\Lambda\in F_{\mathbb{L}} : \mbox{$\mathcal{G}_{\Lambda}(M)$ is connected for all $M\in \mathrm{spec}(S^{(3)}_{ \Lambda} \restriction \h_{\Lambda}^{\rm NT})$}\}.
 \ee
 We always assume that $F_{\mathbb{L}}^{\rm Conn}$ is nonempty.
Let $G^t$ and $G^t_{\Lambda}=(\Lambda, E_{\Lambda}^t)\, (\Lambda\in F_{\mathbb{L}})$ be the graphs generated  by the hopping matrix
 $\{t_{xy}\}$; see  the definition immediately preceding Condition \ref{CoupAss}. 
We impose the following restrictions:

\begin{Assum}\label{ConnAss}\upshape
\indent
\begin{itemize}
\item[(i)] $t_{xy}> 0$ for all $\{x, y\}\in E^t_{\Lambda}$.
\item[(ii)] $G^t$ contains a normal spanning tree, $T(G)$,  contained in $G$.  
For more details about  normal spanning trees,
see Subsection \ref{SubsecHei} for detail.
\item[(iii)] For all $\Lambda\in F_{\mathbb{L}}^{\rm Conn}$ and $M\in \mathrm{spec}(S^{(3)}_{ \Lambda} \restriction \h_{\Lambda}^{\rm NT})$, if $\{{\bs \sigma}, {\bs \sigma}\rq{}\}\in \mathcal{E}_{\Lambda}(M)$, then $\{x({\bs \sigma}), x({\bs \sigma}\rq{})\} \in E_{\Lambda}^t$ holds.
\end{itemize}
\end{Assum}

The following are examples that satisfy Condition \ref{ConnAss}.
\begin{Exa}\upshape
The models with the following (i) and (ii) satisfy Condition \ref{ConnAss}
:
\begin{itemize}
\item[(i)]  $G$ is a triangular, square cubic, fcc or bcc lattice;
\item[(ii)] $t_{xy}$ is nonvanishing between nearest neighbor sites.
\end{itemize}
See  \cite[Lemma 11.9]{Tasaki2020} for details of the proof.\footnote{The property (iii) of Condition \ref{ConnAss} is  called  the {\it connectivity condition}  in   \cite{Tasaki2020}. Note that  Bobrow, Stubis,  and Li derived a necessary and sufficiently condition for the connecitivity condition  in \cite{Bobrow2018}.
}
Moreover, in each of these examples, we can easily construct an   increasing sequence
 $\{\Lambda_n : n\in \BbbN\}\subset  F^{\rm Conn}_{\mathbb{L}}$ satisfying $\bigcup_{n=1}^{\infty}\Lambda_n=\mathbb{L}$.
This fact is needed in Theorems \ref{NTHubbThm}.
\end{Exa}

Using the theory developed in Sections \ref{Sect3} and \ref{Sect4}, we can generalize the NT theorem as follows.
\begin{Thm}\label{NTHubbThm}
Assume Condition \ref{ConnAss}. 
We have the following:
\begin{itemize}
\item[\rm (i)] For each $\Lambda\in F^{\rm Conn}_{\mathbb{L}}$, $H_{\Lambda}^{\rm H, \infty}\in A_{\Lambda, |\Lambda|-1}(\M_{\Lambda}^{\rm NT}, \vphi_{\Lambda}^{\rm NT})$ holds. Hence,  $S_{H_{\Lambda}^{\rm H, \infty}}=(|\Lambda|-1)/2$.

 \item[\rm (ii)] Set ${\bs H}^{\rm H, \infty}=\{H_{\Lambda} ^{\rm H, \infty} : \Lambda \in F^{\rm Conn}_{\mathbb{L}}\}$. Then ${\bs H}^{\rm H, \infty}$ is adapted to $O_{1/2}^{\rm NT}$.  Hence,
for any   increasing sequence of sets $\{\Lambda_n : n\in \BbbN\}\subset  F^{\rm Conn}_{\mathbb{L}}$ satisfying $\bigcup_{n=1}^{\infty}\Lambda_n=\mathbb{L}$,  each NMGS associated with ${\bs H}^{\rm H, \infty}$  exhibits a strict magnetic order with a spin density of  $1/2$.
\end{itemize}
\end{Thm}
\begin{proof}
By applying (vii) of Theorem \ref{ManyHa}, we know that $H_{\Lambda}^{\rm H, \infty} \in A_{\Lambda, |\Lambda|-1}(\M_{\Lambda}^{\rm NT}, \vphi_{\Lambda}^{\rm NT})$ for all $\Lambda\in F_{\mathbb{L}}^{\rm Conn}$, which implies that ${\bs H}^{\rm H, \infty}$ is adapted to $O_{1/2}^{\rm NT}$. 
By applying Theorem \ref{BasicNT}, we obtain the desired assertions in Theorem \ref{NTHubbThm}.
\end{proof}

\subsection{The Holstein--Hubbard model}
In this subsection, we will show that the NT theorem still holds when the electron-phonon interaction is taken into account.
First, let us recall that the Holstein--Hubbard Hamiltonian describing the interacting electron-phonon  system is given as follows:
\be
H^{\rm HH}_{\Lambda} = H^{\rm H}_{\Lambda} +
\sum_{x, y\in \Lambda}
g_{xy} (n_x -1)(b_y^*
 + b_y ) +\sum_{x\in \Lambda}
\omega b_x^*b_x,
\ee
where $H_{\Lambda}^{\rm H}$ is the Hubbard Hamiltonian with $U_{xx}=U\, (x\in \Lambda)$.
At this stage, we consider $H_{\Lambda}^{\rm H}$ to be  an operator acting on $\F_{\Lambda, |\Lambda|-1}\otimes \F^{\rm ph}_{\Lambda}$. 
See Section \ref{Sect5} for the detailed definition of the Hamiltonian.

As in the previous subsection, we can derive the effective Hamiltonian for the case where the Coulomb interaction between electrons is very strong.

\begin{Lemm}
Define the effective Hamiltonian $H_{\Lambda}^{\rm HH, \infty}$ by $H_{\Lambda}^{\rm HH, \infty}=P_{\Lambda}^{\rm G} H_{\Lambda, U=0}^{\rm HH} P_{\Lambda}^{\rm G}$, where
$H_{\Lambda, U=0}^{\rm HH}$ is the Hamiltonian $H_{\Lambda}^{\rm HH}$ with $U=0$. 
For all $z\in \BbbC\setminus \BbbR$, we obtain
\be
\lim_{U\to \infty}(H_{\Lambda}^{\rm HH}-z)^{-1}=(H_{\Lambda}^{\rm HH, \infty}-z)^{-1}P_{\Lambda}^{\rm G}
\ee
 in the operator norm topology.
\end{Lemm}
Note that the effective Hamiltonian $H_{\Lambda}^{\rm HH, \infty}$ acts on $\h_{\Lambda}^{\rm NT}\otimes \F^{\rm ph}_{\Lambda}$.

Recall that the bosonic Fock space $\F^{\rm  ph}_{\Lambda}$ can be identified with
$L^{2}(\BbbR^{|\Lambda|})$.   
With this mind, we define 
\be
\N^{\rm HH}_{\Lambda}=\M^{\rm NT}_{\Lambda} \otimes L^{\infty}(\BbbR^{|\Lambda|}),
\ee
where $\M^{\rm NT}_{\Lambda}$ is given in Subsection \ref{GeneNT}.
Let
\be
\zeta_{\Lambda}^{\rm HH}=\zeta_{\Lambda}\otimes  \Omega_{\Lambda}^{\rm ph}\in \h_{\Lambda}^{\rm NT} \otimes \F^{\rm ph}_{\Lambda}.
\ee
We define the faithful semi-finite normal weight on $\N^{\rm HH}_{\Lambda}$ by 
$\mu_{\Lambda}^{\rm HH}:=\vphi_{\zeta_{\Lambda}^{\rm HH}}$.\footnote{Here, recall that, for a given vector $\eta$, $\vphi_{\eta}$ is defined by  $\vphi_{\eta}(a)=\la \eta|a\eta\ra$. }
In this way, we obtain the IEE system $\{\N_{\Lambda}^{\rm HH}, \mu_{\Lambda}^{\rm HH}\}$. The following properties follow immediately from the definition: 

\begin{itemize}
\item $\displaystyle \h^{\rm NT}_{\Lambda}\otimes \F^{\rm ph}_{\Lambda}=L^2(\N_{\Lambda}^{\rm HH}, \mu_{\Lambda}^{\rm HH})=\int^{\oplus}_{\BbbR^{|\Lambda|}}
L^2(\M^{\rm NT}_{\Lambda}, \vphi^{\rm NT}_{\Lambda}) d{\bs q}
$.
\item $\psi\in L^2(\N_{\Lambda}^{\rm HH}, \mu^{\rm HH}_{\Lambda})_+$, if and only if,
$\psi_{\bs \sigma}({\bs q})\ge 0$ for all ${\bs \sigma}\in \mathcal{S}_{\Lambda}$ and a.e. ${\bs q}\in \BbbR^{|\Lambda|}$, where $\psi_{\bs \sigma}({\bs q})= {}_{\Lambda}\la {\bs \sigma}|\psi({\bs q})\ra_{\Lambda}$.
Note that we use the identification in the first property.
\item For each $\psi=\sum_{{\bs \sigma}\in \mathcal{S}_{\Lambda}}  |{\bs \sigma}\ra_{\Lambda}\otimes \psi_{\bs \sigma}
\in \h_{\Lambda}^{\rm NT}\otimes \F^{\rm ph}_{\Lambda}$, the action of the modular conjugation $J_{\Lambda}$
 is given by $J_{\Lambda} \psi=\sum_{{\bs \sigma}\in \mathcal{S}_{\Lambda} } |{\bs \sigma}\ra_{\Lambda}\otimes \psi_{\bs \sigma}^*$, where $(\psi_{\bs \sigma}^*)(\bs q):=\psi_{\bs \sigma}({\bs q})^*$.
\end{itemize}

\begin{Lemm}\label{VolSys3}
We have the following:
\begin{itemize}
\item[\rm (i)] If $\Lambda, \Lambda'\in F^{\rm Conn}_{\mathbb{L}}$ satisfies $\Lambda\subseteq \Lambda'$, then 
$F(\N^{\rm HH}_{\Lambda\rq{}}, \mu^{\rm HH}_{\Lambda\rq{}})\longrightarrow F(\N^{\rm HH}_{\Lambda}, \mu^{\rm HH}_{\Lambda})$. Hence, the net $O_{1/2}^{\rm NT,HH}=\big\{\{\N_{\Lambda}^{\rm HH}, \mu^{\rm HH}_{\Lambda}\} : \Lambda\in F_{\mathbb{L}}^{\rm Conn}
\big\}$ is a magnetic system.
\item[\rm (ii)] For each $\Lambda 
\in F^{\rm Conn}_{\mathbb{L}}
, $
$F(\M^{\rm NT}_{\Lambda}, \vphi^{\rm NT}_{\Lambda})\longrightarrow F(\M^{\rm NT}_{\Lambda}, \vphi^{\rm NT}_{\Lambda})$.
\end{itemize}
The above results can be summarized as the following commutative diagram:
\be
\begin{tikzcd}
  F(\N^{\rm HH}_{\Lambda}, \mu^{\rm HH}_{\Lambda})   \arrow[d] &
  \ar[l]   F(\N^{\rm HH}_{\Lambda\rq{}}, \mu^{\rm HH}_{\Lambda\rq{}})
  \arrow[d]  \\
  F(\M^{\rm NT}_{\Lambda}, \vphi^{\rm NT}_{\Lambda}) &\arrow[l] F(\M^{\rm NT}_{\Lambda\rq{}}, \vphi^{\rm NT}_{\Lambda\rq{}})
\end{tikzcd}
\ee
Furthermore, when $\{\M_{\Lambda}, \vphi_{\Lambda}\}$  is replaced by $\{\M^{\rm NT}_{\Lambda}, \vphi^{\rm NT }_{\Lambda}\}$ and $\{\M^{\prime}_{\Lambda}, \vphi^{\prime}_{\Lambda}\}$ is replaced by $\{\N^{\rm HH}_{\Lambda}, \mu^{\rm HH}_{\Lambda}\}$ in Theorem \ref{MacroDiagram}, the diagram \eqref{MacDiagComm} is commutative.

\end{Lemm}
\begin{proof}
(i) In this proof, we will use the identifications for fermionic and bosonic Fock spaces given in Lemmas \ref{VolSys} and \ref{ConHoHu}. For simplicity of symbols, we set 
$\frak{K}_{\Lambda}=\h_{\Lambda}^{\rm NT} \otimes \F^{\rm ph}_{\Lambda}$.
Recall the following two facts.  First, $\h_{\Lambda}^{\rm NT} \otimes \h_{\Lambda\rq{}\setminus \Lambda}$  is a subspace of $\h^{\rm NT}_{\Lambda\rq{}}$, as derived in the proof of Lemma \ref{VolSys2}.  Secondly,
the identification $\F^{\rm ph}_{\Lambda\rq{}}=\F^{\rm ph}_{\Lambda} \otimes \F^{\rm ph}_{\Lambda\rq{}\setminus \Lambda}$ used in the proof of Lemma \ref{ConHoHu}.
By using  the  facts, we find  that $\frak{K}_{\Lambda} \otimes \h^{\rm HH}_{\Lambda\rq{}\setminus \Lambda}$ can be regarded as  a closed subspace of $\frak{K}_{\Lambda\rq{}}$, where $\h_{\Lambda}^{\rm HH}$ is defined in Subsection \ref{MLM-HH}.
As before, we define the isometric linear mapping $\kappa : \mathfrak{K}_{\Lambda} \to \mathfrak{K}_{\Lambda\rq{}}$ by 
$\kappa \eta=\eta\otimes \tilde{\xi}^{\rm HH}_{\Lambda\rq{}\setminus \Lambda}$, where $\tilde{\xi}^{\rm HH}_{\Lambda\rq{}\setminus \Lambda}$ is given in the proof of Lemma \ref{ConHoHu}.
By identifying $\kappa \frak{K}_{\Lambda}$ with $\frak{K}_{\Lambda}$, we can regard $\frak{K}_{\Lambda}$
 as a closed subspace of $\frak{K}_{\Lambda\rq{}}$.
 The rest of the proof is similar to that of Lemma \ref{VolSys2}, so we omit it.

 (ii) Let $Q_{\Lambda}^{\rm HH}$ be the orthogonal projection from $\frak{K}_{\Lambda}$ to $\h_{\Lambda}^{\rm NT}$ : $Q_{\Lambda}^{\rm HH}=1\otimes |\Omega^{\rm ph}_{\Lambda}\ra\la \Omega^{\rm ph}_{\Lambda}|$. Then we readily confirm the following:
 \begin{itemize}
\item $Q_{\Lambda}^{\rm HH} L^2(\N_{\Lambda}^{\rm HH}, \mu_{\Lambda}^{\rm HH})=L^2(\M_{\Lambda}^{\rm NT}, \vphi_{\Lambda}^{\rm NT})$.
\item $
Q_{\Lambda}^{\rm HH} \N_{\Lambda}^{\rm HH} Q_{\Lambda}^{\rm HH}=\M^{\rm NT}_{\Lambda}.
$
\item $Q_{\Lambda}^{\rm HH} L^2(\N_{\Lambda}^{\rm HH}, \mu_{\Lambda}^{\rm HH})_+=L^2(\M_{\Lambda}^{\rm NT}, \vphi_{\Lambda}^{\rm NT})_+$
\end{itemize}
To prove the third property, we used the fact that $\Omega^{\rm ph}_{\Lambda}({\bs q})$
is strictly positive for all ${\bs q} \in \BbbR^{\Lambda}$. This completes the proof of (ii).
\end{proof}

\begin{Thm}\label{NTHHThm}
Assume Condition \ref{ConnAss}. 
We have the following:
\begin{itemize}
\item[\rm (i)] For each $\Lambda\in F^{\rm Conn}_{\mathbb{L}}$, $H_{\Lambda}^{\rm HH, \infty}\in \mathscr{A}_{\Lambda, |\Lambda|-1}(\M_{\Lambda}^{\rm NT}, \vphi_{\Lambda}^{\rm NT})$ holds. Hence,  $S_{H_{\Lambda}^{\rm HH, \infty}}=(|\Lambda|-1)/2$.

 \item[\rm (ii)] Set ${\bs H}^{\rm HH, \infty}=\{H_{\Lambda} ^{\rm HH, \infty} : \Lambda \in F^{\rm Conn}_{\mathbb{L}}\}$. Then ${\bs H}^{\rm HH, \infty}$ is adapted to $O_{1/2}^{\rm NT, HH}$.  Hence,
for any   increasing sequence of sets $\{\Lambda_n : n\in \BbbN\}\subset  F^{\rm Conn}_{\mathbb{L}}$ satisfying $\bigcup_{n=1}^{\infty}\Lambda_n=\mathbb{L}$,  each NMGS associated with ${\bs H}^{\rm HH, \infty}$  exhibits a strict magnetic order with a spin density of  $1/2$.
\end{itemize}
\end{Thm}
\begin{proof}

Fix $\Lambda\in F^{\rm Conn}_{\mathbb{L}}$, arbitrarily. By applying (viii) of  Theorem \ref{ManyHa}, we know that 
$
\{e^{-\beta H^{\rm HH, \infty}_{\Lambda, M}} \}_{\beta \ge 0}$ is ergodic w.r.t. $L^2(\N^{\rm HH}_{\Lambda}[M], \mu^{\rm HH}_{\Lambda, M})_+$ for all   $ M\in \mathrm{spec}(S_{\Lambda}^{(3)} \restriction \h^{\rm NT}_{\Lambda})$, which implies that ${\bs H}^{\rm HH, \infty}$ is adapted to $O_{1/2}^{\rm NT, HH}$.
Hence, $H^{\rm HH, \infty}_{\Lambda}\in \mathscr{A}_{\Lambda, |\Lambda|-1}(\M_{\Lambda}^{\rm NT}, \vphi_{\Lambda}^{\rm NT})$
for all $\Lambda\in F^{\rm Conn}_{\mathbb{L}}$.
By using Theorem \ref{BasicNT} and Lemma \ref{VolSys3}, we conclude the desired assertions in Theorem \ref{NTHHThm}.
\end{proof}

In Theorem \ref{NTHubbThm}, it was shown that the NT stability class $\mathscr{C}(O^{\rm NT}_{1/2})$ characterizes the magnetic properties of the ground states of the effective Hamiltonians $H_{\Lambda} ^{\rm H, \infty}\, (\Lambda\in F^{\rm Conn}_{\mathbb{L}})$. On the other hand, in Theorem \ref{NTHHThm}, we proved that the magnetic properties of the ground states of the Hamiltonians   incorporating the electron-phonon interaction in $H_{\Lambda} ^{\rm H, \infty}\, (\Lambda\in F^{\rm Conn}_{\mathbb{L}})$ are also characterized by $\mathscr{C}(O^{\rm NT}_{1/2})$. Together, these theorems allow us to conclude that the magnetic properties of the ground states of the effective Hamiltonians $H_{\Lambda} ^{\rm H, \infty} \, (\Lambda\in F^{\rm Conn}_{\mathbb{L}})$ are stable under the electron-phonon interaction.

\subsection{On some more stability theorems related to the NT stability class}
We conclude this section by mentioning some related topics that we cannot touch on because of the limited number of pages.

In \cite{Miyao2017}, 
we analyze a model describing the  NT system interacting with quantized radiation fields.
This model is also shown to belong to the NT stability class. Thus, the NT magnetic properties are stable even under the influence of  quantized radiation  fields.
 
In \cite{Aizenman1990}, Aizenman and Lieb proved that the NT theorem could be extended to the case of finite temperature. The author then proved that the extended NT theorem of Aizenman and Lieb holds under the  influence of lattice vibrations and  quantized radiation  fields \cite{Miyao2020-2}. The method used in this proof is very different from the one in this paper, but positivity is still essential at the core of the proof.

In \cite{Kollar1996}, Koller {\it et al.}  analyze the ground states of a Hamiltonian that incorporates all possible nearest-neighbor Coulomb interactions\footnote{To be precise, the density-density interaction, bond-charge interaction, exchange interaction and hopping of double occupancie are taken into account.}.
Thanks to this generalization of the model, the limit operation of taking $U$ to infinity, as seen in this paper, is no longer necessary, and we can show that the extended NT theorem holds for  finite $U$. With the theory presented in this paper, this extension can also be explained using the NT stability class. 
Furthermore, we can discuss the stability of magnetic properties under electron-phonon interactions and electron-quantized radiation  field interactions.
Detailed proofs will be given in another paper.

\section{Discussion}\label{Sect7}
\subsection{Other stability classes}
A simple question that should naturally arise in the readers' minds is whether there are other stability classes related to magnetic orders besides the stability classes presented so far in this paper. This question is natural and essential to evaluate the range of adaptation of the theory presented in this paper. In the following, we will outline some stability classes. These stability classes have not been clearly recognized so far, and some of them are pointed out here for the first time. For reasons of page space, we will not discuss the detailed construction of these stability classes, but readers who have read this far should have some idea of what they are.

\subsubsection*{The stability class associated with the attractive Hubbard model}

In his famous paper \cite{Lieb1989}, Lieb proves two theorems, one of which is  discussed in Section \ref{Sect5}. The other is a theorem about the Hubbard model with attractive Coulomb interaction.
As is well known, when the system is at half-filling, 
the hole-particle transformation transforms the repulsive Hubbard model into the attractive Hubbard model, and vice versa.
 For this reason, the two theorems are often regarded as the same.
However, since the models to which the ideas in the proofs of each theorem can be adapted are pretty different, it is more convenient to regard the two theorems as different.
  Needless to say, Ref. \cite{Lieb1989} does not include any discussion of the stability classes given in this paper. Therefore, it is necessary to explain the stability class that can describe the attractive Hubbard model to understand the innovative idea contained in Lieb's paper.

For this purpose,  let us first recall Lieb's theorem on the attractive Hubbard model:
in the attractive Hubbard model, the ground states  have spin angular momentum $S=0$ for every (even) electron filling.
In the context of this paper, this result can be interpreted as follows.
There exists a stability class $\mathscr{C}$, which satisfies the following property:
 if a net of  Hamiltonians is adapted to a magnetic system belonging to the  class $\mathscr{C}$, then their ground states have total spin $S=0$. The net of the attractive Hubbard Hamiltonians can be shown in practice to be adapted to a  magnetic system belonging to   $\mathscr{C}$. 
 Note that the class $\mathscr{C}$ was first pointed out in \cite{Miyao2019}.
 
 In \cite{Freericks1995}, Freericks and Lieb proved that the ground state of the Holstein model is unique and has total spin $S=0$.
 This result can be interpreted that the Holstein model is adapted to a particular magnetic system belonging to $\mathscr{C}$.
 In the summary reviews \cite{Shen1998,Tian2004} on Lieb's theorems, various models are discussed. 
 We can prove that most of the models discussed in those papers are adapted to magnetic systems belonging to the MLM stability class or class $\mathscr{C}$.

\subsubsection*{The stability class associated with one-electron Kondo lattice model}

Sigrist {\it et al.}   proved in \cite{Sigrist1991} that the ground states of the Kondo lattice model with a single conduction electron have total spin $S=(|\Lambda|-1)/2$.
In the context of the present paper, this result is interpreted as follows.
There exists a stability class $\mathscr{C}$ satisfying the following property: if a net of   Hamiltonians is adapted to  a magnetic system belonging to  $\mathscr{C}$, then their ground states have total spin $S$.
It can be shown that the net of the  Hamiltonians analyzed in \cite{Sigrist1991} is   also adapted to  a magnetic system belonging to the  class $\mathscr{C}$.
Furthermore, it is possible to extend the results of \cite{Sigrist1991} using this stability class. For example, consider  Hamiltonians that incorporate the interactions of conduction electrons with lattice vibrations and quantized radiation fields.  In this case, it can be shown that the ground states of the Hamiltonians still have total spin $S$. The proof of this claim will be presented in \cite{Miyao2022-1}.

\subsubsection*{The stability classes associated with $t$-$J$ models}
Tasaki have studies the $t$-$J$ model on a bipartite lattice characterised by negative hopping matrix elements within the same sublattices in \cite{Tasaki_1990}. He proves that the MLM theorem is valid for this model as well: when the two sublattices contain the same number of electrons, the ground state is unique and has total spin $S=0$. In \cite{Wang2014},   Wang and Ye have examined  the $t$-$J$ {\it chain} of arbitrary spin $s$. Their claims is stated as follow: the ground state has
total spin quantum number $S=0$ if $N$ is even, and $S=s$ if
$N$ is odd, where $N$ is the number of spin $s$-\lq\lq{}electrons\rq\rq{}; the ground state is unique apart from the $(2S+1)$-fold spin degeneracy.
We can explain the above results by constructing stability classes and further prove their stability under various perturbations. To be more precise, we can construct magnetic systems to which nets of the Hamiltonian can be adapted;  the stability of the magnetic properties of the ground states can be discussed in detail using the stability classes represented by the constructed magnetic systems. 
We will not enter into a detailed discussion of this model here but will discuss it in another paper.

\subsubsection*{The stability classes associated with the SU($n$) Hubbard models}

In recent years, the SU($n$) generalization of the Hubbard
model has been receiving much attention since it was realized
with ultracold atoms in optical lattices \cite{Cazalilla2014,PhysRevX.6.021030}. This model is thought to be able to explain various phases that cannot be described by the usual (i.e., SU($2$)) Hubbard model, and theoretical understanding is progressing. For example, color superfluid and trion phases are expected to be understood theoretically using the SU($3$) Hubbard model \cite{Rapp2008,Titvinidze2011,Zhao2007}.
As we have seen, the magnetic properties of the ground states  of the ordinary Hubbard model can be described by several stability classes, and 
it is natural to ask what stability classes can be used to describe the ground-state properties of the SU($n$) Hubbard model. Here, we outline some partial results and prospects.

In \cite{Katsura2013}, Katsura and Tanaka extended the Nagaoka--Thouless (NT) theorem to the SU($n$) Hubbard model.
By extending the discussion in this paper, the SU($n$) version of the NT theorem can be explained using stability classes.
On the other hand, there is still room for discussion on the interpretation of the stability class, such as what kind of interaction with the environment is physically possible. A detailed description of the above results will be given in another paper.

In \cite{Yoshida2021}, Katsura and Yoshida applied the method of Majorana reflection positivity to the SU($n$) Hubbard model and obtained interesting results on the ground states.
We note that the origin of Majorana reflection positivity goes back to work on Majoranas by Jaffe and Pedrocchi \cite{Jaffe2015};  Wei {\it et al.} used  the method to analyze the ordinary Hubbard model in  \cite{Wei2015}. 
The results of Katsura and Yoshida can also be explained in terms of stability classes.
On the other hand, there are many unanswered questions, such as what kind of interaction with the environment can be explained by this stability class.

The study of the SU($n$) Hubbard model has the potential to discover  new stability classes that cannot be considered in the SU($2$) model, and further investigation is therefore expected.

\subsection{Potential for further applications}

Many of the examples given above can be explained coherently with the theory we have proposed in this paper.
In this subsection, we would like to mention a few examples that do not necessarily fit into the framework of this paper but could be described by similar ideas.

First, let us discuss the method of spatial reflection positivity.
The method of spin reflection positivity developed by Lieb \cite{Lieb1989} was inspired by spatial reflection positivity in axiomatic quantum field theory \cite{Osterwalder1973,Osterwalder1975}. We will refer to spatial reflection positivity simply as reflection positivity (RP).
The method of RP has a variety of applications, including the analysis of phase transition phenomena. \footnote{
For mathematical aspects of RP, see Ref. \cite{Neeb2018}.
} For example, Kennedy {\it et al.} applied RP to the analysis of long-range order in the ground state of the Heisenberg model \cite{Kennedy1988}. By extending this method, the author of this paper and his collaborators have proposed a new description of RP in one-dimensional fermionic systems in \cite{Miyao2021-3}. 
The basic idea of Ref. \cite{Miyao2021-3} is the same as in this paper, which strongly suggests the existence of a stability class describing the charge density wave order. 
It is intriguing whether there is a stability class that describes an order other than magnetic order.

Finally, we would like to mention flat-band ferromagnetism.
This is a mechanism of ferromagnetism discovered by Mielke and Tasaki \cite{Mielke1993}, and it is more complicated than the examples discussed in this paper.
The nature of flat-band ferromagnetism has been studied in depth by many researchers; see, e.g.,  \cite{Tasaki2020} and references therein.  However, at present, it is not clear what stability class is used to describe this mechanism.

\appendix

\section{Structures of  Hamiltonians}\label{SectA}
\subsection{Two basic types of Hamiltonians}

In Sections \ref{Sect3} and \ref{Sect4}, we have not placed any restrictions on the structure of the Hamiltonians. The most intractable problem in concrete applications of the preceding discussion is to demonstrate that (iii) of Definition \ref{DefAHamiC} is satisfied by the Hamiltonians under consideration. In this appendix, we will consider two types of Hamiltonians that are often encountered in applications, and clarify under what conditions (iii) of Definition \ref{DefAHamiC} is satisfied. To prove this fundamental condition, we note that the structure of the graph describing the many-electron system is also determined to some extent simultaneously.

In many texts on quantum statistical mechanics, the Hamiltonian of a quantum spin system is defined using the concept of ``interaction". Mathematically, an  ``interaction" is usually described by a set of bounded linear operators satisfying some favorable properties; see, e.g., \cite{Bratteli1997,Simon1993} for detail. 
Recall that one of the subjects of this paper is to mathematically examine the effects of the interactions between electrons and environmental systems on the magnetic properties of the ground states.
Typical examples of environmental systems that we have in mind are phonon systems and systems of quantized radiation fields.
 Such environmental systems are generally difficult to characterize using only bounded linear operators, which inevitably requires the definition of Hamiltonians and interactions in terms of unbounded linear operators.

We consider a net of reducible Hilbert spaces: $
\{\h_{\Lambda, [\varrho|\Lambda|]} \in \mathscr{H}_{\Lambda, [\varrho |\Lambda|]} : \Lambda\in  P(\mathbb{L}) \}$.
Let ${\sf H}_{\varrho}=\{\X_{\Lambda} \in \mathscr{I}(\h_{\Lambda, [\varrho |\Lambda|]}) : \Lambda\in  P(\mathbb{L}) \}$ be a net of IEE spaces.
Suppose that we are given  a magnetic system $O_{\varrho} =\{\{\M_{\Lambda}, \vphi_{\Lambda}\} \in \mathscr{Q}(\X_{\Lambda}) : \X_{\Lambda}\in {\sf H}_{\varrho}, \  \Lambda\in  F_{\mathbb{L}} \}$.
Let $\Lambda, \Lambda'\in F_{\mathbb{L}}$. Recall that, if $\Lambda\subset \Lambda'$, then $\X_{\Lambda}$
can be regarded as a closed  subspace of $\X_{\Lambda'}$. Hence, we can write $\X_{\Lambda'}=\X_{\Lambda} \oplus \X_{\Lambda}^{\perp}$. By identifying the von Neumann algebra $\M_{\Lambda}$ on $\X_{\Lambda}$
with $\{x\oplus 0 : x\in \M_{\Lambda}\}$, we can regard $\M_{\Lambda}$ as a von Neumann algebra on $\X_{\Lambda'}$.  By the definition of the  magnetic system,  $\M_{\Lambda}$ is a  von Neumann subalgebra of $\M_{\Lambda'}$. Now we define the  von Neumann algebra on $\X_{\Lambda'}$ by 
\be
\tilde{\M}_{\Lambda}=\{x\oplus 1 : x\in \M_{\Lambda}\},
\ee
where $1$ is the identity operator on $\X_{\Lambda}^{\perp}$.
Then $\tilde{\M}_{\Lambda}$ is a natural extension of $\M_{\Lambda}$ onto $\X_{\Lambda'}$:
$\tilde{\M}_{\Lambda} \restriction \X_{\Lambda}=\M_{\Lambda} \oplus \{0\} \equiv \M_{\Lambda}$.
In this section, given an operator $A$ on $\X_{\Lambda}$, we will write $\tilde{A} =A\oplus 1 $. This notation  makes sense  even when $A$ is an unbounded operator.

\begin{Def}\label{FundA} \upshape
An {\it interaction} is a set of   self-adjoint operators ${\bs \Psi} = \{\Psi(X) : X\in P(\mathbb{L})\}$ satisfying the following:
\begin{itemize}
\item[1.] $\Psi(X)$ acts in $\h_X$ and  is bounded from below for all $X\in P(\mathbb{L})$.
\item[2.] Fix $\Lambda\in P(\mathbb{L})$ arbitrarily. There exists a  dense subspace $\frak{D}_{\Lambda}$ of $\X_{\Lambda}$ such that 
$\frak{D}_{\Lambda} \subseteq Q(\tilde{\Psi}(X))$ for all $X\subseteq \Lambda$,
where $Q$ stands for the form domain.

\item[3.] $e^{-s \Psi(X)}$ commutes with $S_{X}^{(1)}$, $S_{X}^{(2)}$ and $S_{X}^{(3)}$ 
for all $X\in P(\mathbb{L})$ and $s\ge 0$.
\end{itemize}
Given an interaction ${\bs \Psi}$, we define the Hamiltonian  $H^{\bs \Psi}_{\Lambda}$ by 
\be
H_{\Lambda}^{\bs \Psi}=\sum^{\bullet}_{X\subseteq \Lambda} \Phi(X),\ \ \Phi(X)=\tilde{\Psi}(X),
\ee
where $\displaystyle \sum^{\bullet}$ stands for a form sum. 
\end{Def}
 From the property 2. of Definition \ref{FundA}, 
 $H_{\Lambda}^{\bs \Psi}$  is a self-adjoint operator on $\X_{\Lambda}$, bounded from below. In addition, by the property 3., 
 $e^{-s H_{\Lambda}^{\bs\Psi}}$ commutes with the total spin operators $S_{\Lambda}^{(i)}\, (i=1, 2, 3)$.

The fundamental assumption is as follows.
\begin{Assum}\label{PPBasicA}\upshape
For any $\Lambda\in F_{\mathbb{L}}$, it holds that 
$e^{-\beta \Phi(X)} \unrhd 0$ w.r.t. $L^2(\M_{\Lambda}, \vphi_{\Lambda})_+$ for all $\beta \ge 0$ and $X\subseteq \Lambda$.
\end{Assum}

The following proposition is a basic premise in showing that (iii) of Definition \ref{DefAHamiC} is satisfied.
\begin{Prop}
For each  $\Lambda\in F_{\mathbb{L}}$, it holds that 
$e^{-\beta H^{\bs \Psi}_{\Lambda}} \unrhd 0$ w.r.t. $L^2(\M_{\Lambda}, \vphi_{\Lambda})_+$ for all $\beta \ge 0$.
\end{Prop}

\begin{proof}
A straightforward application of 
 the Trotter--Kato product formula \cite[Theorem S.21]{Reed1981}
 yields the desired claim.
\end{proof}

In this paper, we consider the case where the interaction $\Phi(X)$ can be split into two parts: 
\be
\Phi(X)=\Phi^{\rm w}(X)+\Phi^{\rm s}(X),
\ee
where  $\Phi^{\rm w}(X)$ and $\Phi^{\rm s}(X)$ are characterized as follows:
\begin{itemize}
\item $\Phi^{\rm s}(X)$ is the part satisfying the following: $\Phi^{\rm s} (X)$ is bounded and  $-\Phi^{\rm s}(X) \unrhd 0$ w.r.t. $L^2(\M_{\Lambda}, \vphi_{\Lambda})_+$.

\item $\Phi^{\rm w}(X):=\Phi(X)-\Phi^{\rm s}(X)$.
\end{itemize}
\begin{Rem} \upshape
The assumption of boundedness in the definition of $\Phi^{\rm s}(X)$ can be weakened: for example, the following theorems are still true even if $\Phi^{\rm s}(X)$ is assumed to be relatively bounded to $\Phi^{\rm w}(X)$. However, the above setting is sufficient for the consideration of the examples in this paper.
\end{Rem}

For the benefit of the following discussion, we divide the Hamiltonian into two parts as follows:
\be
H_{\Lambda}^{\bs \Psi}=H^{\bs \Psi, {\rm w}}_{\Lambda}+ H^{\bs \Psi, {\rm s}}_{\Lambda},
\ee
where 
\be
H^{\bs \Psi, {\rm w}}_{\Lambda}=\sum_{X\subseteq \Lambda}^{\bullet} \Phi^{\rm w}(X),\ \ H^{\bs \Psi,{\rm s}}_{\Lambda}=\sum^{\bullet}_{X\subseteq \Lambda} \Phi^{\rm s}(X).
\ee

\begin{Exa}
\upshape
A type of interactions that often come up in specific applications is  discussed here.
For each $X\subseteq \Lambda$, suppose that we are given a self-adjoint operator $A_X\in \M_X$ and a set of self-adjoint operators  $\{B_{X, i} : i=1, \dots, N\} \subset \M_{X}$.
We denote by $J_{\Lambda}$ the modular conjugation associated with $\{\M_{\Lambda}, \vphi_{\Lambda}\}$.
Define
\be
\Phi^{\rm w}(X)=\tilde{A}_X+J_{\Lambda}\tilde{A}_XJ_{\Lambda},\ \ \Phi^{\rm s}(X)=-\sum_{i=1}^N \tilde{B}_{X, i}J_{\Lambda} \tilde{B}_{X, i}J_{\Lambda}.
\ee
Then, because of Proposition \ref{AJAJP}, 
 $-\Phi^{\rm s}(X) \unrhd 0$ w.r.t. $L^2(\M_{\Lambda}, \vphi_{\Lambda})_+$ holds, while $\Phi^{\rm w}(X)$ does not satisfy this condition in general.
Using   Proposition \ref{AJAJP} again, we readily confirm that both $\Phi^{\rm s}(X)$ and $\Phi^{\rm w}(X)$  fulfill Condition \ref{PPBasicA}.  
\end{Exa}

The first type of Hamiltonians to be discussed is characterized by the following condition: 

\begin{Assum}\label{PIBasicAI}\upshape
Fix $\Lambda\in F_{\mathbb{L}}$, arbitrarily.
The set of operators $\{\Phi^{\rm s}(X) : X\subseteq \Lambda\}$ are ergodic w.r.t. $L^2(\M_{\Lambda}, \vphi_{\Lambda})_+$ in the following sense.
Take $M\in \mathrm{spec}(S_{\Lambda}^{(3)} \restriction \h_{\Lambda, N})$, arbitrarily.
For any $\xi, \eta\in L^2(\M_{\Lambda}[M], \vphi_{\Lambda, M})_+\setminus \{0\}$, there exists a sequence, 
$X_1, \dots, X_n$, of subsets of $ \Lambda$ such that \be
(-1)^n\la \xi|\Phi^{\rm s}(X_1) \cdots \Phi^{\rm s}(X_n) \eta\ra>0.\label{CondConn}
\ee

\end{Assum}

In  specific applications, to show \eqref{CondConn}, we need to consider the connectivity of the graph $G$;  see, e.g., Condition \ref{CoupAss}  in Section \ref{Sect5}.

In what follows, 
we denote by 
$H_{\Lambda, M}^{\bs \Psi}$ (resp. $\vphi_{\Lambda, M}$) the restriction of $H^{\bs \Psi}_{\Lambda}$ (resp. $\vphi_{\Lambda}$) to the $M$-subspace (resp. $\M_{\Lambda}[M]$).
\begin{Thm}\label{AbstPI1}
Assume Conditions \ref{PPBasicA} and \ref{PIBasicAI}. For each $M\in \mathrm{spec}(S_{\Lambda}^{(3)} \restriction \h_{\Lambda, N})$, 
we have $e^{-\beta H^{\bs \Psi}_{\Lambda, M}} \rhd 0$ w.r.t. $L^2(\M_{\Lambda}[M], \vphi_{\Lambda, M})_+$ for all $\beta >0$.
\end{Thm}
\begin{proof}
Through this proof, we will illustrate how effective order-preserving operator inequalities can be.
By applying the Duhamel formula, we have
\be
e^{-\beta H^{\bs \Psi}_{\Lambda}}=\sum_{n=0}^{\infty} D_n, \label{DuHa}
\ee
where $D_0=e^{-\beta H_{\Lambda}^{\bs \Psi, {\rm w}}}$ and 
\be
D_n=\int_{0\le s_1\le \cdots \le s_n \le \beta}ds_1\cdots ds_n (-H_{\Lambda}^{\bs \Psi, {\rm s}}(s_1)) \cdots (-H_{\Lambda}^{\bs \Psi, {\rm s}}(s_n)) e^{-\beta H_{\Lambda}^{\bs \Psi, {\rm w}}}
\ee
with $H_{\Lambda}^{\bs \Psi, {\rm s}}(s)=e^{-s H_{\Lambda}^{\bs \Psi, {\rm w}}} H_{\Lambda}^{\bs \Psi, {\rm s}}e^{s H_{\Lambda}^{\bs \Psi, {\rm w}}}$. 
We note that the right-hand side of \eqref{DuHa} converges in the operator norm topology.

Because $-H_{\Lambda}^{\bs \Psi, {\rm s}} \unrhd 0$ w.r.t. $L^2(\M_{\Lambda}, \vphi_{\Lambda})_+$, we have
$ (-H_{\Lambda}^{\bs \Psi, {\rm s}}(s_1)) \cdots (-H_{\Lambda}^{\bs \Psi, {\rm s}}(s_n)) e^{-\beta H_{\Lambda}^{\bs \Psi, {\rm w}}} \unrhd 0$ w.r.t. $L^2(\M_{\Lambda}, \vphi_{\Lambda})_+$,
provided that $0\le s_1\le\cdots \le s_n\le \beta$. Hence, $D_n\unrhd 0$ w.r.t. $L^2(\M_{\Lambda}, \vphi_{\Lambda})_+$ for every $n\in \BbbN_0:=\{0\}\cup \BbbN$, which implies that 
\be
e^{-\beta H^{\bs \Psi}_{\Lambda}} \unrhd D_n\ \ \mbox{w.r.t. $L^2(\M_{\Lambda}, \vphi_{\Lambda})_+$ for all $\beta \ge 0$ and $n\in \BbbN_0$}. \label{ExD}
\ee

Given an  $M\in \mathrm{spec}(S_{\Lambda}^{(3)} \restriction \h_{\Lambda, N})$, 
take $\xi, \eta\in L^2(\M_{\Lambda}[M], \vphi_{\Lambda, M})_+ \setminus \{0\}$, arbitrarily. Because $e^{-\beta H_{\Lambda}^{\bs \Psi, {\rm w}}}$ is an injection, 
$e^{-\beta H_{\Lambda}^{\bs \Psi, {\rm w}}} \eta\neq 0$ and $e^{-\beta H_{\Lambda}^{\bs \Psi, {\rm w}}} \eta \ge 0$ w.r.t. $L^2(\M_{\Lambda}[M], \vphi_{\Lambda, M})_+$ for all $\beta \ge 0$.
By Condition \ref{PIBasicAI}, there exists a sequence,  $X_1, \dots, X_n$, of subsets of $\Lambda$ such that 
\be
(-1)^n\la \xi|\Phi^{\rm s}(X_1) \cdots \Phi^{\rm s}(X_n) e^{-\beta H_{\Lambda}^{\bs \Psi, {\rm w}}}\eta\ra>0. \label{ModPI}
\ee 
By using the fact that $-H_{\Lambda}^{\bs \Psi, {\rm s}} \unrhd \Phi(X)$ w.r.t. $L^2(\M_{\Lambda}[M], \vphi_{\Lambda, M})_+$ for all $X\subseteq {\Lambda}$, we find that 
\be
\la \xi|D_n\eta\ra\ge \int_{0\le s_1\le \cdots \le s_n \le \beta}ds_1\cdots ds(-1)^n\la \xi|\Phi^{\rm s}_{s_1}(X_1) \cdots \Phi^{\rm s}_{s_n}(X_n) e^{-\beta H_{\Lambda}^{\bs \Psi, {\rm w}}}\eta\ra, \label{InmInq}
\ee
where $\Phi^{\rm s}_s(X)=e^{-sH_{\Lambda}^{\bs \Psi, {\rm w}}} \Phi^{\rm s}(X)e^{sH_{\Lambda}^{\bs \Psi, {\rm w}}}$. Define 
\be
F(s_1, \dots, s_n)=(-1)^n\la \xi|\Phi^{\rm s}_{s_1}(X_1) \cdots \Phi^{\rm s}_{s_n}(X_n) e^{-\beta H_{\Lambda}^{\bs \Psi, {\rm w}}}\eta\ra.
\ee
Then $F(s_1, \dots, s_n) \ge 0$ provided that $0\le s_1\le \cdots \le s_n \le \beta$,  and $F(0, \dots, 0)>0$
due to \eqref{ModPI}. Combining this with the fact that $F(s_1, \dots, s_n)$ is continuous in $s_1, \dots, s_n$, we conclude that $\la \xi|D_n\eta\ra>0$ holds. By using  \eqref{ExD}, we finally arrive at 
\be
\la \xi|e^{-\beta H_{\Lambda}^{\bs \Psi}} \eta\ra\ge \la \xi|D_n\eta\ra>0.
\ee
This completes the proof of Theorem \ref{AbstPI1}.
\end{proof}

Theorem \ref{AbstPI1} is used in the proofs of the Marshall--Lieb--Mattis theorem on the Heisenberg model,  the Nagaoka--Thouless theorem on the Hubbard model, and so on; see Appendix \ref{SectB} for detail.

\begin{Assum}\label{PIBasicAII}\upshape
Let $\Lambda\in F_{\mathbb{L}}$.
 Choose $M\in \mathrm{spec}(S_{\Lambda}^{(3)} \restriction \h_{\Lambda, N})$, arbitrarily.
For any $\xi, \eta\in L^2(\M_{\Lambda}[M], \vphi_{\Lambda, M})_+\setminus \{0\}$, there exists a sequence, $X_1, \dots, X_n$,  of subsets of $\Lambda$ and non-negative numbers $\beta, s_1, \dots, s_n$ with $0<s_1\le s_2\le \cdots \le s_n \le \beta$ such that 
\be
(-1)^n\Big\la \xi \Big|\Phi^{\rm s}_{s_1}(X_1)\cdots \Phi^{\rm s}_{s_n}(X_n) e^{-\beta H^{\bs \Psi, \rm w}_{\Lambda}}\eta\Big\ra>0, \label{StriInqAss}
\ee
where $\Phi_t^{\rm s}(X)=e^{-tH^{\bs \Psi, \rm w}_{\Lambda}} \Phi^{\rm s}(X)e^{t H^{\bs \Psi, w}_{\Lambda}}$.

\end{Assum}
In practical applications,  Condition \ref{PIBasicAII} is closely related to the properties of graphs; see, e.g., Condition \ref{AssHubbard} in Section \ref{Sect5}.

For readers' convenience, let us explain in some detail the differences between Condition \ref{PIBasicAI} and  Condition \ref{PIBasicAII}:
Condition \ref{PIBasicAI} is characterized using {\it only} $\Phi^{\rm s}(X)\ (X\subseteq \Lambda)$, while  Condition \ref{PIBasicAII} implies that $\Phi^{\rm s}(X)$ and $\Phi^{\rm w}(X)\ (X\subseteq \Lambda)$ cooperate to form the strict inequality  \eqref{StriInqAss}. In general, it is more difficult to show that Condition \ref{PIBasicAII}  is satisfied than it is to show that Condition \ref{PIBasicAI}  is satisfied.

\begin{Thm}\label{AbstPI2}
Assume Conditions \ref{PPBasicA} and \ref{PIBasicAII}. For each $M\in \mathrm{spec}(S_{\Lambda}^{(3)} \restriction \h_{\Lambda, N})$, 
 $\{e^{-\beta H^{\bs \Psi}_{\Lambda, M}} \}_{\beta \ge 0}$  is ergodic w.r.t. $L^2(\M_{\Lambda}[M], \vphi_{\Lambda, M})_+$.
\end{Thm}
\begin{proof}
The proof of Theorem \ref{AbstPI2} is almost the same as that of Theorem \ref{AbstPI1} up to the middle. Therefore, we will use the same symbols in the proof of Theorem \ref{AbstPI1} in this proof. Recall that the key inequalities in the proof are \eqref{ExD} and \eqref{InmInq}.

By  Condition \ref{PIBasicAII}, we can take $\beta \ge 0$, $n\in \BbbN_0$, $X_1, \dots, X_n$ and $s^*_1, \dots, s^*_n$ such that 
\be 
\la \xi|(-\Phi^{\rm s}_{s^*_1}(X_1)) \cdots (-\Phi^{\rm s}_{s^*_n}(X_n)) e^{-\beta H_{\Lambda}^{\bs \Psi, {\rm w}}}\eta\ra
\ee 
is strictly positive. 
With this setting, we introduce the function  $G(s_1, \dots, s_n)$ by 
\be
G(s_1, \dots, s_n)=\la \xi|(-\Phi^{\rm s}_{s_1}(X_1)) \cdots (-\Phi^{\rm s}_{s_n}(X_n)) e^{-\beta H_{\Lambda}^{\bs \Psi, {\rm w}}}\eta\ra.
\ee
Because $G(s_1, \dots, s_n)$ is continuous in $s_1, \dots, s_n$, we conclude that the right hand side of \eqref{InmInq} is strictly positive. Combining this with \eqref{ExD}, we  arrive at 
\be
\la \xi|e^{-\beta H_{\Lambda}^{\bs \Psi}} \eta\ra\ge \la \xi|D_n\eta\ra>0.
\ee
This completes the proof of Theorem \ref{AbstPI2}.
\end{proof}

Theorem \ref{AbstPI2} is employed in the proof of Lieb's theorem on the Hubbard model and its extensions; see Appendix \ref{SectB} for detail.

\subsection{Perturbation theory}\label{SecPertApp}
Suppose that we are given an interaction $\bs \Psi $ satisfying Conditions \ref{PPBasicA}.
Let $\bs \Upsilon$ be another interaction satisfying Condition \ref{PPBasicA}. In this subsection, we will examine the effect of  perturbation of $\bs \Psi$ by $\bs \Upsilon$. For this purpose, set ${\bs \Psi}+{\bs \Upsilon}:=\{\Psi(X)\dot{+}\Upsilon(X) : X\in F_{\mathbb{L}}\}$, where $\dot{+}$ stands for the form sum.  Under this setup, we readily confirm that  ${\bs \Psi}+{\bs \Upsilon}$ is an interaction. Hence, we can define  the Hamiltonian 
$H^{{\bs \Psi}+{\bs \Upsilon}}_{\Lambda}$ for each $\Lambda\in F_{\mathbb{L}}$.

Let us investigate the properties of the Hamiltonian $H^{{\bs \Psi}+{\bs \Upsilon}}_{\Lambda}$  when the perturbation terms $\bs \Upsilon$ satisfies one of the following two conditions.

\begin{Assum}\label{PertbI}\upshape $\Upsilon(X)$ is bounded and 
$-\tilde{\Upsilon}(X) \unrhd 0$ w.r.t. $L^2(\M_{\Lambda}, \vphi_{\Lambda})_+$ for all $X\subseteq \Lambda$.
\end{Assum}
We note that the assumption of boundedness of $\Upsilon(X)$ can be weakened. However, for this paper, such a generalization is unnecessary, so we will not enter into generalizing the assumption any further.

\begin{Assum}\label{PertbII}\upshape
We denote by $E_X(I)\, (I\in \mathbb{B}^1)$ be the spectral measure of $\Upsilon(X)$, where
 $\mathbb{B}^1$ stands for the Borel sets of $\BbbR$.
If $\xi, \eta\in L^2(\M_{\Lambda}, \vphi_{\Lambda})_+$ satisfies $\la \xi|\eta\ra=0$, 
then $\la \xi|\tilde{E}_X(I) \eta\ra=0$ for all $X\subseteq \Lambda$  and $I\in \mathbb{B}^1$.
\end{Assum}

The following theorem asserts that the ergodicity of the semigroup generated by the Hamiltonian $H^{\bs \Psi}_{\Lambda}$ still holds even when one adds a perturbation term.

\begin{Thm}\label{PertbThm}
Assume  that $\bs \Psi$ and $\bs \Upsilon$ satisfy Condition \ref{PPBasicA}. 
If $\bs \Psi$ satisfies Condition  \ref{PIBasicAI} or  \ref{PIBasicAII} and  $\bs \Upsilon$ satisfies  Condition \ref{PertbI} or \ref{PertbII}, then $\{e^{-\beta H_{\Lambda, M}^{\bs \Psi+\Upsilon}}\}_{\beta \ge 0}$ is ergodic  w.r.t. $L^2(\M_{\Lambda}[M], \vphi_{\Lambda, M})_+$ for each $M\in \mathrm{spec}(S_{\Lambda}^{(3)} \restriction \h_{\Lambda, N})$.

\end{Thm}
\begin{proof}
Here, the proof is given only when $\bs \Psi$ satisfies Condition \ref{PIBasicAI}. However, we can prove the case where $\bs \Psi$ satisfies  Condition \ref{PIBasicAII} as well.

First, assume that $\bs \Upsilon$ satisfies Condition \ref{PertbI}.
 By using Theorem \ref{AbstPI1}, we find that 
$
e^{-\beta H^{\bs \Psi}_{\Lambda, M}} \rhd 0 $ w.r.t. $L^2(\M_{\Lambda}[M], \vphi_{\Lambda, M})_+$ for all $\beta >0$. Hence, by applying Theorem \ref{Mono}, we obtain 
\be
e^{-\beta H^{\bs \Psi+\Upsilon}_{\Lambda, M}}  \unrhd e^{-\beta H^{\bs \Psi}_{\Lambda, M}} \rhd 0 \ \ \mbox{ w.r.t. $L^2(\M_{\Lambda}[M], \vphi_{\Lambda, M})_+$ for all $\beta >0$.}
\ee

Next, assume that $\bs \Upsilon$ satisfies  Condition \ref{PertbII}.
 We apply the method established in \cite{Faris1972,Miyao2021}.
Given an $\ell\in \BbbN$, set $\Upsilon_{\ell}(X)=E_X(I_{\ell}) \Upsilon(X)$, where
$I_{\ell}=[e, e+\ell)$ with $e=\inf \mathrm{spec}(\Upsilon(X))$. Let $\xi, \eta\in   L^2(\M_{\Lambda}[M], \vphi_{\Lambda, M})_+ \setminus\{0\}$ satisfying $\la \xi|\eta\ra=0$. By using Condition \ref{PertbII} and the functional calculus  for self-adjoints operators,
 we have
 \be
 \la \eta|e^{-\beta \tilde{\Upsilon}_{\ell}(X)} \xi\ra=0\ \ (\beta\in \BbbR,\ X\subseteq \Lambda). \label{Vanish}
 \ee
 
 Given a $\xi\in   L^2(\M_{\Lambda}[M], \vphi_{\Lambda, M})_+\setminus \{0\}$, define
 \be
 K(\xi)=\big\{\eta\in   L^2(\M_{\Lambda}[M], \vphi_{\Lambda, M})_+ : \la \eta|e^{-\beta H_{\Lambda, M}^{\bs \Psi+\Upsilon}} \xi\ra=0\ \forall \beta \ge 0\big\}.
 \ee
 Take $\eta\in K(\xi)$. By applying \eqref{Vanish}, we know that $
 \la \eta|e^{+\beta \tilde{\Upsilon}_{\ell}(X)} e^{-\beta H_{\Lambda, M}^{\bs \Psi+\Upsilon}} \xi\ra=0
 $ holds for all $\beta\ge 0$, which implies that $e^{+\beta \tilde{\Upsilon}(X)} K(\xi) \subseteq K(\xi)$.
 Because $ e^{-\beta H_{\Lambda, M}^{\bs \Psi+\Upsilon}} K(\xi)\subseteq K(\xi)$ for all $\beta \ge 0$, we obtain that $
 \big(
  e^{-\beta H_{\Lambda, M}^{\bs \Psi+\Upsilon}/n}  e^{+\beta \tilde{\Upsilon}_{\ell}(X)/n}
 \big)^n K(\xi) \subseteq K(\xi)
 $. Taking the limit $n\to \infty$, we have
 $ e^{-\beta H_{\Lambda, M}^{\bs \Psi+\Upsilon-\Upsilon_{\ell}}} K(\xi)\subseteq K(\xi)$ for all $\beta \ge 0$ and $\ell\in \BbbN$. Note that $\big\{
  H_{\Lambda, M}^{\bs \Psi+\Upsilon-\Upsilon_{\ell}} : \ell\in \BbbN
 \big\}$ is a sequence of non-decreasing self-adjoint operators. Hence, by using \cite[Theorem S. 16]{Reed1981},
 $H_{\Lambda, M}^{\bs \Psi+\Upsilon-\Upsilon_{\ell}}$ converges to $H_{\Lambda, M}^{\bs \Psi}$
  in the strong resolvent sense as $\ell\to \infty$, which implies that $e^{-\beta H_{\Lambda, M}^{\bs \Psi} }
  K(\xi)\subseteq K(\xi)$. Hence, for any $\eta\in K(\xi)$, it holds that $\la \eta| e^{-\beta H_{\Lambda, M}^{\bs \Psi} } \xi\ra=0$ for all $\beta \ge 0$. Because $e^{-\beta H_{\Lambda, M}^{\bs \Psi} }\rhd 0$ w.r.t. $  L^2(\M_{\Lambda}[M], \vphi_{\Lambda, M})_+$ for all $\beta >0$, $\eta$ must be zero. I.e., $K(\xi)=\{0\}$, which means that $\{e^{-\beta H_{\Lambda, M}^{\bs \Psi+\Upsilon}}\}_{\beta \ge 0}$ is ergodic  w.r.t. $  L^2(\M_{\Lambda}[M], \vphi_{\Lambda, M})_+$.
\end{proof}

To end this subsection, we provide a remark for determining when Condition \ref{PertbII} is satisfied:

\begin{Rem} \upshape
By applying arguments similar to those in \cite{Miyao2021}, we can prove that Condition \ref{PertbII} is satisfied if the following are fulfilled:
\begin{itemize}
\item[\rm (i)] $(\tilde{\Upsilon}(X)+\im)^{-1}\in Z(\M_{\Lambda})$, where $Z(\M_{\Lambda})$ is the center of $\M_{\Lambda}$.
\item[\rm (ii)] $\Delta_{\Lambda}^{\im t} \tilde{\Upsilon}(X) \subseteq \tilde{\Upsilon}(X) \Delta_{\Lambda}^{\im t}$ for all $t\in \BbbR$, where $\Delta_{\Lambda}$ is the modular operator associated with $\{\M_{\Lambda}, \vphi_{\Lambda}\}$.
\end{itemize}
Although this characterization looks complicated at first glance, it can be quickly confirmed in the actual applications in Appendix \ref{SectB} since $\M_{\Lambda}$ is commutative and $\Delta_{\Lambda}=1$.

\end{Rem}

\section{Ergodic properties of semigroups generated by various Hamiltonians}\label{SectB}
\subsection{Results}

In this appendix, we explain the ergodic  properties of the semigroups generated by the Hamiltonians discussed in this paper.
The actual proof of the ergodic property of a semigroup generated by a specific Hamiltonian is often technically complicated.
However, for all the Hamiltonians discussed in this paper, the general theory given in Appendix \ref{SectA} can be applied. Therefore, we will not give the details of the proofs here but will only clarify which of the results in Appendix \ref{SectA} applies to each Hamiltonian.
For readers who want to know the details of the proof, we give references as appropriate.
In the outline of the proofs given in Subsection \ref{PfErgAp}, care is taken to recognize the correspondence between the descriptions in the references and the descriptions in this paper.

In this appendix, we continue to use the symbols from Sections \ref{Sect5} and \ref{Sect6}.
The following theorem summarizes the properties of the semigroups generated by the Hamiltonians appearing in Sections \ref{Sect5} and \ref{Sect6}:
\begin{Thm}\label{ManyHa}
We have the following:
\begin{itemize}
\item[\rm (i)]  $e^{-\beta H^{\rm MLM}_{\Lambda, M}} \rhd 0$ w.r.t. $L^2(\M_{\Lambda}^{\rm MLM}[M], \vphi^{\rm MLM}_{\Lambda, M})_+$ for each $M\in \mathrm{spec}(S_{\Lambda}^{(3)} \restriction \h_{\Lambda})$ and $\beta>0$.
\item[\rm (ii)]  $e^{-\beta H^{\rm Hei}_{\Lambda, M}} \rhd 0$ w.r.t. $L^2(\M_{\Lambda}^{\rm MLM}[M], \vphi^{\rm MLM}_{\Lambda, M})_+$ for each $M\in \mathrm{spec}(S_{\Lambda}^{(3)} \restriction \h_{\Lambda})$ and $\beta>0$.
\item[\rm (iii)]  $\{e^{-\beta H^{\rm H}_{\Lambda, M}} \}_{\beta\ge 0}$ is ergodic w.r.t. $L^2(\M_{\Lambda}^{\rm H}[M], \vphi^{\rm H}_{\Lambda, M})_+$ for each $M\in \mathrm{spec}(S_{\Lambda}^{(3)} \restriction \h_{\Lambda}^{\rm H})$.
\item[\rm (iv)]  $\{e^{-\beta H^{\rm HH}_{\Lambda, M}} \}_{\beta \ge 0}$ is ergodic w.r.t. $L^2(\M_{\Lambda}^{\rm HH}[M], \vphi^{\rm HH}_{\Lambda, M})_+$ for each $M\in \mathrm{spec}(S_{\Lambda}^{(3)} \restriction \h^{\rm H}_{\Lambda})$.
\item[\rm (v)]  $\{e^{-\beta H^{\rm K}_{\Lambda, \sharp, M}} \}_{\beta\ge 0}$ is ergodic w.r.t. $L^2(\M_{\Lambda}^{\rm K}[M], \vphi^{\rm K}_{\Lambda, \sharp,  M})_+$ for each $M\in \mathrm{spec}(S_{\Lambda}^{(3)} \restriction \h^{\rm K}_{\Lambda})$  and $\sharp={\rm AF}, {\rm F}$.
\item[\rm (vi)]  $\{e^{-\beta H^{\rm KH}_{\Lambda, \sharp, M}} \}_{\beta \ge 0}$ is ergodic w.r.t. $L^2(\M_{\Lambda}^{\rm KH}[M], \vphi^{\rm KH}_{\Lambda, \sharp,  M})_+$ for each $M\in \mathrm{spec}(S_{\Lambda}^{(3)} \restriction \h^{\rm K}_{\Lambda})$ and $\sharp={\rm AF}, {\rm F}$.
\item[\rm (vii)] 
  $e^{-\beta H^{\rm H, \infty}_{\Lambda,  M}} \rhd 0$ w.r.t. $L^2(\M_{\Lambda}^{\rm NT}[M], \vphi^{\rm NT}_{\Lambda, M})_+$ for each $M\in \mathrm{spec}(S_{\Lambda}^{(3)} \restriction \h^{\rm NT}_{\Lambda})$ and  $\beta>0$.
 \item[\rm (viii)] $e^{-\beta H^{\rm HH, \infty}_{\Lambda,  M}} \rhd 0$ w.r.t. $L^2(\N_{\Lambda}^{\rm HH}[M], \mu^{\rm HH}_{\Lambda,   M})_+$ for each $M\in \mathrm{spec}(S_{\Lambda}^{(3)} \restriction \h^{\rm NT}_{\Lambda})$ and  $\beta>0$.
\end{itemize}
\end{Thm}

\subsection{Proof of Theorem \ref{ManyHa}} \label{PfErgAp}
\subsubsection*{Outline of   proofs of (i) and (ii)}
Given $X, Y\subseteq \Lambda$, we set
\be
|X, Y)_{\Lambda}=
\Bigg[
\prod'_{x\in X} c_{x\up}^*
\Bigg]
\Bigg[
\prod'_{y\in Y} c_{y\down}^*
\Bigg]|\varnothing\ra_{\Lambda}.
\ee
We define the {\it hole-particle transformation} $W: \h_{\Lambda, |\Lambda|} \to \h_{\Lambda, |\Lambda|}$ by
\be
W|X, \overline{Y}\ra_{\Lambda}=|X, Y)_{\Lambda}\ \  (X, Y\subseteq \Lambda,\ |X|=|Y|),
\ee
where $|X, Y\ra_{\Lambda}$ is defined by \eqref{DefCONS}.
We readily confirm that 
\be
Wc_{x\up}W^{-1}=c_{x\up},\ \ \ W c_{x\down} W^{-1}=\gamma_x c_{x\down}^*,\label{PHT}
\ee
where $\gamma_x=+1$ if $x\in \Lambda_A$, $\gamma_x=-1$ if $x\in \Lambda_B$.
In what follows, the restriction of $W$ to $\h_{\Lambda}$ is denoted by the same symbol.
Given a linear operator $A$ on $\h_{\Lambda, |\Lambda|}$, we also use the notaton $\hat{A}=WAW^{-1}$.
Using \eqref{PHT}, we obtain $\hat{Q}_{\Lambda}=\prod_{x\in \Lambda}(n_x-1)^2$. Hence, by \eqref{HSingE},
we find that $\hat{\h}_{\Lambda}=W\h_{\Lambda}=\hat{Q}\h_{\Lambda, |\Lambda|}$.
Furthermore, if we set $\Cone^{\rm MLM}_{\Lambda}=WL^2(\M^{\rm MLM}_{\Lambda}, \vphi^{\rm MLM}_{\Lambda})_+ $, we see that 
$
\Cone^{\rm MLM}_{\Lambda}=\mathrm{coni}\{|X, X)_{\Lambda} : X\subseteq \Lambda\},
$
where, for a given set $S$, $\mathrm{coni}(S)$ stands for the conical hull of $S$.
The hole-particle transformed Heisenberg Hamiltonian, $\hat{H}^{\rm Hei}_{\Lambda}$ can be written as follows:
\be
\hat{H}^{\rm Hei}_{\Lambda}=\sum_{X\subset \Lambda; |X|=2} \Phi(X),
\ee
where $\Phi(X)$ is given by $\Phi(X)=\Phi^{\rm s}(X)+\Phi^{\rm w}(X)$ with
\be
\Phi^{\rm s}(\{x, y\})=J_{xy}(\hat{S}_x^{(1)} \hat{S}_y^{(1)}+\hat{S}_x^{(2)} \hat{S}_y^{(2)}),\ \ \Phi^{\rm w}(\{x, y\})=J_{xy}\hat{S}_x^{(3)} \hat{S}_y^{(3)}.
\ee
Under the settings described above, using Theorem \ref{AbstPI1}, we obtain (ii) of Theorem \ref{ManyHa}.
For a more detailed proof, see \cite[Theorem B.5]{Miyao2019}.
Since the MLM Hamiltonian is a particular case of the Heisenberg Hamiltonian, it means that (i) of Theorem \ref{ManyHa} has been proved simultaneously.

\subsubsection*{Outline of   proof of (iii)}
 When we put $\hat{\h}^{\rm H}_{\Lambda}=W\h_{\Lambda}^{\rm H}$, using the properties of the hole-particle transformation, we get that $\hat{\h}^{\rm H}_{\Lambda}=\h^{\rm H}_{\Lambda}$.
Similarly, for each $M$-subspace $\h^{\rm H}_{\Lambda}[M]\, (M\in \mathrm{spec}(S_{\Lambda}^{(3)} \restriction \h^{\rm H}_{\Lambda}))$, we set $\hat{\h}^{\rm H}_{\Lambda}[M]=W\h^{\rm H}_{\Lambda}[M]$.
 To make the symbol simpler, set $\Cone^{\rm H}_{\Lambda}[M]=WL^2(\M^{\rm H}_{\Lambda}[M], \vphi^{\rm H}_{\Lambda, M})_+$.
 Due to the fact $W|X, Y\ra_{\Lambda}=|X, \overline{Y})_{\Lambda}$,
any vector $\psi$ in $\hat{\h}^{\rm H}_{\Lambda}[M]$ can be expressed as $\psi=\sum_{X, Y\in\mathcal{Q}_{\Lambda}[M]}\psi_{XY}|X, Y)_{\Lambda}$, where $\mathcal{Q}_{\Lambda}[M]=\{X\subseteq \Lambda : |X|=M^{\dagger}\}$ with $M^{\dagger}=M+|\Lambda|/2$.
A necessary and sufficient condition for a vector $\psi\in \hat{\h}^{\rm H}_{\Lambda}[M]$ to belong to $\Cone^{\rm H}_{\Lambda}[M]$ is that the matrix $\{\psi_{XY} : X, Y\in \mathcal{Q}_{\Lambda}[M]\}$ be a positive semi-definite matrix.
We will now clarify an alternative representation of this fact. This other representation is more tractable in proving the claims that we  desire.
Using the identification between fermionic Fock spaces used in Lemma \ref{VolSys}, we see that \footnote{
To show \eqref{HilIdn}, we use the following fact: if $N_{\Lambda}$ is the number operator of electrons (i.e., 
$N_{\Lambda}=\sum_{x\in \Lambda} n_x$), 
$WS^{(3)}_{\Lambda}W^{-1}=(N_{\Lambda}-|\Lambda|)/2$ and $WN_{\Lambda}W^{-1}=2S^{(3)}_{\Lambda}+|\Lambda|$.
}
\be
\hat{\h}^{\rm H}_{\Lambda}[M]=\bigoplus_{M} \F_{\Lambda, M^{\dagger}} \otimes \F_{\Lambda, M^{\dagger}}.\label{HilIdn}
\ee
 Note that it is essential here that $|\Lambda|$ is even.
 In general, for a given Hilbert space $\X$, we have the identification $\X\otimes \X=\mathscr{L}^2(\X)$, where
 $\mathscr{L}^2(\X)$ is the set of Hilbert--Schmidt operators on $\X$.\footnote{
 More precisely, there exists a unitary operator, $\tau$,  from $\X\otimes \X$ to $\mathscr{L}^2(\X)$ such that 
 $\tau x\otimes y=|x\ra\la \vartheta y|$, where $\vartheta$ is some antiunitary operator on $\X$.
 }
 Combining this with \eqref{HilIdn}, we have
 \be
 \hat{\h}^{\rm H}_{\Lambda}[M]=\bigoplus_{M} \mathscr{L}^2(\F_{\Lambda, M^{\dagger}}).
 \ee
 Furthermore, under this identification, we know that 
 \be
 \Cone^{\rm H}_{\Lambda}[M]=\bigoplus_M \mathscr{L}^2(\F_{\Lambda, M^{\dagger}})_+, 
 \ee
 where $\mathscr{L}^2(\F_{\Lambda, M^{\dagger}})_+$ stands for  the set of all positive operators in $\mathscr{L}^2(\F_{\Lambda, M^{\dagger}})$.
 The hole-particle transformed Hubbard model can be written as follows:
 \be
 \hat{H}^{\rm H}_{\Lambda}=\sum_{|X|=2; X\subseteq \Lambda}\Phi(X), \label{TildeH}
 \ee
 where the interaction $\Phi(X)$ is given by $\Phi(X)=\Phi^{\rm s}(X)+\Phi^{\rm w}(X)$ with
 \be
 \Phi^{\rm s}(\{x, y\})=- U_{xy} n_{x\up} n_{y\down},\ \ \Phi^{\rm w}(\{x, y\})=
 \sum_{\sigma=\up, \down} \Big(t_{xy}c_{x\sigma}^*c_{y\sigma}+\frac{1}{2} U_{xy}n_{x\sigma}n_{y\sigma}\Big).
 \ee
 Note that, as is well known, the hole-particle transformation converts the repulsive model to the attractive model.
 Under this expression, by applying Theorem \ref{AbstPI2}, we obtain (iii) of Theorem \ref{ManyHa}. See \cite{Miyao2012} for the details of the proof.
 
\subsubsection*{Outline of   proof of (iv)}
If  the hole-particle transformation is denoted as $W$, we set $\hat{\h}^{\rm HH}_{\Lambda}=W\h^{\rm HH}_{\Lambda},  \hat{\h}^{\rm HH}_{\Lambda}[M]=W\h^{\rm HH}_{\Lambda}[M]$. Also, set $\Cone^{\rm HH}_{\Lambda}[M]=WL^2(\M^{\rm HH}_{\Lambda}[M], \vphi^{\rm HH}_{\Lambda, M})_+$.
Using the identification \eqref{IdnLF}, the Hilbert space $\hat{\h}^{\rm HH}_{\Lambda}[M]$ can be expressed as $\hat{\h}^{\rm HH}_{\Lambda}[M]=\int^{\oplus}_{\BbbR^{|\Lambda|}} \hat{\h}^{\rm H}_{\Lambda}[M]d{\bs q}$.
 Combining this with \eqref{vNTensor}, we get the fiber direct integral decomposition:
$\Cone^{\rm HH}_{\Lambda}[M] =\int^{\oplus}_{\BbbR^{|\Lambda|}} \Cone^{\rm H}_{\Lambda}[M] d{\bs q}$.\footnote{For the definition of the fiber-direct integral of self-dual cones, see \cite[Appendix I]{Miyao2019}.}

If we denote $\hat{H}^{\rm HH}_{\Lambda}$ as the hole-particle transformation of $H^{\rm HH}_{\Lambda}$,  $\hat{H}^{\rm HH}_{\Lambda}$ can be expressed as follows: 
\be
\hat{H}^{\rm HH}_{\Lambda}=\hat{H}^{\rm H}_{\Lambda}+\sum_{x, y\in \Lambda} g_{xy}(n_{x\up}-n_{y\down})
(b_y^*+b_y)+\sum_{x\in \Lambda} \omega b_x^*b_x,
\ee
where $\hat{H}^{\rm H}_{\Lambda}$ is given by \eqref{TildeH}.
The Lang--Firsov transformation $G$ is the unitary operator on $\hat{\h}^{\rm HH}_{\Lambda}$  defined by $G=e^{-i \frac{\pi}{2}N_{\rm ph}} e^L$,  where 
\be
L=-\im \omega^{-3/2} \sum_{x, y\in \Lambda} g_{xy}(n_{x\up}-n_{y\down})p_y,\ \ p_y=\im \sqrt{\frac{\omega}{2}}
\overline{(b_x^*-b_x)}.
\ee
Here, for a given closable operator $A$, $\overline{A}$ stands for the closure of $A$. 
If we denote $G \hat{H}^{\rm HH}_{\Lambda} G^{-1}$ by $K^{\rm HH}_{\Lambda}$, then $K^{\rm HH}_{\Lambda}$  can be expressed as follows: 
\be
K^{\rm HH}_{\Lambda}=\sum_{X\subseteq \Lambda} \Phi(X),
\ee
where $\Phi(X)$ is given by $\Phi(X)=\Phi^{\rm s}(X)+\Phi^{\rm w}(X)$ with
\begin{align}
\Phi^{\rm s}(X)&=\begin{cases}
-U_{{\rm eff},  xy} n_{x\up}n_{y\down} & \mbox{if $X=\{x ,y\}\subseteq \Lambda$}\\
0 & \mbox{otherwise},
\end{cases}\\
\Phi^{\rm w}(X) &=\begin{cases}
\omega b_x^*b_x & \mbox{if $X=\{x\}\subseteq \Lambda$}\\
\sum_{\sigma=\up, \down} (t_{xy} e^{\im \theta_{xy, \sigma}} c_{x\sigma}^*c_{y\sigma} +\frac{U_{{\rm eff}, xy}}{2} n_{x\sigma}n_{y\sigma}) & \mbox{if $X=\{x, y\}\subseteq \Lambda$}\\
0 & \mbox{otherwise}.
\end{cases}\label{HHPhi}
\end{align}
If we set $\kappa_{xy}=\sqrt{2} \omega^{-1/2}\sum_{z\in \Lambda} (g_{xz}-g_{yz})q_z$ with $q_x=\overline{(b_x^*+b_x)}/\sqrt{2\omega}$, 
the phase factor $\theta_{xy, \sigma}$ appearing in \eqref{HHPhi} is defined by 
$\theta_{xy, \sigma}=\kappa_{xy}$ if $\sigma=\up$, $\theta_{xy, \down}=-\kappa_{xy}$ if $\sigma=\down$.
In the above setting, we can apply Theorem \ref{AbstPI2}, which results in (iv) of Theorem \ref{ManyHa}. See \cite{Miyao2017} for the details of the proof.

\subsubsection*{Outline of   proof of (v)}
The Hilbert space $\h^{\rm K}_{\Lambda}$ of the Kondo lattice system is defined by replacing $\Lambda$ with $\Lambda\sqcup \Lambda$ in the definition of $\h^{\rm H}_{\Lambda}$.
Recall that the bipartite structure of the underlying graph is given by \eqref{AFBI} or \eqref{FBI}, depending on whether the sign of the interaction constant $J$ is positive or negative. 
 Therefore, there are also two kinds of hole-particle transformations, $W_{\rm AF}$ and $W_{\rm F}$, corresponding to each. 
 The annihilation operators are transformed as follows:
 \be
 W_{\sharp}c_{x\up}W^{-1}_{\sharp}=c_{x\up}, \ W_{\sharp}f_{x\up}W^{-1}_{\sharp}=f_{x\up},
 W_{\sharp}c_{x\down}W^{-1}_{\sharp}=\gamma_x c_{x\down}^*,\ W_{\sharp}f_{x\down}W^{-1}_{\sharp}= \alpha_{\sharp}\gamma_x f_{x\down}^*,
 \ee
 where $\alpha_{\sharp}=+1$ if $\sharp=\rm AF$, $\alpha_{\sharp}=-1$ if $\sharp=\rm F$. Comparing these equations with \eqref{PHT}, we see that the difference between the hole-particle transformations in the Kondo system and the ordinary  system becomes more apparent.
 We now define $\hat{\h}^{\rm K}_{\Lambda, \sharp}=W_{\sharp} \h^{\rm K}_{\Lambda}$, $\hat{\h}^{\rm K}_{\Lambda, \sharp}[M]=W_{\sharp}\h^{\rm K}_{\Lambda}[M]$
 and $\Cone^{\rm K}_{\Lambda, \sharp}[M]=W_{\sharp}L^2(\M_{\Lambda}^{\rm K}[M], \vphi^{\rm K}_{\Lambda, \sharp})_+\ (\sharp={\rm AF},  {\rm F})$.
 The hole-particle transformed Hamiltonian $\hat{H}^{\rm K}_{\Lambda, \sharp}=W_{\sharp} H^{\rm K}_{\Lambda, \sharp}W^{-1}_{\sharp}$ is expressed as follows:
 \be
 \hat{H}^{\rm K}_{\Lambda, \sharp}=\sum_{X\subseteq \Lambda} \Phi(X),
 \ee
 where $\Phi(X)=\Phi^{\rm s}(X)+\Phi^{\rm w}(X)$ with 
 \begin{align}
 \Phi^{\rm s}(X)&=\begin{cases}
 -\frac{|J|}{2} (c_{x\up}^* f_{x\up} c_{x\down}^*f_{x\down}+f_{x\up}^*c_{x\up}f_{x\down}^*c_{x\down}) & \mbox{if $|X|=\{x\}\subseteq \Lambda$}\\
 -2 U_{xy} n_{x\up}^c n_{y\down}^{c} & \mbox{if $X=\{x, y\}\subseteq \Lambda$}\\
 0 & \mbox{otherwise},
 \end{cases}\\
 \Phi^{\rm w}(X) &=\begin{cases}
 \frac{J}{4}(n_x^c-1)(n_{x}^f-1) & \mbox{if $X=\{x\}\subseteq \Lambda$}\\
 \sum_{\sigma=\up, \down}(t_{xy} c_{x\sigma}^*c_{y\sigma}+U_{xy}n_{x\sigma}^cn_{y\sigma}^c) & \mbox{if $X=\{x, y \}\subseteq \Lambda$}\\
 0 & \mbox{otherwise}.
 \end{cases}
 \end{align}
 By applying Theorem \ref{AbstPI2}, we obtain (v) of Theorem \ref{ManyHa}. See Ref. \cite{Miyao2021-2} for the details of the proof.\footnote{
 More precisely, in Ref. \cite[Theorem 3.23]{Miyao2021-2}, we analyze the Hamiltonian $H^{\rm KH}_{\Lambda}$ taking into account the interaction between electrons and lattice vibrations. However, the result obtained in  \cite{Miyao2021-2} is also valid in the absence of electron-phonon interaction, and thus we obtain (v) of Theorem \ref{ManyHa}.
 }

\subsubsection*{Outline of   proof of (vi)}
As before, let $\hat{\h}^{\rm KH}_{\Lambda, \sharp}=W_{\sharp} \h^{\rm KH}_{\Lambda}$, $\hat{\h}^{\rm KH}_{\Lambda, \sharp}[M]=W_{\sharp} \h^{\rm KH}_{\Lambda}[M]$ and $\Cone^{\rm KH}_{\Lambda, \sharp}=W_{\sharp} L^2(\M^{\rm KH}_{\Lambda}, \vphi^{\rm KH}_{\Lambda, \sharp})_+$. Then, using the identification \eqref{IdnLF}, we obtain 
$\hat{\h}^{\rm KH}_{\Lambda, \sharp}[M]=\int^{\oplus}_{\BbbR^{|\Lambda|}} \hat{\h}^{\rm K}_{\Lambda, \sharp}[M] d{\bs q}$ and  $\Cone^{\rm KH}_{\Lambda, \sharp}[M]=\int^{\oplus}_{\BbbR^{|\Lambda|}} \Cone^{\rm K}_{\Lambda, \sharp}[M] d{\bs q}$.
As in the case of the Holstein--Hubbard model, we first apply the hole-particle transformation and then the Lang--Firsov transformation to the Hamiltonian $H^{\rm KH}_{\Lambda, \sharp}$ to obtain the Hamiltonian $K^{\rm KH}_{\Lambda, \sharp}$ defined by
\be
K^{\rm KH}_{\Lambda, \sharp}=\sum_{X\subseteq \Lambda} \Phi(X),
\ee
where $\Phi(X)=\Phi^{\rm s}(X)+\Phi^{\rm w}(X)$ with 
 \begin{align}
 \Phi^{\rm s}(X)&=\begin{cases}
 -\frac{|J|}{2} (c_{x\up}^* f_{x\up} c_{x\down}^*f_{x\down}+f_{x\up}^*c_{x\up}f_{x\down}^*c_{x\down}) & \mbox{if $|X|=\{x\}\subseteq \Lambda$}\\
 -2 U_{\mathrm{eff}, xy} n_{x\up}^c n_{y\down}^{c} & \mbox{if $X=\{x, y\}\subseteq \Lambda$}\\
 0 & \mbox{otherwise},
 \end{cases}\\
 \Phi^{\rm w}(X) &=\begin{cases}
\omega b_x^*b_x+ \frac{J}{4}(n_x^c-1)(n_{x}^f-1) & \mbox{if $X=\{x\}\subseteq \Lambda$}\\
 \sum_{\sigma=\up, \down}(t_{xy}e^{\im \theta_{xy, \sigma}} c_{x\sigma}^*c_{y\sigma}+U_{\mathrm{eff}, xy}n_{x\sigma}^cn_{y\sigma}^c) & \mbox{if $X=\{x, y\}\subseteq \Lambda$}\\
 0 & \mbox{otherwise}.
 \end{cases}
 \end{align}
By applying Theorem \ref{AbstPI2}, we obtain (vi) of Theorem \ref{ManyHa}. See Ref. \cite[Theorem 3.23]{Miyao2021-2} for the details of the proof.\footnote{
 In \cite{Miyao2021-2}, only the case of $M=0$ is considered, but the proof method used there can be applied to the general case of $M$.}

\subsubsection*{Outline of   proof of (vii)}
 The Hamiltonian $H_{\Lambda}^{\rm H, \infty}$ can be expressed as 
\be
H_{\Lambda}^{\rm H, \infty}=\sum_{|X|=2; X\subseteq \Lambda} \Phi(X),
\ee
where $\Phi(X)$ is given by $\Phi(X)=\Phi^{\rm s}(X)+\Phi^{\rm w}(X)$ with
\be
\Phi^{\rm s}(\{x, y\})=\sum_{\sigma=\up, \down} t_{xy} P^{\rm G}_{\Lambda} c_{x\sigma}^*c_{y\sigma}P^{\rm G}_{\Lambda},\ \ 
\Phi^{\rm w}(\{x, y\})=\frac{U_{xy}}{2} (n_x-1)(n_y-1)P^{\rm G}_{\Lambda}.
\ee
Because $-P^{\rm G}_{\Lambda}c_{x\sigma}^*c_{y\sigma } P^{\rm G}_{\Lambda} \unrhd 0$ w.r.t. $L^2(\M^{\rm NT}_{\Lambda}[M], \vphi^{\rm NT}_{\Lambda, M})_+$, we have $-\Phi^{\rm s}(\{x, y\}) \unrhd 0$ w.r.t. $L^2(\M^{\rm NT}_{\Lambda}[M], \vphi^{\rm NT}_{\Lambda, M})_+$.
Recall that $\M^{\rm NT}_{\Lambda}$ is the von Neumann  algebra of diagonal operators associated with 
$\{|{\bs \sigma}\ra_{\Lambda} : {\bs \sigma} \in \mathcal{S}_{\Lambda}\}$.
Because $\Phi^{\rm w}(\{x, y\})\in  \M^{\rm NT}_{\Lambda}$, we readily confirm that $e^{-\beta\Phi^{\rm w}(\{x, y\}) } \unrhd 0$ w.r.t. $L^2(\M^{\rm NT}_{\Lambda}[M], \vphi^{\rm NT}_{\Lambda, M})_+$.
 From these facts and Trotter's formula, we know that $\Phi(X)$ satisfies Condition \ref{PPBasicA}.
Furthermore, we can show that $\{\Phi(X) : X\subseteq \Lambda\}$ satisfies Condition \ref{PIBasicAI}. Therefore, we can conclude (vii) of Theorem \ref{ManyHa} by applying Theorem \ref{AbstPI1}.

\subsubsection*{Outline of proof of (viii)}
Consider the following two interactions:
\begin{align}
\Phi(X)&=
\begin{cases}
\omega b_x^*b_x P^{\rm G}_{\Lambda} & \mbox{if  $X=\{x\}\subseteq \Lambda$}\\
\sum_{\sigma=\up, \down}t_{xy} P^{\rm G}_{\Lambda} c_{x\sigma}^*c_{y\sigma}  P^{\rm G}_{\Lambda} & \mbox{if $X=\{x, y\}\subseteq \Lambda$}\\
0 & \mbox{otherwise},
\end{cases}\\
\tilde{\Upsilon}(X)&=\begin{cases}
\Big[g_{xy}(n_x-1)(b_y+b_y^*)+\frac{U_{xy}}{2}(n_x-1)(n_y-1) \Big]P^{\rm G }_{\Lambda}& \mbox{if $X=\{x, y\}\subseteq \Lambda$}\\
0 & \mbox{otherwise}.
\end{cases}
\end{align}
Note that using the notation of Subsection \ref{SecPertApp}, we can express $H_{\Lambda}^{\rm HH, \infty}$  as $
H_{\Lambda}^{\rm HH, \infty}=H_{\Lambda}^{{\bs \Psi}+{\bs \Upsilon}}
$. 
In other words, $\Upsilon(X)$ can be regarded as a perturbation term. First of all, let us investigate the properties of the perturbed term $\Phi(X)$.
 To do so, we split $\Phi(X)$ as 
$\Phi(X)=\Phi^{\rm s}(X)+\Phi^{\rm w}(X)$, where
\be
\Phi^{\rm s}(\{x, y\})=\sum_{\sigma=\up, \down}t_{xy}P^{\rm G}_{\Lambda} c_{x\sigma}^*c_{y\sigma}P^{\rm G}_{\Lambda}, \ \
\Phi^{\rm w}(\{x\})=
\omega b_x^*b_x P^{\rm G}_{\Lambda}.
\ee
Because $b_x^*b_x$ can be identified with the self-adjoint operator $\frac{1}{2}(-\Delta_{q_x}+q_x^2-1)$ acting in $L^2(\BbbR^{|\Lambda|})$, we see that $e^{-\beta \omega b_x^*b_x}\rhd 0$ w.r.t. $L^2(\BbbR^{|\Lambda|})_+$ for all $\beta >0$.
Using this property, we can prove that $\{\Phi(X) : X\subseteq \Lambda\}$ satisfies Condition \ref{PIBasicAII}.
Since $\Upsilon(\{x, y\})$ is a diagonal operator, it is easy to check that Condition \ref{PertbII} is satisfied.
Therefore, using Theorem \ref{PertbThm},
we find that $\{e^{-\beta H^{\rm HH, \infty}_{\Lambda} }\}_{\beta \ge 0}$ is ergodic w.r.t. $L^2(\N_{\Lambda}^{\rm HH}[M], \mu^{\rm HH}_{\Lambda, \sharp,  M})_+$.
 Furthermore, since $\N_{\Lambda}^{\rm HH}[M]$ is abelian, we can use Theorem \cite[Theorem A.4]{Miyao2021} to obtain
 (viii) of Theorem \ref{ManyHa}.
 See \cite{Miyao2017} for details of the proof.\footnote{
 The outline of the proof given here is somewhat different from the one in \cite{Miyao2017}. However, the essence of the two proofs is the same, and the details of the above outline can be easily reproduced using some of the lemmas given in \cite{Miyao2017}.
 }

\section{Some technical results concerning order preserving operator inequalities}\label{SectC}
In this appendix, we collect some technical facts about order preserving operators needed in this paper for the convenience of the reader.
Suppose we are given a von Neumann algebra $\M$ and a faithful semi-finite normal weight  $\vphi$ on $\M$.
Set $\Cone=L^2(\M, \vphi)_+$.

We begin with the following proposition.
\begin{Prop}
Let $\{A_n\}_{n\in \BbbN}$ be a sequence of bounded operators on $L^2(\M, \vphi)$ and let $A$ be a bounded operator on $L^2(\M, \vphi)$.
Suppose that $A_n$ strongly converges to $A$ as $n\to \infty$.
If $A_n\unrhd 0$ w.r.t. $\Cone$ for all $n\in \BbbN$, then $A\unrhd 0$ w.r.t. $\Cone$.
\end{Prop}
\begin{proof}
See \cite[Proposition 2.8]{Miyao2019-2}.
\end{proof}

The following proposition is a basic tool.

\begin{Prop}\label{BasicPertPP}
Let $A$ be a positive self-adjoint operator on $L^2(\M, \vphi)$  and let $B$ be a  self-adjoint bounded
 operator on $L^2(\M, \vphi)$. Assume the following:
\begin{itemize}
\item[{\rm (i)}] $0\unlhd e^{-tA}$ w.r.t. $\Cone$ for all $t\ge 0$.
\item[{\rm (ii)}]$0\unlhd -B$ w.r.t. $\Cone$.
\end{itemize} 
Then $0\unlhd e^{-t(A+B)}$ w.r.t. $\Cone$ for all $t\ge 0$.
\end{Prop}
\begin{proof} See \cite[Theorem A. 18]{Miyao2019-2}.\footnote{
In \cite{Miyao2019-2}, the claim is proved in the case where $B$ is  unbounded and symmetric.
}
\end{proof}

Theorem \ref{Mono} below plays an  important role in practical applications.
\begin{Thm}[Monotonicity]\label{Mono}
Let $A, B$ be self-adjoint positive 
 operators on $L^2(\M, \vphi)$. 
Assume that there exists a bounded self-adjoint operator $C$ satisfying  $B=A-C$.
In addition, suppose 
 that 
\begin{itemize}
\item[{\rm (i)}] $e^{-\beta A}\unrhd 0$ w.r.t. $\Cone$ for
	     all $\beta\ge 0$;
\item[{\rm (ii)}] $C\unrhd 0$ w.r.t. $\Cone$.
\end{itemize} 
Then we have $e^{-\beta B }\unrhd e^{-\beta A}$
 w.r.t. $\Cone$ for all $\beta\ge  0$. 
\end{Thm} 
\begin{proof} See \cite[Theorem A. 4]{Miyao2020}.  \end{proof}


\end{document}